\documentclass[aps,prx,longbibliography,superscriptaddress,twocolumn,showpacs]{revtex4-2}
\usepackage[utf8]{inputenc}
\usepackage[american,british]{babel}
\usepackage[T1]{fontenc}
\usepackage[normalem]{ulem}
\usepackage{float}%for figure positions
\usepackage{graphicx,amsmath,amsthm,amssymb,mathrsfs,mathtools,amsfonts,braket,psfrag,epsfig,url,eepic,xcolor,rotating,epstopdf,leftindex}
\usepackage{enumitem}
\newtheorem{theorem}{Theorem}[section]

\newtheorem{lemma}[theorem]{Lemma}

\usepackage{simplewick}

\usepackage{bm}% bold math
\usepackage[normalem]{ulem}

\usepackage[makeroom]{cancel}
\usepackage[caption=false]{subfig}

\usepackage[colorlinks=true, linktocpage=true]{hyperref}
\hypersetup{
    allcolors = {black}
}

\usepackage[percent]{overpic} % Allows coordinates to be specified as percentages

\usepackage{mathdots}
\usepackage{dsfont}

\DeclareMathAlphabet{\mathscrbf}{OMS}{mdugm}{b}{n}

 \newcommand{\RN}[1]{%Roman numbers
  \textup{\expandafter{(\romannumeral#1)}}%
}
\usepackage{comment}

\usepackage{bigints}
\usepackage{float}
\usepackage{scalerel,stackengine}

\begin{document}
\title{
Efficient Pseudomode Representation and Complexity of Quantum Impurity Models}

\author{Julian Thoenniss}
\thanks{These authors contributed equally to this work.}
\affiliation{Department of Theoretical Physics,
University of Geneva, 30 Quai Ernest-Ansermet,
1205 Geneva, Switzerland}

\author{Ilya Vilkoviskiy}
\thanks{These authors contributed equally to this work.}
\affiliation{Department of Theoretical Physics,
University of Geneva, 30 Quai Ernest-Ansermet,
1205 Geneva, Switzerland}
\affiliation{Department of Physics, Princeton University, Princeton, New Jersey 08544, USA}

\author{Dmitry A. Abanin}
\affiliation{Department of Physics, Princeton University, Princeton, New Jersey 08544, USA}
\affiliation{Google Research, Brandschenkestrasse 150, 8002 Z\"urich, Switzerland}

\date{\today}

\begin{abstract}

Out-of-equilibrium fermionic quantum impurity models (QIM), describing a small interacting system coupled to a continuous fermionic bath, play an important role in condensed matter physics. Solving such models is a computationally demanding task, and a variety of computational approaches are based on finding approximate representations of the bath by a finite number of modes. In this paper, we formulate the problem of finding efficient bath representations as that of approximating a kernel of the bath's Feynman-Vernon influence functional by a sum of complex exponentials, with each term defining a fermionic pseudomode. Under mild assumptions on the analytic properties of the bath spectral density, we provide an analytic construction of pseudomodes, and prove that their number scales polylogarithmically with the maximum evolution time $T$ and the approximation error $\varepsilon$. We then demonstrate that the number of pseudomodes can be significantly reduced by an interpolative matrix decomposition (ID).
Furthermore, we present a complementary approach, based on constructing rational approximations of the bath's spectral density using the ``AAA'' algorithm, followed by compression with ID. The combination of two approaches yields a pseudomode count scaling as $N_\text{ID} \sim \log(T)\log(1/\varepsilon)$, and the agreement between the two approches suggests that the result is close to optimal. Finally, to relate our findings to QIM, we derive an explicit Liouvillian that describes the time evolution of the combined impurity-pseudomodes system. These results establish bounds on the computational resources required for solving out-of-equilibrium QIMs, providing an efficient starting point for tensor-network methods for QIMs. 

%Quantum impurity models out of equilibrium are computationally challenging, with exact solutions requiring exponentially scaling resources with bath size. To mitigate this, numerical approximations have been developed to simplify bath descriptions, but prior studies often relied on empirical error estimates without rigorous scaling analysis.
%In this paper, we address this gap by mapping the full problem to the one of approximating a Fourier integral by a finite sum of decaying exponentials, with each term defining a fermionic pseudomode. We prove that for certain bath spectral densities, the number of modes needed for accuracy $\varepsilon$ up to evolution time $T$ scales as $N_\text{bath} \sim \log(T/\varepsilon)\log(1/\varepsilon)$. By applying interpolative matrix decomposition (ID), we further compress these modes, achieving even milder error scaling.
%Additionally, we present a complementary approach that constructs rational approximations of the bath's spectral density using the ``AAA'' algorithm, again followed by compression with ID. With the combination of methods, our strategy is applicable to arbitrary spectral densities and achieves a mode count scaling as $N_\text{ID} \sim \log(T)\log(1/\varepsilon)$, with strong agreement across methods, indicating near-optimality of the approximations obtained.
%Finally, we relate our findings back to the original impurity problem and derive an explicit Liouvillian that describes the time evolution of the combined impurity-pseudomode system.
\end{abstract}
\maketitle

\section{Introduction}

Quantum impurity models (QIM), describing a small interacting system such as an electronic orbital coupled to continuous baths of non-interacting fermions, play a special role in quantum many-body physics. Despite their simplicity, they host a wealth of phenomena, including Fermi-edge singularities, Kondo effect, and non-Fermi liquid behavior~\cite{BullaRMP2008}. With the advances in experimental techniques for both solid-state and cold-atomic systems~\cite{Knap2012,Grimm2016}, non-equilibrium properties of fermionic QIM have been attracting much interest, in particular, in the context of quantum transport phenomena~\cite{Pustilnik_2004}. 

Fermionic QIMs are also a centerpiece of several computational techniques for the predictive modeling of strongly correlated materials. In these approaches, dynamical mean-field theory (DMFT)~\cite{GeorgesRMP,aoki14neqdmftreview} and density matrix embedding theory~\cite{DMEmbeddingChan2018} being prime examples, a correlated material is modeled by a quantum impurity coupled to one or more non-interacting baths, with spectral properties computed self-consistently. 

Due to the ubiquity and broad applications of QIMs, tremendous efforts have been dedicated to developing efficient numerical impurity solvers, both in and out of equilibrium. Quantum Monte  Carlo (QMC) methods for real time evolution~\cite{Muhlbacher08Realtime,Schiro09realtime,Werner09Diagrammatic,gull11NumericallyExact,MonteCarloReview,cohen11memory,cohen13neqkondo} are generally limited by the fermionic sign problem and statistical errors, although the former can be partially addressed with inchworm diagrammatic QMC~\cite{cohen15taming}.

A number of tensor-networks methods for non-equilibrium QIM  have been developed. The starting point is often to approximate a continuous bath by a discrete set of fermionic modes, followed by tensor-network compression of the time-evolved state of impurity and the discrete bath. Approaches with modes described by the Hamiltonian~\cite{prior10efficient,WolfPRB14,Nusseler20Efficient} and dissipative evolution~\cite{dorda14auxiliary,Dorda2017,Lotem20renormalized, Schwarz18Nonequilibrium} have been put forward, and different schemes for bath discretization were investigated. Further, a hierarchical equation of motion approach, combined with MPS representation, has been applied to the Anderson impurity model~\cite{DanPRB2023HEOM}. 

Recent works developed a conceptually different approach~\cite{ThoennissPRB2022, ThoennissPRB23,NgPRB23,ParkPRB24} to out-of-equilibrium fermionic QIM, based on matrix-product states (MPS) representations of the Feynman-Vernon influence functional (IF)~\cite{FeynmanVernon}---an object that arises when the bath degrees of freedom are integrated out and that captures non-Markovian properties of the bath. It was shown that the IF of a fermionic bath displays an area-law scaling of temporal entanglement~\cite{lerose2021,ThoennissPRB2022} at any non-zero temperature. This suggests that the QIM can be simulated with computational resources that scale polynomially with both the evolution time and the target error. This intuition was confirmed in numerical implementations of such IF-MPS methods applied to the single-orbital Anderson QIM~\cite{ThoennissPRB23,NgPRB23}. 

Despite recent methodological advances, rigorous results on the computational complexity of the non-equilibrium QIM remain scarce~\footnote{We note that a quasipolynomial-time classical algorithm for finding the ground state energy of QIM is available~\cite{Bravyi2016}}. Additionally, compressing an IF into MPS form with a high bond dimension, though aided by its Gaussian structure, is generally relatively resource-intensive. This complexity poses challenges for applying the method to DMFT, where solving the QIM requires handling multiple bath parameters within a self-consistency loop.

\begin{figure}
    \centering
    \begin{overpic}
        [width=\linewidth]{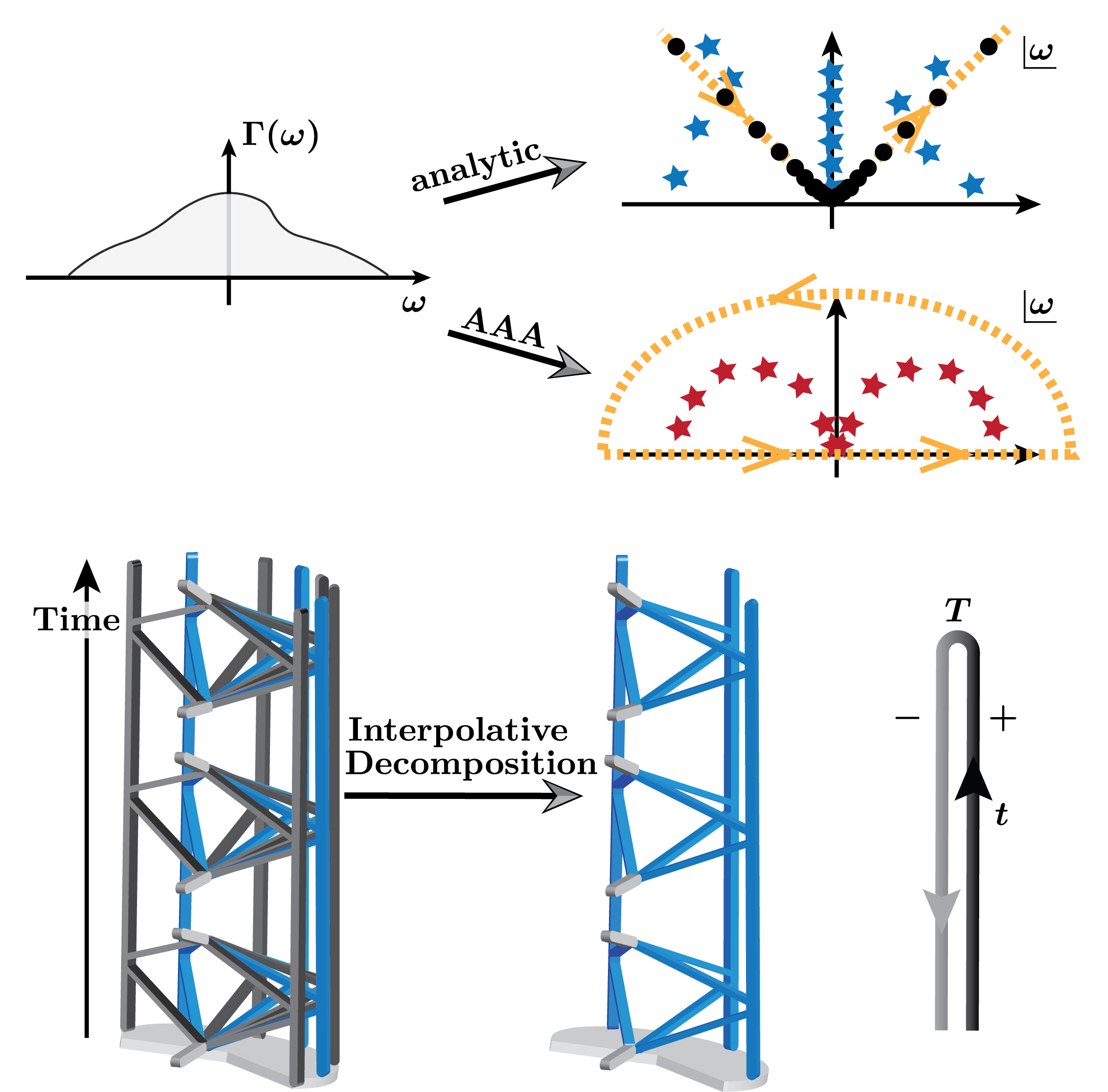}
        \put(0,97){\text{a) Bath spectral density}}
         \put(50,97){\text{b) Determine pseudomodes}}
          \put(0,50){\text{c) Compress pseudomodes}}
          \put(67,50){\text{d) Keldysh Contour}}
    \end{overpic}
    \caption{Schematic illustration of the approach used in this work. Starting from a continuous spectral density (a), we determine a set of complex modes through either analytic construction (b, top) or the AAA algorithm (b, bottom). The yellow dotted line indicates the integration contour used for the Fourier integrals that define the hybridization funtion. In the analytic construction, pseudomodes are obtained via exponential frequency parametrization along a rotated contour in the complex plane (black dots). Blue stars represent poles from the spectral density and the Fermi-Dirac distribution, which must be considered when rotating the contour. The AAA algorithm identifies pseudomodes as sets of poles and residues in the complex plane (red stars).
c) In order to compress the set of modes, we consider the hybridization function which encodes the sum of temporal correlations induced by each mode. By employing interpolative matrix decomposition, we extract a subset of pseudomodes with renormalized couplings which approximates the hybridization function with a controlled error. d) Illustration of the Keldysh contour.}
    \label{fig:schematic_figure}
\end{figure}
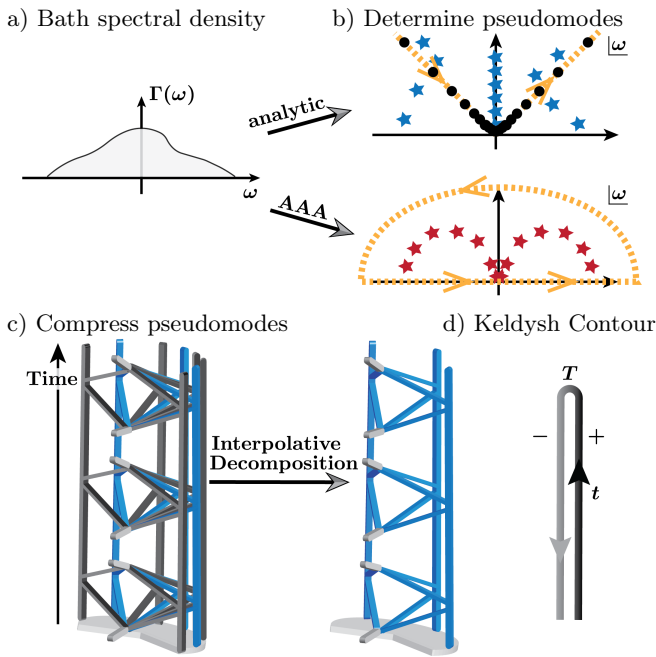
In this paper, we investigate the computational complexity of non-equilibrium QIM. Our approach combines approximating the continuous bath by a finite set of ``pseudomodes'' (fermionic levels coupled to the impurity, with possibly complex frequencies and couplings)~\cite{dorda14auxiliary}, with the use of influence functionals to bound errors on observables introduced by such an approximation, recently used in the context of spin-boson model~\cite{VilkoviskiyPRB2024}. In particular, we prove that, if the IF is well-approximated, then the time evolution of an impurity coupled to such a set of pseudomodes closely approximates the evolution in the original QIM problem. It is worth noting that pseudomode representations have been previously fruitfully employed in the context of bosonic baths~\cite{TamascelliPRL2018,Somoza2019,Xu2022}. 

Conceptually, our method rests on reducing the full problem to the description of kernel functions that take the form of ordinary Fourier integrals. The challenge of constructing an efficient set of pseudomodes is then reframed as the broader problem of approximating a Fourier integral with a finite sum of decaying exponentials:
\begin{equation}\label{eq:Fourier_integral_intro}
    \int \frac{d\omega}{2\pi} f(\omega) e^{i\omega t} \approx \sum_k \Gamma_k e^{i\omega_k t}, 
\end{equation} with $\omega_k,\Gamma_k \in \mathbb{C}.$
To achieve this, we employ a combination of complex analysis techniques and efficient numerical sampling algorithms that incorporate both spectral information and temporal correlations. This approach enables the construction of compact, discrete representations of general Fourier transforms.

Relating our findings back to original quantum impurity problems, we obtain a number of results regarding the approximation of different baths with pseudomodes. 
First, we prove a theorem which bounds the computational resources required to describe time evolution of a QIM up to a given time $T$ with a desired error $\varepsilon$ that is defined as time integral of the deviation between a kernel function and its approximation. This result, proven under certain assumptions regarding the analytic structure of the spectral density of the bath, states that the bath can be approximated by a number of fermionic pseudomodes that grows polylogarithmically both in $T$ and $\varepsilon^{-1}$: 
\begin{equation}\label{eq:Nbath}
N_{\rm bath}  \sim \log(T/\varepsilon)\log(1/\varepsilon). 
\end{equation}
Interestingly, such compact approximation of the bath requires pseudomodes with complex frequencies, corresponding to dissipative evolution, and with complex, unphysical couplings to the impurity level~\footnote{Nevertheless, the {\it combination} of pseudomodes closely approximates the original, physical evolution, and, moreover, it can be complemented to fully physical evolution if a dissipative decay channel is added to the impurity evolution.}.

Second, we show that the resulting pseudomode sets can be further compressed by employing interpolative matrix decomposition (ID), which identifies and renormalizes a relevant subset of pseudomodes based on their contribution to the bath hybridization function. ID has recently been fruitfully applied in order to construct pseudomodes for equilibrium problems, and has been shown to simplify the evaluation of vertex functions~\cite{KayeDiscrete,kiese24Discrete}.

Lastly, we develop a complementary numerical approach based on finding a rational approximation to the product of the bath's spectral density and thermal distribution function by means of the adaptive Anatoulas-Anderson (AAA) algorithm~\cite{NakatsukasaAAA}. The AAA algorithm---which was recently successfully used in a closely related context of HEOM approach to QIM~\cite{DanPRB2023HEOM}---yields a set of poles and residues in the complex plane that define a set of pseudomodes. Here we also find that the obtained pseudomode sets can be strongly compressed with ID to yield a more compact representation of the bath.

The combination of the above approaches allows us to construct an efficient pseudomode approximation for arbitrary spectral densities, including those where the rigorous proof does not apply or spectral information is known only as numerical data points. In this work, focused on non-equilibrium setups, we study a number of physically relevant kinds of baths (flat band, semicircle spectral density, gapped, linear spectral density) and show that zero frequency singularities, including cusps, can be overcome. Across all cases, we find the compressed number of pseudomodes to consistently scale as \begin{equation}\label{eq:Nbath_improved}
    N_\text{ID}\sim \log(T)\log(1/\varepsilon),
\end{equation} 
with quantitative agreement between both approaches.

The representations of fermionic baths, along with the accuracy guarantees we obtain, indicate the polynomial computational complexity of non-equilibrium QIM. As we will see below, both analytic and AAA approaches combined with ID compression lead to representations with a quantitatively similar number of modes $N_\text{ID}$. This suggest that the representations we obtain are nearly optimal. 

Lastly, we connect our findings to the Lindblad framework and derive an explicit Liouvillian $\mathcal{L}$ that describes the time evolution of the combined system of impurity and pseudomodes in the form,
\begin{equation}\label{eq:Lindblad_intro}
    \frac{d}{dt}\ket{\rho} = \mathcal{L}\ket{\rho}.
\end{equation}    
 It is conceivable that using MPS for time-evolution will allow to further reduce required computational resources, although we leave an exploration of this issue for future work.

The paper is organized as follows: In Sec.~\ref{sec:hyb_func}, we review the Keldysh formalism for QIM and reduce the problem of describing the bath to that of approximating a Fourier integral as sum of decaying exponentials. Sec.~\ref{sec:analytic_bound} derives the scaling law in Eq.~(\ref{eq:Nbath}). In Sec.~\ref{sec:compression}, we introduce interpolative decomposition (ID) to compress the obtained approximation. Sec.~\ref{sec:AAA} presents the AAA algorithm for generating pseudomodes, suited to more general spectral densities. In Sec.~\ref{sec:complexity_general_baths}, we compare different spectral densities, treated with both approaches, and present indication for near-optimality of the approximation. Finally, Sec.~\ref{sec:dynamics} connects the pseudomode construction to the Lindblad formalism.

\section{Pseudomode Approximation to Bath Dynamics}\label{sec:hyb_func}

In this section, we review the main equations of the Keldysh path integral approach to non-equilibrium quantum impurity problems, which form the foundation of our work. A key feature of this formalism is the natural emergence of the hybridization function $\Delta(\tau,\tau^\prime)$ which fully encodes the dynamic influence of the bath on the impurity and is the central object of this work.

We consider the quench dynamics of a localized quantum impurity: At time $t=0,$ the impurity is coupled to a fermionic bath, after which the joint system evolves according to the Hamiltonian,
\begin{align}\label{eq:Ham_full}
 H &= H_\text{imp} + H_\text{hyb} + H_\text{bath},\\
H_\text{hyb} &= \sum_k t_k\, d^\dagger c_{k} + h.c., \label{eq:H_hyb}\\
H_\text{bath} &= \sum_k \epsilon_k c_{k}^\dagger c_{k}.\label{eq:H_bath}
\end{align}
Here $d$ ($d^\dagger$) and $c_{k}$ ($c_{k}^\dagger$) are annihilation (creation) operators on the impurity and the $k^\text{th}$ bath mode, respectively. They obey the standard fermionic anti-commutation relations:
\begin{gather}
\{d^{\dagger},d\}=1 \,,\\
\{c^{\dagger}_p,c_k\}=\delta_{p,k} \,.
\end{gather}
The local interactions on the impurity are encoded in $H_\text{imp}$ which we keep unspecified for the sake of generality.

The problem of describing local quantities on the impurity, such as physical observables and temporal correlation functions, can be phrased in terms of a $0+1$ dimensional path integral, where the dynamics is fully defined by $H_\text{imp}$ and by the hybridization function $\Delta(\tau,\tau^\prime):$\\
\begin{multline}
\langle \hat{O}_{\text{imp}}(t)\rangle
\;\propto\;
\int \bigg(\prod_{\tau} d\bar{\eta}_{\tau}d\eta_{\tau}\bigg)\mathcal{O}_{\text{imp}}(\bar{\eta}_{t},\eta_{t}) \\
\times \exp\bigg(\int_\mathcal{C} d\tau  \Big[\bar{\eta}_{\tau}\partial_{\tau} \eta_{\tau}
-i \mathcal{H}_\text{imp}(\bar{\eta}_{\tau},\eta_{\tau})\Big] \bigg)\mathcal{\rho}_\text{imp}[\bar{\bm{\eta}}_{0},\bm{\eta}_{0}]\\
\times \exp\bigg(\int_\mathcal{C}d\tau \int_\mathcal{C} d\tau^\prime \bar{\eta}_{\tau} \Delta(\tau,\tau^\prime) \eta_{\tau^\prime} \bigg).
\label{eq:expect_value}
\end{multline}
Here, $\tau\in \mathcal{C}$ parametrizes the Keldysh time contour up to a final time $T,$ defined as $\mathcal{C}=(0^+\to T^+ \to T^-\to 0^-)$, see Fig.~\ref{fig:schematic_figure}d. Moreover, $\mathcal{\rho}_\text{imp}[\bar{\bm{\eta}}_{0},\bm{\eta}_{0}]$ is the density matrix kernel of the initial impurity state and the operator $\hat{O}_\text{imp}$ and its corresponding kernel $\mathcal{O}_\text{imp}$ describes the observable at time $t\in\mathcal{C}$.

The hybridization function \(\Delta(\tau, \tau^\prime)\) can be decomposed into particle and hole components,
\begin{equation}\label{eq:hybridization_function}
    \Delta(\tau, \tau^\prime) = \Delta^p(\tau, \tau^\prime) + \Delta^h(\tau, \tau^\prime).
\end{equation}
While each component has a specific structure on the Keldysh contour, they can each be expressed in terms of a kernel function with a single positive time argument \( t \geq 0 \):
\begin{align}
    \Delta^p(t) &= \int \frac{d\omega}{2\pi} \Gamma(\omega) g^p(t; \omega), \label{eq:particle_integral}\\
    \Delta^h(t) &= \int \frac{d\omega}{2\pi} \Gamma(\omega) g^h(t; \omega),\label{eq:hole_integral}
\end{align} for the particle and hole components $\Delta^p$ and $\Delta^h$, respectively. For the exact relations, see Eqs.~(\ref{eq:Keldysh_particle_relation1}--\ref{eq:Keldysh_hole_relation4}).
Here, we introduced
\begin{align}
    g^p(t; \omega) &= (1-n_\text{F}(\omega)) \, e^{i\omega t}, \label{eq:gp_definition}\\
    g^h(t; \omega) &= n_\text{F}(\omega) \, e^{i\omega t}, \label{eq:gh_definition}\\
    \Gamma(\omega) &= \sum\limits_k |t_k|^2 \delta(\omega - \epsilon_k).
\end{align}

This provides all the necessary prerequisites to summarize the conceptual foundation of this work:

\begin{enumerate}[label=(\roman*)]
\item The kernel functions $\Delta^{p}(t)$ and $\Delta^{h}(t),$ viewed as ordinary Fourier integrals for $t \in \mathbb{R}^+$, can be approximated as a sum of decaying exponentials,
\begin{equation}
    \label{eq:exponential_approx}
    \Delta(t) \approx \tilde{\Delta}(t) = \sum_{k} \Gamma_k e^{i\omega_k t}
\end{equation} with $\Gamma_k,\omega_k \in \mathbb{C}$ and $\Im(\omega_k)\geq 0.$
Constructing efficient and controlled approximations of this type is the central focus of Secs.~\ref{sec:analytic_bound}--\ref{sec:complexity_general_baths}.
\item Each term in Eq.~(\ref{eq:exponential_approx}), characterized by a tuple $(\Gamma_k, \omega_k),$ defines a dissipative fermionic pseudomode. Here, $\Gamma_k$ parametrizes the coupling to the impurity, $\Omega_k=\Re(\omega_k)$ is the oscillation frequency, and $\gamma_k=\Im(\omega_k)$ is the decay rate. In Sec.~\ref{sec:dynamics}, we leverage this identification to connect our results back to the original problem and derive an explicit Liouvillian that describes the dynamics of the combined impurity-pseudomode system.
\end{enumerate}

This naturally raises the question of the minimal number of pseudomodes \( N_\text{bath} \) required to approximate the kernel functions to a given accuracy, i.e.,
\begin{equation}\label{eq:exponential_decomposition}
\Delta(t) = \sum\limits_{k=1}^{N_\text{bath}} \Gamma_k e^{i\omega_k t} + \delta(t), \text{ with} 
\int\limits_{0}^{T} |\delta(t)| \, dt < \varepsilon,
\end{equation}
for any $t\leq T$, where $T$ is maximum simulation time. Note that with this definition of the error, $\varepsilon$ has the dimension of energy. In Sec.~\ref{sec:Bound_on_the_observables}, we will show that the actual error on observables is given by $T\varepsilon$.
The remainder of this article is dedicated to investigating various approaches to this question, encompassing both rigorous error bounds and a numerical compression algorithm.

We conclude this section by noting a subtle issue: While for $\Gamma_k \in \mathbb{R},$ a parameter tuple $(\Gamma_k, \omega_k)$ corresponds to a physical fermionic mode that is described by standard Lindblad evolution, complex $\Gamma_k \notin \mathbb{R}$ violate physicality conditions and a straight-forward Lindbladian description is not possible in this case. In Sec.~\ref{sec:dynamics} and App.~\ref{app:Lindblad_construction}, we further explore this aspect and demonstrate how the case $\Gamma_k \notin \mathbb{R}$ can be reconciled with the Lindblad framework.

\section{Analytic Error Bound}
\label{sec:analytic_bound}

In this Section, we present our first main result: an analytic bound on the number of exponential terms needed to approximate the kernel function, Eq.~(\ref{eq:exponential_approx}), with a desired accuracy. First, in Subsection~\ref{sec:bound}, we prove the corresponding theorem. Second, in Subsection~\ref{sec:Bound_on_the_observables}, we show that the impurity dynamics can be accurately reproduced with the approximate bath composed of pseudomodes, and derive an error bound for observables.

\subsection{Proof of Bound for Error Scaling}\label{sec:bound}

Our goal is to analytically derive an exponential sum representation of the kernel function in Eq.~(\ref{eq:exponential_approx}). We achieve this with the help of complex analysis techniques, relying on assumptions regarding the analytic properties of the spectral density $\Gamma(\omega)$. Conceptually, our approach in this Section generalizes the method presented in Ref.~\cite{Beylkin2010ApproximationBE}, which derives a similar exponential representation for the power function $t^{-\alpha}$, and which we used to bound the complexity of an Ohmic bosonic bath~\cite{VilkoviskiyPRB2024}.

We prove the following theorem:
\begin{theorem} \label{th:finite_approximation} For a fixed angle $0<r_{\text{max}}<\frac{\pi}{4}$, consider a spectral density $\Gamma(\omega)$ which is meromorphic in the upper half-plane. Further, assume that this function decays at least exponentially as $|\omega|\to\infty$,
\begin{equation}\label{eq:th_condition}
|\Gamma(\omega)|< \Gamma e^{-\nu|\Re(\omega)|}\,, \quad \{|\omega| \gg \Lambda\,, \ 0<\arg(\omega)<2r_{\text{max}}\} \ ,
\end{equation}
where $\Re(\omega)$ is the real part of $\omega$, and has a finite number of poles $\omega=\Omega_k$ in the sector $0<\arg(\omega)<2r_{\text{max}}$. 

Then, for any inverse temperature $\beta$, the positive (negative) part of the particle (hole) kernel function 
\begin{equation} \label{eq:cont_integral}
\Delta^p_+(t)=\int\limits_{0}^{\infty}\Gamma(\omega)(1-n_{\text{F}}(\omega))e^{i\omega t} \frac{d\omega}{2\pi}
\end{equation}
can be approximated by a finite number of exponentials:
\begin{equation}\label{eq:theorem_discrete_sum}
\Delta^p_+(t)=\sum\limits_{k=1}^{N_{\rm{bath}}} \Gamma_k e^{i\omega_k t}+\delta(t)\,,
\end{equation}
with a bounded error
\begin{equation}
\int\limits_{t=0}^{T}|\delta(t)|dt<\varepsilon,
\end{equation}
and the total number of terms scales as:
\begin{equation}\label{eq:N_bath_scaling}
N_{\rm{bath}}\sim\frac{1}{2\pi r}\log(\varepsilon^{-1})\log(T\varepsilon^{-1})\,, \quad \text{for any} \ 0<r<r_{\text{max}}
\end{equation}
\begin{proof}
In this Section, we sketch the main steps of the proof while referring the reader to Appendix \ref{sec:proof} for more details. 

Starting from Eq.~(\ref{eq:exponential_approx}), we split the integration domain into two intervals, corresponding to positive and negative frequencies. In what follows, we focus on the positive frequency interval of the particle component---which we denote as $\Delta_+^p$---noting that the treatment of negative frequencies contribution and the hole component is analogous. 

\begin{figure}[h]
\includegraphics[scale=0.6]{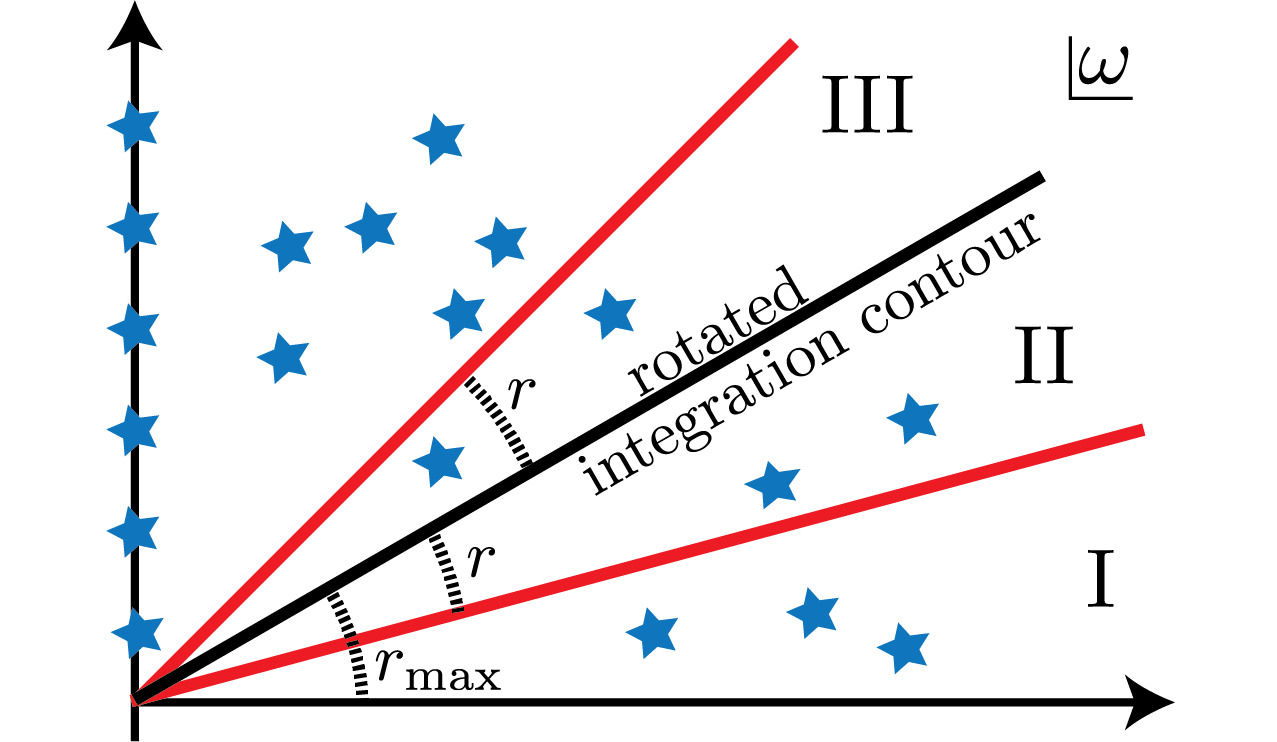}
    \caption{Integration contour and different groups of poles.}
    \label{fig:rotated_contour}
\end{figure}

Our proof consists of three steps. First, we rotate the integration contour in Eq.~\eqref{eq:cont_integral} into the complex plane. Second, we discretize the rotated integral and estimate the difference between the continuous and discretized integral. The third step is to truncate the obtained infinite summation, obtaining the desired finite sum approximation. As we will see below, the exponentials in Eq.~$\eqref{eq:theorem_discrete_sum}$ can be grouped as follows:
\begin{equation}\label{eq:three_sums_theorem}
\sum\limits_{k=1}^{N_{\text{bath}}}\Gamma_k e^{i\omega_k t}=D_1(t)+D_2(t)+\tilde{\Delta}^p_{+,\text{rotated}}(t),
\end{equation}
where each of the terms $D_1(t),D_2(t),$ and $\tilde{\Delta}^p_{+,\text{rotated}}(t),$ has a distinct interpretation.

Let us describe each of the steps in more detail. First, we rotate the integral into the complex plane, \begin{equation}
\omega\to\omega e^{ir_{\text{max}}},
\label{eq:first_variable_transform}
\end{equation}
and define:
\begin{equation}\label{eq:rotated_integral}
\Delta^p_{+,\text{rotated}}(t)=e^{ir_{\text{max}}}\int\limits_{0}^{\infty}f(t,\omega e^{ir_{\text{max}}})\frac{d\omega}{2\pi},
\end{equation}
\begin{equation}\label{eq:kernel_f}
f(t,\omega)=\Gamma(\omega)(1-n_{\text{F}}(\omega))e^{i\omega t}.
\end{equation} This step is schematically illustrated in Fig.~\ref{fig:rotated_contour}.
The rotated integral differs from the original one by the contribution of $\kappa_2$ residues, situated in regions I and II:
\begin{equation}\label{eq:three_groups_of_exponentials}
\Delta^p_+(t)-\Delta^p_{+,\text{rotated}}(t)=D_1(t)=\sum\limits_{k=1}^{\kappa_2} R_k(1-n_{\text{F}}(\Omega_k))e^{i\Omega_k t}, 
\end{equation}
where, $R_k=i\ \text{res}_{\omega=\Omega_k}\Gamma(\omega)$. This expression, which is a sum of exponentials, gives the first term in Eq.~\eqref{eq:three_groups_of_exponentials}. We note that, by the conditions of the theorem, the number of terms in $D_1(t)$ is finite.

Next, following Ref.~\cite{Beylkin2010ApproximationBE}, we discretize the complex integral. The main difficulty in approximating the kernel function $\Delta_+^p(t)$ is its possible slow polynomial decay at large times, common for many physically relevant spectral densities. In order to efficiently approximate the long-time behavior, we choose a non-uniform frequency discretization, concentrating more points around zero frequency. This is achieved by another change of variables,
\begin{equation}\label{eq:second_variable_transform}
    \omega=\Gamma e^{x},\end{equation} 
    with a uniform discretization in the variable $x$. To estimate the difference between the continuous integral of any function $\int\limits_{-\infty}^{\infty}g(x)dx$ and the discrete sum approximation $h\sum\limits_{k=-\infty}^{\infty} g(hk),$ we represent the latter as a contour integral:
\begin{equation}
h\sum\limits_{k-\infty}^{\infty}g(h k)=-\frac{i}{2}\oint g(x)\cot\frac{\pi x}{h}dx,
\end{equation}
and apply a version of Theorem 5.2 from Ref.~\cite{McNamee}. In Appendix~\ref{sec:proof}, we provide a careful analysis of the integral in Eq.~\eqref{eq:rotated_integral}, yielding the following estimate for the difference between the discretized and continuous integrals:
\begin{equation}\label{eq:Delta-Delta_discr}
\Delta^p_{+,\text{rotated}}(t)-\tilde{\Delta}^p_{+,\text{rotated}}(t)=\delta_r(t)+D_2(t).
\end{equation}
Here, $\delta_r(t)\sim e^{-\frac{2\pi r}{h}}$ is the small discrepancy, $\tilde{\Delta}^p_{+,\text{rotated}}(t)$ is the sum arising when the integral is discretized (see Eq.~\eqref{eq:rotated_integral_approx} below for an explicit expression). Further, we defined another pole contribution $D_2(t)$:
\begin{multline}\label{eq:D2}
    D_2(t)=\sum\limits_{k=\kappa_1+1}^{\kappa_2}R_k\frac{(1-n_{\text{F}}(\Omega_k))e^{-2i\pi x_k/h}}{1-e^{-2i\pi x_k/h}}e^{i\Omega_k t}-\\
-\sum\limits_{k=\kappa_2+1}^{\kappa_3}R_k\frac{\big(1-n_{\text{F}}(\Omega_k) \big)e^{2i\pi x_k/h}}{1-e^{2i\pi x_k/h}}e^{i\Omega_k t}, 
\end{multline}
involving the poles $k=\kappa_1+1,...\kappa_2$ situated in region II, as well as poles $k=\kappa_2+1,...\kappa_3$ situated in region III. Above, we introduced $x_k=\log\frac{\Omega_k}{\Gamma}$. 

The final step is to truncate the infinite sum,
\begin{multline}\label{eq:rotated_integral_approx}
\tilde{\Delta}^p_{+,\text{rotated}}(t)=\\= \frac{h}{2\pi}\sum\limits_{k=-\infty}^{\infty}\Gamma(\omega_k)\big(1-n_{\text{F}}(\omega_k)\big)e^{i\omega_k t}\omega_k,
\end{multline} with $\omega_k=\Gamma e^{hk+ir_{\text{max}}}.$
%Eqs.(\ref{eq:three_groups_of_exponentials}),(\ref{eq:Delta-Delta_discr}),(\ref{eq:D2}) yield the approximation of the kernel function in terms of three sums of exponentials (\ref{eq:three_sums_theorem}).
First, we note that the norm of the discrepancy $\|\delta(t)\|_{L_1}$ in Eq.~\eqref{eq:Delta-Delta_discr} is exponentially small in $\frac{1}{h}$
\begin{equation}
\|\delta(t)\|_{L_1}\sim e^{-\frac{2\pi r}{h}}.
\end{equation} 
This implies that a choice $h\sim \frac{2\pi r}{\log(\varepsilon^{-1})}$ of the summation step ensures an approximation error below $\varepsilon$. Further, the theorem conditions guarantee the exponential decay of the summand in Eq.(\ref{eq:rotated_integral_approx}) both in the limit of high frequencies $k\to \infty$ and in the limit of low frequencies $k\to -\infty$, which allows us to truncate the infinite sum. In the limit of large $k$, we use the theorem condition from Eq.~\eqref{eq:th_condition} to bound the upper summation limit $k_{\text{max}}=N-1$ as:
\begin{equation}
\Gamma e^{-\nu \Gamma \cos(r_{\text{max}})e^{hN}}\sim \varepsilon\,,\ \text{or} \ N\sim\frac{1}{h}\log\log\frac{\Gamma }{\varepsilon}.
\end{equation}
In the limit of low frequencies, one may neglect the frequency dependence of the spectral density and the Fermi distribution. The suppression comes from the $\omega_k$ prefactor in Eq.~\eqref{eq:rotated_integral_approx} and provides the lower bound cutoff $k_{\text{min}}=-M+1$:
\begin{equation}
T \Gamma(0)\Gamma e^{-h M}\cos(r_{\text{max}})\sim \varepsilon\,,\ \text{or} \ M\sim \frac{1}{h}\log\frac{T}{\varepsilon}.
\end{equation}
The factor $T$, the total evolution time, appears here because we are bounding the $L_1$ norm of the discrepancy, and the norm of exponents with small frequency is proportional to $T$. Altogether, we obtain the following scaling for the total number of exponentials:
\begin{equation}
N_{\text{bath}}=N+M+\kappa_3-2\sim \log( T\varepsilon^{-1})\log(\varepsilon^{-1}).
\end{equation}
\end{proof}
\end{theorem}

Our proof relies on the analytic properties of the spectral density and its asymptotic behavior at low and high frequencies. Neither of these properties are affected by the prefactor $(1-n_{\text{F}}(\omega))$, and our proof holds uniformly for any value of $\beta$. The underlying physics, however, strongly depends on temperature, and it is known that more compact descriptions can be obtained for weakly coupled impurities and high temperatures. The advantage of our method is that it is applicable to arbitrary impurity realizations, including non-perturbative regimes. Further compression can be achieved by combining the pseudomode description with standard tensor compression techniques applied during time evolution.

\subsection{Error Bound on Observables}
\label{sec:Bound_on_the_observables}
At this stage, we assume that we have found an explicit exponential approximation of the kernel function in the form of Eq.~\eqref{eq:exponential_approx}. This yields the following decomposition of the hybridization function:
\begin{equation}
\Delta(\tau-\tau^{\prime})=\tilde{\Delta}(\tau-\tau^{\prime})+\delta(\tau-\tau^{\prime}),
\end{equation}
\begin{equation} \label{eq:delta_bound}
\int\limits_0^{T}|\delta(\tau)|d\tau<\varepsilon,
\end{equation}
where $\tilde{\Delta}(\tau-\tau^\prime)$ is described by a finite number $N_{\text{tot}}$ of pseudomodes. The explicit relation between the exponential approximation of a particle/hole kernel function pseudomodes is discussed in detail in Sec.~\ref{sec:dynamics}.

In this subsection we derive a bound on the error of expectation values of local observables, as defined in Eq.~\eqref{eq:expect_value}. Our starting point is an approximated hybridization function with error $\varepsilon$. We will prove that neglecting the small correction $\delta(\tau-\tau^\prime)$ induces an error on observables that is of order $\varepsilon T$. 

The correlator from Eq.~\eqref{eq:expect_value}
can be rewritten as follows:
\begin{multline}\label{eq:correction}
\langle \hat{\mathcal{O}}_{\text{imp}}(t)\rangle_{\tilde{\Delta}}=\\ = \langle\exp\bigg(-\int\limits_{\mathcal{C}}d\tau\int\limits_{\mathcal{C}}d\tau^\prime{\bar{\eta}_{\tau}\delta(\tau-\tau^\prime)\eta_{\tau^\prime}}\bigg)\hat{\mathcal{O}}_{\text{imp}}(\bar{\boldsymbol{\eta}}_\tau,\boldsymbol{\eta}_\tau) \rangle,
\end{multline}
where $\langle ... \rangle_{\tilde{\Delta}}$ denotes the same averaging as in Eq.~\eqref{eq:expect_value} but with respect to a hybridization function $\tilde{\Delta}$.

The expression in Eq.~\eqref{eq:correction} can be rewritten as a series:
\begin{multline}
\langle \hat{\mathcal{O}}_{\text{imp}}(t)\rangle-\langle \hat{\mathcal{O}}_{\text{imp}}(t)\rangle_{\tilde{\Delta}}=\\=\sum\limits_{n=1}^{\infty} \frac{(-1)^n}{n!}\int\limits_{\mathcal{C}}\prod\limits_{k=1}^n \bigg(d\tau_kd\tau^\prime_k\bigg)\prod\limits_{k=1}^{n}\delta(\tau_k-\tau^\prime_k)\times \\ \times \langle \prod\limits_{k=1}^n\bar{\eta}_{\tau_{k_1}}\eta_{\tau^\prime_{k_1}} \hat{\mathcal{O}}_{\text{imp}}(\bar{\boldsymbol{\eta}}_\tau,\boldsymbol{\eta}_\tau)\rangle_{\Delta}.
\end{multline}
As the hybridization function $\Delta$ describes physical evolution, we have a natural bound for the correlators in the above equation:
\begin{equation}\label{eq:correlator_bound}
\langle \prod\limits_{k=1}^n\bar{\eta}_{\tau_{k_1}}\eta_{\tau^\prime_{k_1}}\hat{\mathcal{O}}_{\text{imp}}(\bar{\boldsymbol{\eta}}_\tau,\boldsymbol{\eta}_\tau)\rangle_{\Delta}<\|\hat{\mathcal{O}}_{\text{imp}}\|,
\end{equation}
where we consider the class of observables $\hat{\mathcal{O}}_{\text{imp}}$ which are products of local fermionic operators at different times. The norm of the operator is defined in the usual operator sense:
\begin{equation}
\|\hat{\mathcal{O}}_{\text{imp}}\|=\text{max}_{\|v\|=1}\langle v|\hat{\mathcal{O}}_{\text{imp}}|v\rangle.
\end{equation}
Using this relation, together with the bound from Eq.~\eqref{eq:delta_bound}, we note that the correction in Eq.~\eqref{eq:correction} is indeed bounded:
\begin{equation}
|\langle \hat{\mathcal{O}}_{\text{imp}}(t)\rangle-\langle \hat{\mathcal{O}}_{\text{imp}}(t)\rangle_{\tilde{\Delta}}|<(e^{\varepsilon T}-1) \|\hat{\mathcal{O}}_{\text{imp}}\|.
\end{equation}
This proves the desired bound and shows that if $\varepsilon T$ is small, then all local observables are well approximated for the hybridization function $\tilde{\Delta}$.

\section{Numerical Scaling and Interpolative Compression}\label{sec:compression}

\subsection{Overview}
In this Section, we present a numerical approach to determine an optimal set of modes to approximate a kernel function to a given accuracy, see Eq.~(\ref{eq:exponential_approx}). To provide a comprehensive understanding, we begin by reviewing the interpolative matrix decomposition (ID). We then leverage ID to compress a set of modes by connecting their spectral information to the temporal correlations induced by each mode.

This approach enables us to compare the scaling obtained from the analytical construction in Sec.~\ref{sec:bound}---also validated here numerically---with that of the compressed, and arguably optimal, ensemble of modes. As one of the main results of this Section, we establish the scaling of the number of compressed modes with $T,\epsilon$ in Eq.~(\ref{eq:Nbath_improved}). 
Moreover, we provide insights into the nature of the numerical frequency renormalization induced by ID.\\

 For the numerical analysis, we introduce a discrete time grid with points: \begin{equation}\label{eq:time_grid}
t_i = i\cdot \delta t, \end{equation} where $
i \in \{0,\dots, N_t\}$ and $
\delta t = T / N_t.$ On this time grid, we define the relative error of the time series $\tilde{\Delta}(t_i)$ with respect to a reference time series, $\Delta(t_i),$ as
\begin{equation}\label{eq:numerical_error}
    \epsilon = \sum_{i=1}^{N_t} \frac{|\tilde{\Delta}(t_i) - \Delta(t_i)|}{\sum_{j=1}^{N_t}\big(|\tilde{\Delta}(t_j)| + |\Delta(t_j) |\big)}.
\end{equation} 
This definition of the numerical error differs from the error defined in   Eq.~\eqref{eq:exponential_decomposition} and used in the proof. Such definition is convenient since it yields a  dimensionless error, with a numerical value between 0 and 1, that can be used as numerical error tolerance for the compression algorithm introduced in this Section. We note, however, that the two error definitions differ by a factor given by the $L_1$ norm of the kernel function $\|\Delta(t)\|_{L_1}$. All the kernel functions considered below are bounded and decay at large times as $1/t$ or more quickly, therefore their $L_1$ norm is bounded as $\|\Delta(t)\|_{L_1}<C \log T$. Such a contribution does not change the final scaling of pseudomode number on $T,\epsilon$, only generating a subleading contribution. Unless stated otherwise, all numerical data presented in this Article uses the definition of the error in Eq.~(\ref{eq:numerical_error}), where the reference time series $\Delta(t_i)$ is the kernel function obtained by exact numerical integration.

\subsection{Interpolative Matrix Decomposition}\label{sec:interpolative_decomposition}
The interpolative decomposition (ID) \cite{GuEfficient, ChengCompression2005, LibertyRandomized, Woolfe08Fast} is the approximation of a matrix $\bm{A}\in \mathbb{C}^{m\times n}$ as the product of a matrix containing $r$ selected columns of $\bm{A}$ and a ``projection matrix'' $\bm{P} \in \mathbb{C}^{r\times n}$:
\begin{equation}\label{eq:ID}
\bm{A} \approx  \bm{A}(:,\mathcal{J}) \bm{P}, \quad \mathcal{J} \subset \{1, \dots, n\}. 
\end{equation}\\
The ID possesses a range of useful mathematical properties that have been comprehensively presented, for example, in Ref.~\cite{LibertyRandomized}. Here, we briefly summarize the key conceptual aspects that we exploit in this work.

First, the error in Eq.~(\ref{eq:ID}) is bounded as
\begin{equation}
    \left\| \bm{A} - \bm{A}(:,\mathcal{J}) \bm{P} \right\|_{2} \leq \sqrt{r(n-r)+1} \sigma_{r+1},
\end{equation}
where $\sigma_{r+1}$ is the $(r+1)^\text{th}$ largest singular value of $\bm{A}$. This motivates a comparison to singular value decomposition (SVD). For truncated SVD, the error guarantee is only slightly better, namely $\epsilon \leq \sigma_{r+1}$. Conceptually, however, ID and SVD differ in that SVD involves a change of basis, whereas ID only selects columns from the original matrix $\bm{A}$. This property is crucial for our approach, where $\bm{A}$ encodes the relationship between spectral properties and temporal correlations generated by each mode. Importantly, this physical interpretation of $\bm{A}$ is preserved during compression with ID.

Second, the \textit{interpolative} property of ID ensures that all columns with index $j \in \mathcal{J}$ are represented exactly, i.e.,
\begin{equation}\label{eq:interpolation_property}
\forall i: \, \big(\bm{A}(:,\mathcal{J}) \bm{P}\big)_{ij} = \bm{A}_{ij} \text{ if } j \in \mathcal{J},
\end{equation}
while all other entries are interpolated between these. A manifestation of this is that a subset of the columns of the projection matrix $\bm{P}$ forms the $r \times r$ identity matrix. 
Lastly, we note that the rank $r$ needed to represent $\bm{A}$ exactly is bounded by the smaller of the two matrix dimensions:
\begin{equation}\label{eq:rank_bound}
r \leq \text{min}(m,n).
\end{equation}
\begin{comment}

\end{comment}
\begin{figure}[h] % 'h' stands for 'here', try to place the figure in this position
    \centering % Center the figure
    \includegraphics[width=.95\linewidth]{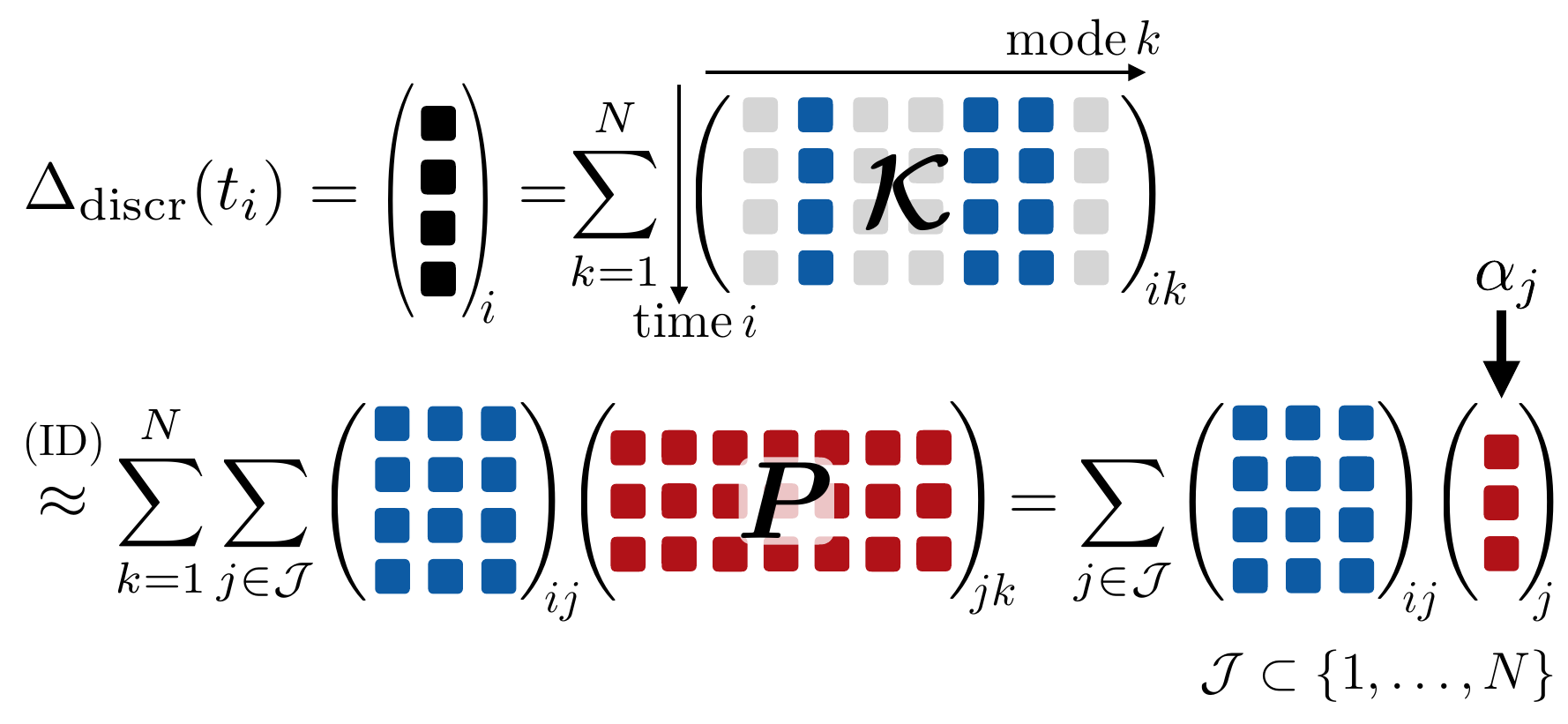} % Adjust the width as needed
    \caption{Schematic illustration of the mode compression: Performing an interpolative decomposition of the kernel matrix $\mathcal{K}$ 
    selects a subset of modes $\mathcal{J}$ with renormalized couplings encoded in the weights $\alpha_j.$} 
    \label{fig:kernel_decomp} % Add a label for referencing
\end{figure}

Next, we demonstrate how ID can be used to identify an optimal set of modes that approximate the kernel functions in Eqs.~(\ref{eq:particle_integral},\ref{eq:hole_integral}) within a specified accuracy. We will show that this approach results in a renormalized spectral density for the compressed set of modes.

\subsection{Compression of Kernel Functions}

The basis for the applying ID lies in approximating the integrands of Eqs.~(\ref{eq:particle_integral}--\ref{eq:hole_integral}) using a finite set of modes, see Eq.~(\ref{eq:exponential_approx}). 
On a discrete time grid, Eq.~(\ref{eq:exponential_approx}) can be expressed as 
\begin{equation}
\tilde{\Delta}(t_i) = \sum_{k=1}^{N_\text{bath}} \bm{\mathcal{K}}_{ik},
\end{equation}
where we defined the kernel matrix with elements
\begin{equation}\label{eq:kernel_def}
    \bm{\mathcal{K}}_{ik} = \Gamma_k e^{i\omega_k t_i}.
\end{equation}
Rows and columns of $\bm{\mathcal{K}}$ correspond to time step and mode index, respectively.

By applying ID to the kernel matrix with a predefined relative error tolerance $\epsilon_\text{ID}$,
\begin{equation}\label{eq:ID_on_kernel}
\bm{\mathcal{K}}_{ik} \stackrel{\text{ID}}{\approx} \sum_{j \in \mathcal{J}} \bm{\mathcal{K}}_{ij} \bm{P}_{jk},
\end{equation}
we identify a subset of modes $\mathcal{J}\subset \{1,\dots,N_\text{bath}\}$ that suffices to approximate the kernel function with an error of $\mathcal{O}(\epsilon_\text{ID})$,
\begin{equation}
\tilde{\Delta}(t_i) \stackrel{\text{ID}}{\approx}
\sum_{j \in \mathcal{J}} \bm{\mathcal{K}}_{ij} \sum_{k=1}^{N_\text{bath}} \bm{P}_{jk} = \sum_{j \in \mathcal{J}} \bm{\mathcal{K}}_{ij} \alpha_j.
\end{equation}

Here, we defined the weights,
\begin{equation}
\alpha_j = \sum_{k=1}^{N_\text{bath}} \bm{P}_{jk},
\end{equation}
which renormalizes the couplings of the modes that have been selected by ID. This compression step is schematically illustrated in Fig.~\ref{fig:kernel_decomp}. After ID, the compressed kernel matrix is given by
\begin{equation}
    \mathcal{K}^\text{ID}_{ik} = (\Gamma_k \alpha_k) e^{i\omega_k t_i}, \text{ with } k \in \mathcal{J}.
\end{equation} 

In practice, we set the error tolerance \(\epsilon_\text{ID}\) to match the error inherent in the finite-mode approximation of \(\bm{\mathcal{K}}\). This ensures that the application of ID does not significantly increase the total error of $\tilde{\Delta}(t_i),$ relative to the exact kernel functions, Eqs.~(\ref{eq:particle_integral}--\ref{eq:hole_integral}). By doing so, we achieve stronger compression, especially when the discretization error is substantial.
Our numerical implementation for compressing kernel matrices utilizes the \texttt{SciPy} Python package~\cite{Virtanen20Scipy}, which includes implementations of the matrix decomposition algorithms described in Refs.~\cite{ChengCompression2005, LibertyRandomized}.

\subsection{Numerical Scaling and Bath Renormalization}
As a first numerical application, we follow the construction of the proof in Sec.~\ref{sec:bound}, confirm the predicted scaling numerically, and compare it to the scaling after compression with ID. Interestingly, we find the latter to be slightly more compact than the theoretically predicted one.

In order to obtain a realistic estimate for the number of modes $N_\text{bath}$ needed for a given error $\epsilon$ before compression, we begin by optimizing the lower and upper cutoffs of the frequency grid numerically.

\subsubsection{Optimizing the Kernel Matrix}

We consider $\tilde{\Delta}_+^p(t)$ from Eq.~(\ref{eq:rotated_integral_approx}), where the frequency integration contour is rotated into the complex plane by an angle $r_{\text{max}}=\pi/4$, and discretized with points according to Eqs.~(\ref{eq:first_variable_transform}--\ref{eq:second_variable_transform}):
\begin{equation}\label{eq:exp_param_numerical}
\omega_k =\Gamma e^{hk} e^{i\pi/4}.    
\end{equation} 
As in Sec.~\ref{sec:bound}, we restrict ourselves to the particle component and to frequencies with positive real part.
On a discrete time grid such as defined in Eq.~(\ref{eq:time_grid}), one has
\begin{equation}\label{eq:discr_approx}
    \tilde{\Delta}_+^p(t_i) = \sum_{k=-M}^{N} \bm{\mathcal{K}}_{ik},
\end{equation} with 
\begin{align}
\bm{\mathcal{K}}_{ik} &= \Gamma_k e^{i\omega_{k} t_i},\label{eq:kernel_matrix_flatband}\\\label{eq:couplings}
    \Gamma_k &= \frac{h}{2\pi} \omega_k \, \Gamma(\omega_k) \, \big(1-n_\text{F}(\omega_k)\big).
\end{align}
In Eq.~(\ref{eq:couplings}), the factor $h\omega_k$ is the Jacobian resulting from frequency discretization.
The high and low energy cutoffs are determined by the summation limits $M$ and $N$ in Eq.~(\ref{eq:discr_approx}), respectively. For the numerical application, we start by setting $M$ and $N$ to large values, ensuring that the frequency range covered by the discrete sum includes all relevant scales. This approach guarantees that the approximation error, $\epsilon,$ of the finite sum is solely controlled by the discretization parameter $h.$ In order to obtain a meaningful estimate of the number of modes,
\begin{equation}
    N_\text{bath} = N + M,
\end{equation} needed for a given value of $h,$ we then reduce the values of $M$ and $N$ incrementally until the relative error between the optimized and unoptimized versions of $\tilde{\Delta}_+^p(t_i)$ grows to $1\%$ of the discretization error $\epsilon,$ doing so independently for both $M$ and $N.$ 

\subsubsection{Example: Numerical Scaling for a Flat Band}

As a first illustration, we choose a wide flat band with smooth cutoffs, defined by the spectral density
\begin{equation}\label{eq:flat_band}
    \Gamma_\text{flat}^\Lambda(\omega)  = \frac{\Gamma}{(1+e^{(\omega-\Lambda)\nu})(1+e^{-(\omega+\Lambda)\nu})}.
\end{equation}
Here $\nu$ defines the sharpness of the cutoff at frequency $\omega =\pm\Lambda.$ In this Section, we fix these parameters to $\Lambda = 10^5\, \Gamma$ and $\nu = 20/\Lambda,$ respectively. We note that this spectral density satisfies the conditions necessary for Theorem~ \ref{th:finite_approximation} to hold.

\begin{figure}[h] % 'h' stands for 'here', try to place the figure in this position
    \centering % Center the figure

    \begin{minipage}[b]{.48\textwidth}
    \begin{overpic}[width=\linewidth]{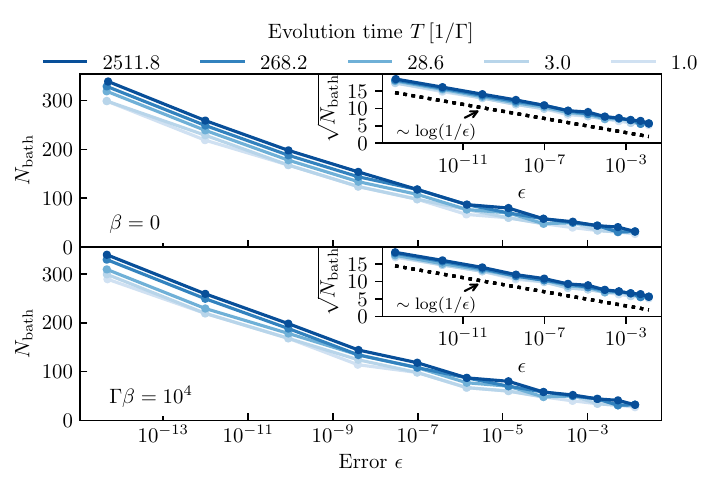}
    \put(5,65){\text{a)}}
    \end{overpic}
    \end{minipage}

    \begin{minipage}[b]{.48\textwidth} % 'h' stands for 'here', try to place the figure in this position
   \begin{overpic}[width=\linewidth]{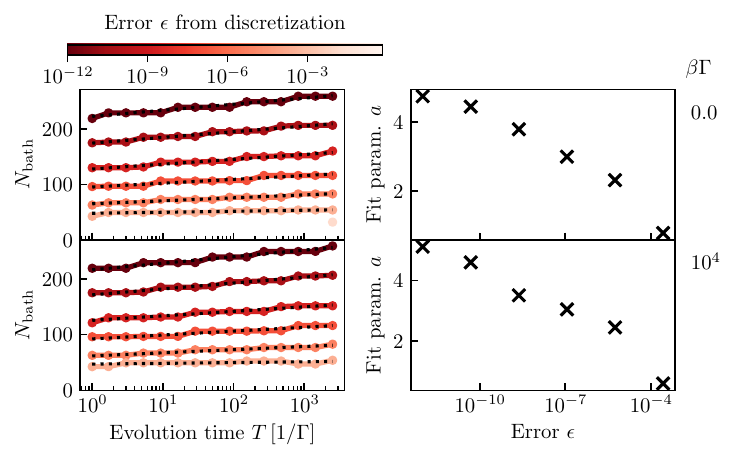} 
   \put(5,60){\text{b)}}
    \end{overpic}

\end{minipage}
    \caption{Before compression: Scaling of $N_\text{bath}$ for a {\it wide flat band}, $\Gamma_\text{flat}^\Lambda(\omega),$ with width $\Lambda/\Gamma =  10^5$ and sharpness $\nu = 20 /\Lambda$ at inverse temperatures $\beta  = 0$ (top) and $\beta = 10^4$ (bottom). This figure refers to the positive-frequency branch of the particle component, $\Delta^p_+(t)$ in Eq.~(\ref{eq:cont_integral}). a) Scaling of $N_\text{bath}$ with the error $\epsilon.$ Inset: Squareroot of same data, illustrating the scaling law $N_\text{bath}\sim \log^2(1/\epsilon).$ The dotted line serves as guide to the eye. b) Left column: Scaling of number of modes with the time for fixed error (interpolated). The curves are fitted with a function $a \log(T) + \text{const.}$ (dotted). Right column: Dependence of fit parameter $a$ on the error. } 
    \label{fig:scaling_wideband} % Add a label for referencing
    \end{figure}

\begin{figure}[h]
\centering
\begin{minipage}[b]{.48\textwidth}
    \begin{overpic}[width=\linewidth]{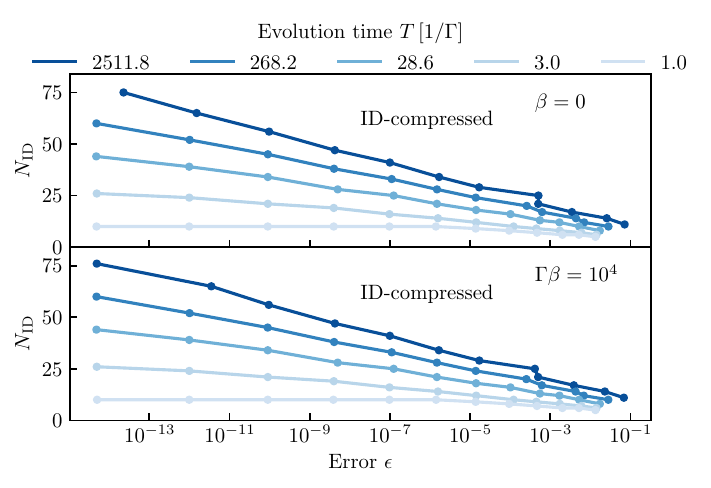}
    \put(5,65){\text{a)}}
    \end{overpic}
    \end{minipage}
    \begin{minipage}[b]{.48\textwidth} % 'h' stands for 'here', try to place the figure in this position
    \begin{overpic}[width=\linewidth]{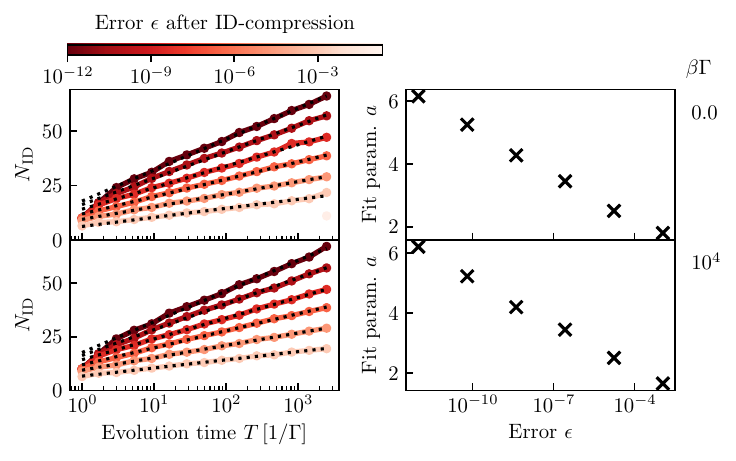} \put(5,60){\text{b)}}
    \end{overpic}
\end{minipage}
\caption{After compression: Scaling of $N_\text{ID}$ for a {\it wide flat band}. All parameters are as specified in Fig.~\ref{fig:scaling_wideband}. a) Scaling of the number of modes retained after compression, $N_\text{ID},$ with the error $\epsilon$.  b) Left column: Scaling of number of modes with time for fixed error (interpolated). The curves are fitted with a function $a \log(T) + \text{const.}$ (dotted). Right column: Dependence of the fit parameter $a$ on the error.}
  \label{fig:scaling_ID_wideband} % Add a label for referencing
\end{figure}

\paragraph{Scaling before compression.} In Fig.~\ref{fig:scaling_wideband}, we examine the scaling of $N_\text{bath}$ that we determined after employing the described optimization approach. As main result here, we confirm the scaling predicted in Sec.~\ref{sec:bound} and give an estimate for the absolute values of $N_\text{bath}.$ 

In Fig.~\ref{fig:scaling_wideband}a, we show $N_\text{bath}$ as function of the error $\epsilon$ which is defined with respect to the continuous integral along the same path in the complex plane as the discrete frequencies, see Eq.~(\ref{eq:exp_param_numerical}).
Importantly, we observe the predicted polylogarithmic dependence on the error, $
N_\text{bath} \sim \log^2(1/\epsilon),$ for different fixed values of the evolution time $T.$ This behavior, spanning many orders of magnitude, becomes especially apparent in the inset, where we show $\sqrt{N_\text{bath}}$ as a function of $\epsilon,$ which thus yields a straight line on the logarithmic error scale. We show data computed for two different values of the inverse temperatures, $\beta = 0$ and $\beta\Gamma = 10^4$ (top and bottom, respectively). In line with our discussion in Sec.~\ref{sec:bound}, we observe no significant dependence on temperature.

Next, we study the numerical scaling of $N_\text{bath}$ with evolution time $T,$ shown in Fig.~\ref{fig:scaling_wideband}b for different fixed $\epsilon.$ Since \(\epsilon\) values are evaluated {\it a posteriori}, we employ numerical interpolation to map the values of $N_\text{bath}$ onto a predefined error grid across all $T$ values. 
Our results in the left column validate the predicted scaling for a fixed \(\epsilon\), namely
    $N_\text{bath}\sim \log(T).$
We fit these curves with the function \( a \log(T) + \text{const.} \) and present the \(\epsilon\)-dependence of the fit parameter \(a\) in the right column. As predicted in Sec.~\ref{sec:bound}, our findings indicate that the parameter \(a\) exhibits logarithmic dependence on the inverse error, \(a \sim \log(1/\epsilon)\).
Again, no temperature dependence is observed in these results, such that we identify the numerical scaling {\it before compression} as:
\begin{equation}\label{eq:scaling_before_compression}
    N_\text{bath}\sim \log(T/\epsilon)\log(1/\epsilon).
\end{equation} This coincides with the analytic predicition in Eq.~(\ref{eq:N_bath_scaling}).
In terms of absolute values on the error and time scales considered, we find $\mathcal{O}(N_\text{bath}) \sim 10^2.$

\paragraph{Scaling after compression.} 
Next, we employ ID to the kernel matrix in Eq.~(\ref{eq:kernel_matrix_flatband}) and find the number of modes $N_\text{ID}$ {\it after compression} to scale as
\begin{equation}\label{eq:numerical_scaling}
    N_\text{ID}\sim \log(T) \log(1/\epsilon),
\end{equation} 
with absolute values of $\mathcal{O}(N_\text{ID}) \sim 10^1$ for the error and time scales considered. 

In Fig.~\ref{fig:scaling_ID_wideband}a, we show $N_\text{ID}$ as a function of the error $\epsilon.$ Here, three aspects are noteworthy:
\begin{enumerate}
    \item[i)] Remarkably, the scaling of $N_\text{ID}$ with $\epsilon$ is at most $\log(1/\epsilon)$ for fixed evolution time $T.$ This is to be contrasted with the scaling $\log^2(1/\epsilon)$ for the uncompressed system, as predicted analytically and numerically confirmed in Fig.~\ref{fig:scaling_wideband}a.
\item[ii)] As opposed to the curves in Fig.~\ref{fig:scaling_wideband}a, we observe a strong dependence on the evolution time $T,$ in line with the expectation that longer evolution times require more frequency points in order to achieve the same error $\epsilon.$ We study this aspect in more details below. 
\item[iii)] As the number of modes required for an exact representation is bounded according to Eq.~(\ref{eq:rank_bound}), the compressed number of modes $N_\text{ID}$ for short evolution times $T$ coincides with the number of time steps as visible for the curve at $T\Gamma=1.0$ corresponding to $N_t = 10$ time steps. 
\end{enumerate}
We also note that the significant reduction of $N_\text{ID}$ values compared to $N_\text{bath}$ occurs despite the numerical optimization of the latter described above. Moreover, the ID results are unchanged even when ID is applied to the unoptimized kernel matrix.

To determine the scaling with evolution time $T$, we proceed analogously to before: In Fig.~\ref{fig:scaling_ID_wideband}b, we find $
    N_\text{ID}\sim \log(T)$ for fixed $\epsilon,$
    we fit the resulting curves, and find a scaling of the fit parameter of $a\sim \log(1/\epsilon).$ This suggests an overall scaling of $N_\text{ID}$ as stated in Eq.~(\ref{eq:numerical_scaling}).

Given the substantial effect of compression on scaling and absolute numbers, we next examine the physical background of this compression, specifically focusing on the density of frequencies selected by ID.

\begin{figure}
\centering
\includegraphics[width = \linewidth] {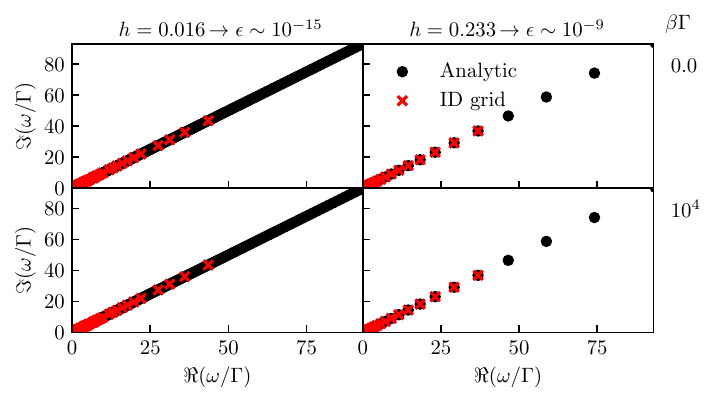}
\caption{Black dots: Complex frequencies according to the exponential frequency parametrization, Eq.~(\ref{eq:exp_param_numerical}). Red crosses: Subset of modes retained after compression with ID for an evolution time $T\Gamma \approx 2500.$}
\label{fig:explicit_grid}
\end{figure}

 \subsubsection{ID as Bath Renormalization} 
 In Fig.~\ref{fig:explicit_grid}, we visualize the frequency grid points in the complex plane: Black dots show the points of the exponential grid according to Eq.~(\ref{eq:exp_param_numerical}), and red crosses show the subset of frequencies selected by ID. The left and right columns refer to two different values of the discretization parameter $h$. This illustrates that for very fine grids (left column), the ID effectively samples only a small density of the frequency points, while for coarse grids, ID select most of the available points. In both cases, ID automatically incorporates a high energy cutoff determined by the finite time step (here $\delta t = 0.1/ \Gamma$) and does not sample frequencies beyond this point, here around $\Re(\omega) \approx 50 \Gamma.$ The low energy cutoff set by ID, not visible on this scale, is governed by $T$ and is further investigated below.

\begin{figure}
    \centering
    \begin{minipage}[b]{.48\textwidth}
    \begin{overpic}[width=\linewidth]{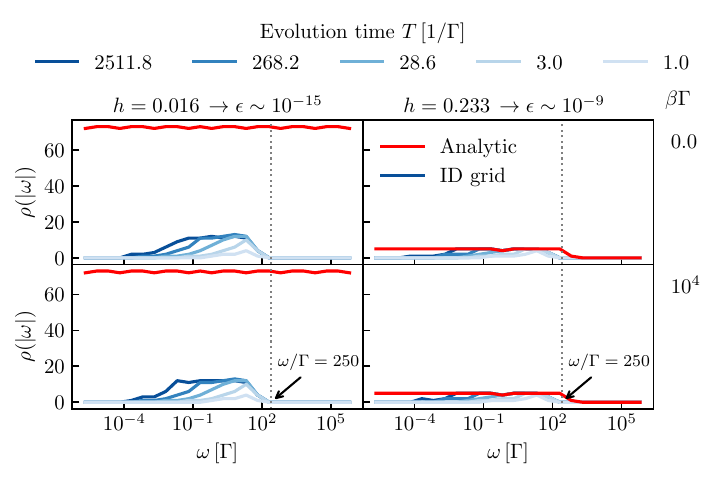}
    \put(5,62){\text{a)}}
    \end{overpic}
      \end{minipage}
       \begin{minipage}[b]{.48\textwidth}
    \begin{overpic}[width=\linewidth]{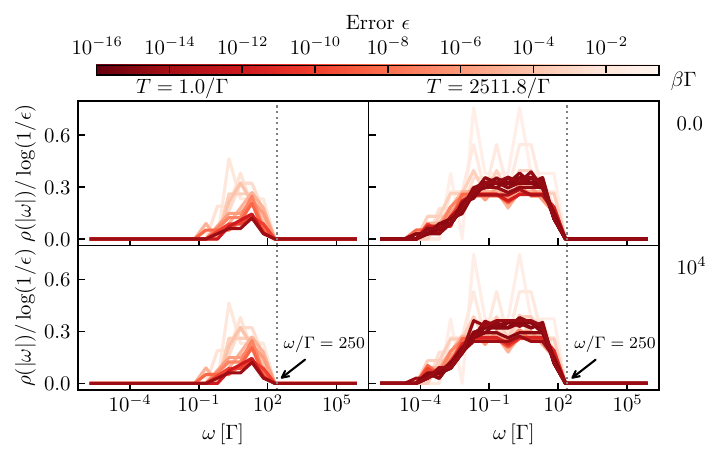}
    \put(5,60){\text{b)}}
    \end{overpic}
      \end{minipage}
      \caption{Density of frequency points, $\rho(|\omega|)$ for a {\it wide flat band}, $\Gamma_\text{flat}^\Lambda(\omega),$ with width $\Lambda/\Gamma =  10^5$ and sharpness $\nu = 20/ \Lambda$. Rows corresponds to different inverse temperatures $\beta$ as indicated. a) Frequency density as function of evolution time $T:$ Red shows the initial (time-independent) exponential grid and shades of blue show the compressed ID grid. The columns correspond to two different values of the discretization parameter $h$ as indicated. b) Frequency density of ID grid, rescaled by $\log(1/\epsilon)$ as a function of the error $\epsilon$ (shade of red), for two different evolution times $T$ (left and right column).}
      \label{fig:grid_flatband}
\end{figure}

 Important information is encoded in the distribution of selected frequency points $\rho(|\omega|),$ which we define as number of frequency points in a given interval on the {\it logarithmic} scale, \begin{equation}
     [10^{n\delta }, 10^{(n+1)\delta }],\quad  n\in \mathbb{Z}.
      \end{equation} Here, we set $\delta  = 0.5.$ 
      
  Fig.~\ref{fig:grid_flatband}a illustrates the distribution of frequency points for both pure exponential grid (red) and the ID grid (various shades of blue). Focusing initially on the left column, which represents a very fine initial grid, we observe that, for large $T,$ the ID grid also exhibits an exponential frequency spacing in a window which increases with higher values of \(T\).  
  %This is a nontrivial fact as the exponential spacing of the ID grid is nowhere imposed but emerges naturally.
  
Remarkably, the density of points of the ID grid is significantly lower than of the initial one and coincides for all evolution times $T$. We emphasize that in the process of ``coarsening'' the grid, the ID performs a renormalization of the couplings. Thus, rather than just determining a new frequency grid, the ID should be understood as an effective bath renormalization scheme that integrates out intermediate energy scales and redistributes their weights to the modes of a coarser exponential grid.
      
      In the right column, starting from a coarse initial grid, as $T$ is increased, the ID grid exhausts the available frequency points within the range of the energy cutoffs. Lastly, Fig.~\ref{fig:grid_flatband}b shows the frequency density of the ID grid, rescaled by $\log(1/\epsilon)$ for fixed $T$ as a function of the error $\epsilon.$ While the high energy cutoff is fixed, the low energy cutoff depends on $T$ and, for fixed $T,$ converges with $\epsilon\to 0.$ This explains the milder numerical scaling after compression compared to the analytical estimate, see Eqs.~(\ref{eq:scaling_before_compression}) and (\ref{eq:N_bath_scaling}), respectively.

      The results in Fig.~\ref{fig:grid_flatband} allow us to draw the following conclusions about ID for long evolution times $T:$
\begin{enumerate}
    \item[i)] The temporal correlations encoded in the kernel functions \(\Delta(t)\) are most efficiently approximated by modes with exponential frequency spacing and renormalized couplings to the impurity.
    \item[ii)] The density of exponential frequency points, \(\rho(|\omega|)\), is governed by the error \(\epsilon\).
    \item[iii)]  The width of the frequency window is determined by the parameters of the time grid, with the final time \(T\) setting the low-energy cutoff and the time step \(\delta t\) controlling the high energy cutoff.
\end{enumerate}

    Overall, our analysis provides a clear interpretation of the projection matrix $P$ in Eq.~(\ref{eq:ID_on_kernel}): It acts as a bridge that translates between different scales. Specifically, it allows for the mapping of functions from the fine grid to the coarse grid, reflecting a change in scale, and conversely, enables the coarsened kernel matrix to be mapped back to the fine grid.

\section{ ``AAA'' Algorithm to Determine Modes}
\label{sec:AAA}

\subsection{Overview}

In the previous sections, we exploited the freedom to rotate the integration contour into the complex plane within the integrals defining the kernel functions, Eqs.~(\ref{eq:particle_integral},\ref{eq:hole_integral}). These integrals have the form of a standard Fourier transform,
\begin{equation}\label{eq:general_form}
\Delta(t) = \int \frac{d\omega}{2\pi} f(\omega) e^{i\omega t},
\end{equation} where $f(\omega)$ is determined by the spectral density $\Gamma(\omega)$ and the Fermi-Dirac distribution $n_\text{F}(\omega).$
So far, we have relied on the analytic understanding of \( f(\omega) \), allowing for an analytic continuation into the complex plane and ensuring that any poles crossed during the contour deformation can be accounted for exactly. As the Matsubara poles stemming from \( n_\text{F}(\omega) \) are located on the imaginary axis and therefore do not interfere with the contour deformation, all poles of \( f(\omega) \) relevant to the contour deformation are determined by the spectral density $\Gamma(\omega)$. 

In this Section, we give up the requirement of knowing $f(\omega)$ and its pole structure analytically, which generalizes our approach to arbitrary spectral densities $\Gamma(\omega).$ Our strategy consists of two main steps:
\begin{enumerate}
    \item Using the adaptive
Antoulas–Anderson (AAA) algorithm Ref.~\cite{NakatsukasaAAA,BerrutBarycentric}, we determine a rational approximation of $f(\omega),$ which is associated with a finite number of poles and their residues in the complex plane. By closing the integration contour in the upper half-plane, we can approximate $\Delta(t)$ as a finite sum of decaying exponentials as in Eq.~(\ref{eq:exponential_approx}).
\item By expressing this finite sum as a kernel matrix $\bm{\mathcal{K}}$ as in Eq.~(\ref{eq:kernel_def}), we can subsequently use ID to obtain a compressed ensemble of modes for the approximation of $\Delta(t)$.
\end{enumerate}

Before presenting algorithmic details and numerical results, we highlight two key aspects that are crucial for the wide applicability of this method in practical applications. First, the AAA algorithm requires as input a discrete set of sample points, \(\big(\omega_i, f(\omega_i)\big)\). This makes it directly applicable to spectral functions that are available only as numerical data, as is often the case in computational methods such as dynamical mean-field theory approaches. Second, when applied to analytical spectral densities $\Gamma(\omega),$ a careful choice of the sample points $\omega_i$ allows to accurately approximate spectral densities with cusps.

We begin by briefly sketching the main aspects of the AAA algorithm. A comprehensive presentation can be found in the original derivation, Ref.~\cite{NakatsukasaAAA}, which we closely follow here.

\subsection{AAA Algorithm}\label{sec:AAA_algorithm}
\paragraph{Construction.}
The AAA algorithm provides a stable method for determining a rational barycentric interpolation approximation of a given function \( f(z) \) \cite{BerrutBarycentric, NakatsukasaAAA}: Starting with a set of \( M \gg 1 \) distinct support points \( Z = \{z_1, \dots, z_M\} \subset \mathbb{C}\) and their corresponding function values \( F = \{f_1 = f(z_1), \dots, f_M = f(z_M)\} \), the AAA algorithm iteratively selects support points from \( Z \) to construct a rational approximation of $f(z)$ of the form:
\begin{equation}\label{eq:barycentric_formula}
    r(z) = \sum_{j \in \mathcal{J}} \frac{w_j f_j}{z - z_j} \Bigg/ \sum_{m \in \mathcal{J}} \frac{w_m}{z - z_m}.
\end{equation}
Here, \( \mathcal{J} \) is an index set updated iteratively by adding the index of each newly selected support point, and \( Z_\mathcal{J} \) denotes the subset of points \( z_j \) with \( j \in \mathcal{J} \).
The interpolation property ensures that the rational approximation \( r(z) \) matches the exact function values at all selected support points, i.e., \( r(z_j) = f(z_j) \) for all $j\in \mathcal{J}.$
 The barycentric weights \(\{w_j\}\) are optimized at each iteration in order to minimize the global error of all $f(z_i)$ with $z_i \in Z\setminus Z_\mathcal{J}$. At every step, the next support point $z_i$ is greedily selected as the one where the error \( f(z_i)- r(z_i) \) is the largest. The algorithm terminates when the error at all sample points falls below a defined error threshold. Our numerical implementation builds upon the \texttt{Bayrat} Python package from Ref.~\cite{Hofreither21Algorithm}, which we have tailored to the applications discussed in this work.

\paragraph{Constructing the Sample Grid $Z$.}
Since the accuracy of polynomial interpolation is highly sensitive to the choice of support points, we construct the set \(Z\) using composite Chebyshev grids, which are known---and in some cases proven~\cite{Gimbutas20Fast}---to be optimally suited for Lagrange interpolation. In particular, following the approach of Ref.~\cite{KayeDiscrete}, we define dyadically refined partitions towards the origin, with interval boundaries $[a_i,b_i]$ determined by $a_1 = 0$, $a_i = b_{i-1} = 2^{-(m-i+1)}$ for $i=2,\dots,m$, and $b_m = 1$. Within each interval, we construct Chebyshev grids of order $p$. By assembling different patches of this grid, we optimize the resolution of local features in the function $f(\omega)$. In practice, we choose $p \approx 60$ and $m\approx 50,$ and verify that our results are converged in these parameters. 

\paragraph{Application to pseudomode decomposition.}
To approximate the kernel functions \(\Delta^p(t)\) and \(\Delta^h(t)\) from Eqs.~(\ref{eq:particle_integral}--\ref{eq:hole_integral}) in terms of pseudomodes, we first find a rational approximation of the product of spectral density and the Fermi distribution,
\begin{align}\label{eq:integrand_particle}
    f^p(\omega)  = \Gamma(\omega) \,\big(1 - n_\text{F}(\omega) \big) &\xrightarrow{\text{AAA}} r^p(z), \\
    \label{eq:integrand_hole}
    f^h(\omega)  = \Gamma(\omega) \,n_\text{F}(\omega) &\xrightarrow{\text{AAA}} r^h(z),
\end{align}
for the particle and hole kernel functions, respectively.

For our purposes, it is useful to express the rational approximations \( r(z) \) in terms of their poles $\Omega_k$ and corresponding residues $R_k,$
\begin{equation}
r(z) = \sum_{k} \frac{R_k}{z - \Omega_k}.
\end{equation}
Substituting $f(\omega)\to r(\omega)$ in the frequency integral, Eq.~(\ref{eq:general_form}), and closing the integration contour in the upper half-plane yields the desired approximation of the kernel function as a finite sum of exponentials:
\begin{equation}
    \Delta(t) =  \sum_{k \, : \, \Im(\Omega_k) > 0} (iR_k) e^{i\Omega_k t}.
\end{equation}

%\subsection{Combining AAA and ID}
%\label{sec:AAA_compression}
Up to this point, the algorithm is fully based on the {\it spectral} content of $\Delta(t),$ namely the function $f(\omega).$ However, as we numerically demonstrate below, the resulting sets of pseudomodes can be further significantly compressed with ID, analogously to Sec.~\ref{sec:interpolative_decomposition}.
As before, this is thanks to their dissipative nature combined with the explicit incorporation of the simulation time scales $\delta t$ and $T$ through the kernel matrix $\bm{\mathcal{K}}.$ For completeness, we explicitly state the relation defining the elements of the kernel matrix for a discrete time grid with points $t_j$, and refer to Sec.~\ref{sec:AAA_ID_flatband} for a numerical illustration:
\begin{equation}\label{eq:kernel_AAA}
    \bm{\mathcal{K}}_{jk} = iR_k\, e^{i\Omega_k t_j}.
\end{equation}

\subsection{Illustration: Approximating a Flat Band}
\label{sec:AAA_ID_flatband}

To illustrate the approach combining AAA and ID, we again consider a flat band defined by $\Gamma_\text{flat}(\omega)$ in Eq.~(\ref{eq:flat_band}) of width $\Lambda = 50\Gamma,$ and sharpness $\nu = 20/\Lambda.$

  \begin{figure}

    \begin{minipage}[b] {.48\textwidth}
    \begin{overpic}[width=\linewidth]{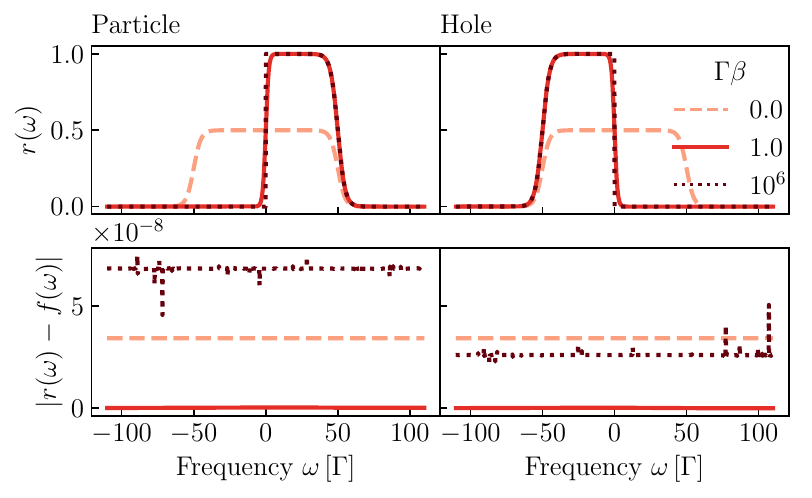} 
    \put(5,59){\text{a)}}
    \end{overpic}

    \end{minipage}
  
    \begin{minipage}[b]{.48\textwidth}
    \begin{overpic}[width=\linewidth]{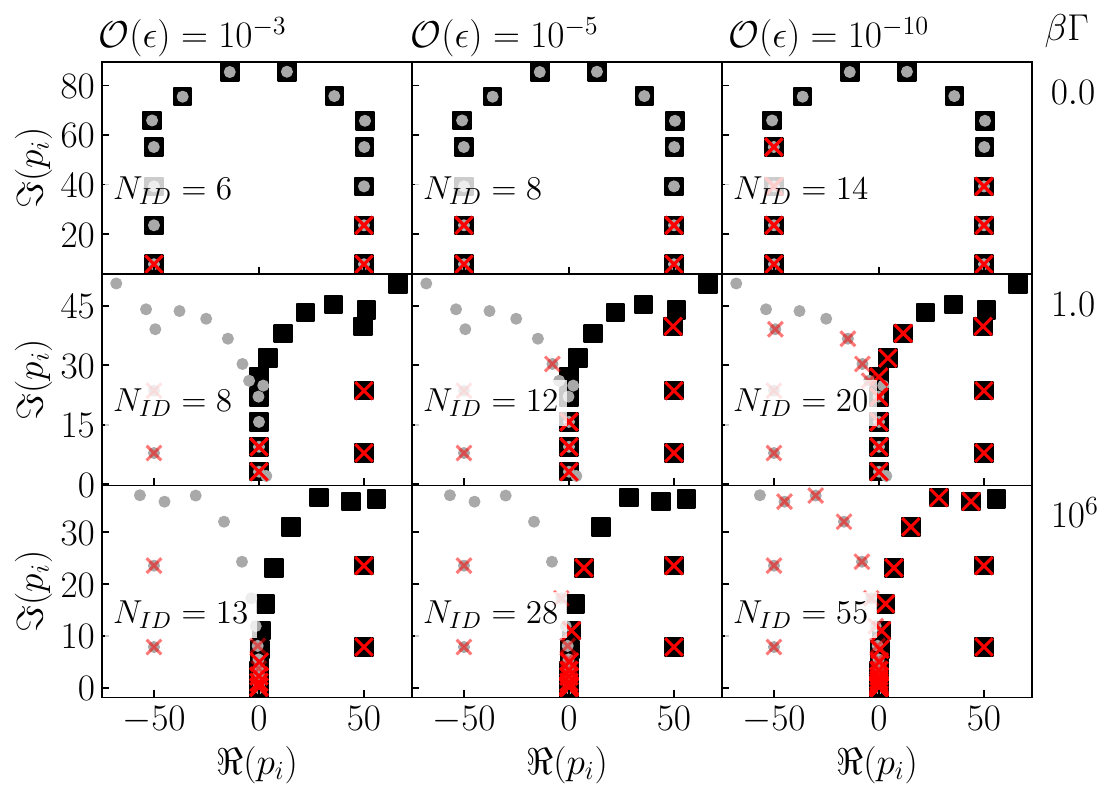} 
    \put(5,70){\text{b)}}
    \end{overpic}
    \end{minipage}
    \caption{Rational approximation from AAA for a {\it flat band}, $\Gamma_\text{flat}^\Lambda(\omega),$ with widths $\Lambda/\Gamma =  50$ and sharpness $\nu = 20/\Lambda$. a) Absolute error between the rational approximations $r(\omega)$ and the exact kernel function $f(\omega)$. The left and right column show the particle and hole component, respectively. Linestyles correspond to different inverse temperatures $\beta.$ b) Poles $p_i$ in the upper half-plane as determined by the AAA algorithm for the particle (black square) and hole (gray circle) components. The poles marked in red are retained after compression with ID. Columns corresponds to different error scales; rows correspond to different inverse temperatures $\beta$ as indicated. The respective number of pseudmodes after compression, $N_\text{ID}$, is given as inset. The ID has been performed for an evolution time $T=100/\Gamma.$}
    \label{fig:AAA_flatband}
    
\end{figure}

\subsubsection{AAA Approximation of $f(\omega)$}
\label{sec:AAA_Approximation_of_integrand}
We apply AAA according to Eqs.~(\ref{eq:integrand_particle},\ref{eq:integrand_hole}) in order to obtain a rational approximation of the functions $f^p(\omega)$ and $f^h(\omega)$ for the particle and hole component, respectively. The results are shown in Fig.~\ref{fig:AAA_flatband}a for three different values of inverse temperature $\beta$. The top row shows the rational approximation $r(\omega)$ as function of $\omega,$ while the bottom row shows the absolute error between $r(\omega)$ and the analytically known function $f(\omega),$ here converged to an error of $\mathcal{O}(\epsilon) = 10^{-8}$~\footnote{The differences in the particle and hole component reflect algorithmic subtleties related to the order of sample points and do not contradict the equivalence of the particle and hole component upon substituting $\omega \to -\omega.$}.

In Fig.~\ref{fig:AAA_flatband}b, we show the location of the poles as determined by AAA (grey squares and dots), where each row corresponds to a different value of $\beta$ as indicated. This reveals a few interesting aspects:
\begin{enumerate}
    \item In all cases, AAA correctly identifies the poles that can be derived analytically from the cutoff function in Eq.~(\ref{eq:flat_band}). They have a real part $\pm \Lambda = \pm 50\Gamma,$ and an imaginary part $\frac{2(n+1)\pi}{\nu} = \frac{2(n+1)\pi}{20/50},$ with $n\in \mathbb{N}.$
    \item For intermediate finite temperature, $\beta = 1/\Gamma,$ AAA determines the expected finite-temperature Matsubara poles stemming from the Fermi-Dirac distribution, located on the imaginary axis at positions $\Omega_n = 2(n+1)\pi$ with $n\in \mathbb{N}.$
    \item At a very low temperature, $\beta = 10^6 /\Gamma,$ AAA determines a very high density of poles around $\Re(\omega) = 0,$ reflecting that the Matsubara frequency spacing goes to zero for vanishing temperature.
    \item In all cases, the analytically known pole structures---which, in principle, extend to infinitely large imaginary parts---are determined up to some finite value on the imaginary axis, above which they are smoothly connected by further poles. This point reflects the approximate nature of the AAA algorithm.
\end{enumerate}

\subsubsection{ID Compression} 

Next, we construct the kernel matrix according to Eq.~(\ref{eq:kernel_AAA}) and apply ID. As before, we use \(\delta t = 0.1/\Gamma\) and, this time, consider a final evolution time of \(T = 100 /\Gamma\). In Fig.~\ref{fig:AAA_flatband}b, the poles retained during ID compression are marked with red crosses. The different columns correspond to varying values of the specified ID error \(\epsilon_\text{ID}\), with the effective errors in \(\Delta(t)\) on the same scale, as indicated. Interestingly, for larger errors, ID selects poles primarily in the lower spectrum, i.e., those with smaller \(\Im(\omega)\). As the error decreases, more poles are selected from higher regions of the spectrum. The inset indicates the total number of poles selected. As expected, the highest number of poles is selected at very low temperatures (last row), where many poles around \(\Re(\omega) = 0\) with small imaginary parts are determined to be relevant, highlighting the inherent complexity of low-temperature physics.

    \begin{figure}

     \includegraphics[width=\linewidth]{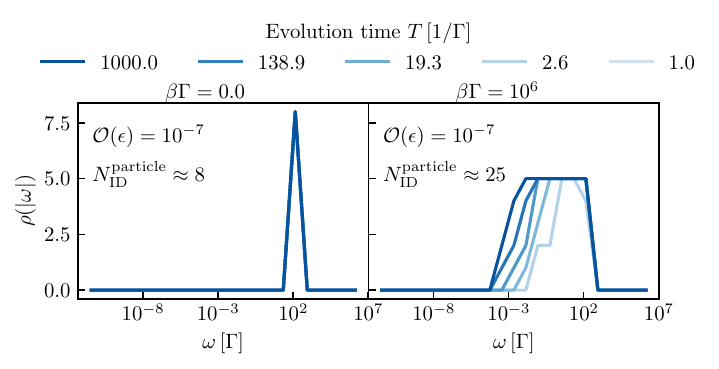}
     \caption{For a {\it flat band}, $\Gamma_\text{flat}^\Lambda(\omega),$ with widths $\Lambda/\Gamma =  50$ and sharpness $\nu = 20/\Lambda$: Density of frequency points, $\rho(|\omega|),$ for the particle component, before and after compression. For long evolution times and small temperatures, the ID grid shows exponential frequency spacing over several orders of magnitude, similarly to the construction from analytic modes, see Fig.~\ref{fig:grid_flatband}a. Columns correspond to different inverse temperatures $\beta$ as indicated.}
     \label{fig:freq_dens_AAA_flat}
\end{figure}

\subsubsection{Frequency Density} 

We conclude this Subsection by examining the density of frequencies after ID compression. In Fig.~\ref{fig:freq_dens_AAA_flat}, we present \(\rho(|\omega|)\) for \(\beta = 0\) (left) and \(\beta = 10^6/\Gamma\) (right), with different color shades indicating varying evolution times. At infinite temperature (\(\beta = 0\)), all frequencies cluster around \(\mathcal{O}(|\omega|) = 10^2\). However, for \(\beta = 10^6/\Gamma\), we observe a more intriguing behavior: as the evolution time \(T\) increases, the frequency spacing becomes exponential within a window whose lower cutoff shifts to smaller frequencies with increasing $T$. This pattern is consistent with the ID grid obtained from the analytic mode construction, as shown in Fig.~\ref{fig:grid_flatband}a. The fact that this exponential spacing emerges independently in both the analytic mode construction and the AAA-derived modes suggests that this grid choice is indeed optimal. Moreover, it underscores the effectiveness of ID in compressing modes, regardless of their origin. For a study of how the number of modes scales with error, we refer the reader to the next section.

\section{Complexity of General Baths}
\label{sec:complexity_general_baths}
\subsection{Overview}

At this stage, we have developed two complementary approaches to construct pseudomodes: i) the analytic ansatz from Sec.~\ref{sec:compression}, and ii) the AAA algorithm from Sec.~\ref{sec:AAA}, both of which are combined with ID to compress the resulting pseudomode sets. To simplify the discussion in this Section, we will refer to these methods as the ``analytic approach'' and ``AAA,'' with the understanding that all results presented refer to modes after compression.

In this Section, we examine a variety of spectral densities and compare the scaling of the number of modes \(N_\text{ID}\) obtained using both approaches. A key finding is that, after ID compression, all cases are compatible with the numerical scaling found in Sec.~\ref{sec:compression}, Eq.~(\ref{eq:numerical_scaling}). 
%\begin{equation}\label{eq:numerical_scaling_restatement}
%    N_\text{ID} \sim \log(T)\log(1/\epsilon).
%\end{equation}
Absolute values of \(N_\text{ID}\) are in excellent agreement in many cases, suggesting that the resulting mode sets are close to optimal. Furthermore, we identify cases where one approach is preferable over the other, or where only one approach is applicable.

In light of the numerical scaling in Eq.~(\ref{eq:numerical_scaling}), we introduce the following \textit{rescaled number of modes} which we use as proxy for the complexity:
\begin{equation}\label{eq:complexity}
\mathcal{C} \sim\frac{N_\text{ID}}{\log(T)\log(1/\epsilon)} .
\end{equation} We evaluate this quantity as a function of error $\epsilon$ and evolution time $T$ for all cases considered in this Section and present all results on the same scale, such that the corresponding values can be compared easily across the cases.

For reference, we list the four spectral densities that we consider in this Section:
\begin{enumerate}
\item  A flat band, introduced in Eq.~(\ref{eq:flat_band}), and considered in the previous sections.
    \item A linear spectral density with exponential cutoff of frequency $\omega =\pm \Lambda:$
    \begin{equation}\label{eq:linear_spec_dens}
        \Gamma_\text{lin}^\Lambda(\omega) = |\omega| e^{-|\omega|/\Lambda}.
    \end{equation}
    \item A superposition of Gaussians, representing a generic gapped bath:
    \begin{equation}\label{eq:gaussian_spec_dens}
\Gamma_\text{gauss}^\nu(\omega) = \Gamma \sum_{\omega_0} \exp\left(-\frac{|\omega - \omega_0|^2}{\nu}\right).
 \end{equation} For the results in this Section, we consider three Gaussian peaks, located at $\omega_0/\Gamma \in \{-4, 0, 4\}.$
 \item A semicircular spectral density with smooth edges:
 \begin{equation}
     \label{eq:spec_dens_circle}
    \Gamma_\text{circle}^\Lambda(\omega) = \Gamma_{\chi}^\Lambda(\omega) \times \Gamma_\text{flat}^\Lambda(\omega),
 \end{equation} 
 where we defined a regularized semicircle:
\begin{equation}
\Gamma_{\chi}^\Lambda(\omega) = 
    \begin{cases}
        \max(\sqrt{\Lambda^2 - \omega^2}, \chi \Lambda) & \text{if } |\omega| < \Lambda, \\
        \chi \Lambda& \text{else}.
    \end{cases}
\end{equation}
\end{enumerate}
Here $\chi$ is a regularization parameter, removing the sharp cusp with infinite first derivative at $\omega=\Lambda$. The exact semicircle shape is restored for $\chi=0$.

We note that all results presented in the remainder of this Article apply to both the particle and hole components. For modes obtained via analytic construction, both components are described using the same set of frequencies. By including both components in a single kernel matrix, we ensure that the ID algorithm selects a frequency set capable of accurately representing both components. The total number of modes corresponds to the ID rank of this kernel matrix. In Sec.~\ref{sec:combining_particle_hole}, we show how a particle and hole component corresponding to the same complex frequency \(\omega_k\) can be combined into a single mode.

For pseudomodes obtained via AAA, the particle and hole components are parametrized by separate sets of frequencies, and the modes are compressed individually for each component. The total number of modes in this case is the sum of the ID ranks for both components. In this way, the absolute values of $\mathcal{C}$ are indicative of the resources required to fully approximate a continuous bath. 

\subsection{Numerical Results}

\subsubsection{Flat Band} 
\label{sec:flat_band_complexity}

\begin{figure}[h] % 'h' stands for 'here', try to place the figure in this position
    \centering % Center the figure
    \includegraphics[width=\linewidth]{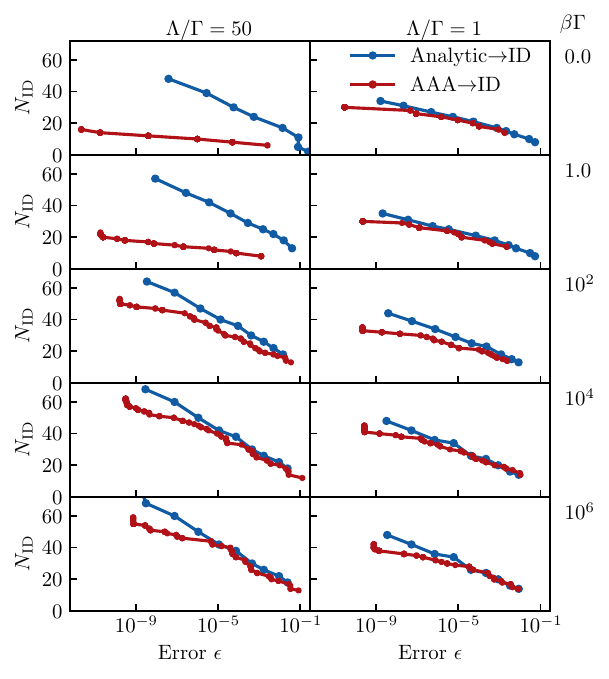} % Adjust the width as needed
    \caption{For a {\it flat band}, $\Gamma_\text{flat}^\Lambda(\omega),$ with widths $\Lambda/\Gamma =  50$ and $\Lambda/\Gamma =  1$ (left and right column, respectively), and sharpness $\nu = 20/\Lambda$: Scaling of compressed number of modes $N_\text{ID}$ with the error $\epsilon$ for an evolution time $T = 100/\Gamma$ after ID-compression. Input to the compression is a set of modes obtained by the analytic construction (blue) and the AAA algorithm (red), respectively. Rows correspond to different inverse temperatures $\beta$ as indicated.} % Add a caption
    \label{fig:error_scaling_wideband} % Add a label for referencing
    \end{figure}

    In Fig.~\ref{fig:error_scaling_wideband}, we present the scaling of \(N_\text{ID}\) with the error \(\epsilon\) for a flat band, Eq.~(\ref{eq:flat_band}). The Figure shows results for various inverse temperatures \(\beta\) (rows) and two different widths, \(\Lambda = 50\Gamma\) and \(\Lambda = 1\Gamma\) (columns). While the scaling is at most \(N_\text{ID} \sim \log(T)\log(1/\epsilon)\) in all cases, Fig.~\ref{fig:error_scaling_wideband} highlights key differences between the approaches: i) The AAA algorithm tailors the pole structures for each \(\beta\), offering a more efficient representation when low-temperature poles are unnecessary, as discussed in Sec.~\ref{sec:AAA_Approximation_of_integrand} and shown in Fig.~\ref{fig:AAA_flatband}b. ii) The analytic approach, with its temperature-independent construction, fixes a contour in the complex plane that may be suboptimal at high temperatures. While AAA smoothly connects the pole structure in the upper complex plane with a few poles (Fig.~\ref{fig:AAA_flatband}b, top row), the analytic contour always connects near \(\omega = 0\), close to the real axis. As seen in the left column of Fig.~\ref{fig:explicit_grid} corresponding to $\Lambda = 50\Gamma,$ if the input set of modes is chosen suboptimally, ID compression cannot effectively compress the frequency set for high temperatures. As a result, AAA offers favorable results at high temperatures. For a very narrow band, $\Lambda = \Gamma$ (see right column), the question of connecting the cutoff-poles in the complex plane concerns only a small frequency window such that the difference between the approaches becomes negligible here.
    
    Interestingly, at low temperatures where the analytic construction is ideally suited, both the slope and absolute numbers of $N_\text{ID}$ in Fig.~\ref{fig:error_scaling_wideband} are in good agreement for both methods and for both values of $\Lambda$.

    \begin{figure}[h] 
    \centering
    \includegraphics[width=\linewidth]{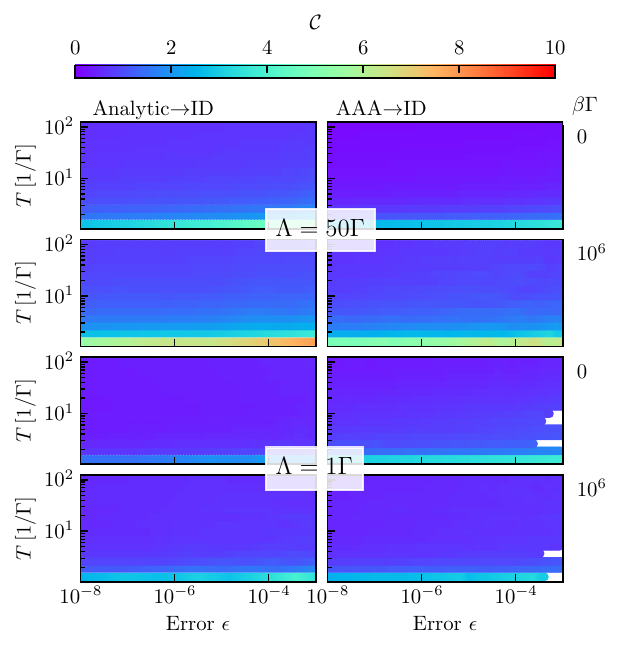} % Adjust the width as needed
    \caption{For a {\it flat band}, $\Gamma_\text{flat}^\Lambda(\omega),$ with widths $\Lambda/\Gamma =  50$ and $\Lambda/\Gamma =  1$ (first two and last two rows, respectively), and sharpness $\nu = 20/\Lambda$: Rescaled number of compressed modes, $\mathcal{C},$ as measure for the numerical complexity. Input to the compression is a set of modes obtained by the analytic construction (blue) and the AAA algorithm (red), respectively. Results are shown as a function of evolution time $T$ and error $\epsilon$ for inverse temperatures $\beta = 0$ and $\beta \Gamma = 10^6$ as indicated. } % Add a caption
    \label{fig:complexity_flat} % Add a label for referencing
    \end{figure}
    
In Fig.~\ref{fig:complexity_flat}, we further examine the complexity $\mathcal{C}$ as defined in Eq.~(\ref{eq:complexity}). The color tone encodes the value of $\mathcal{C}$ which we compare for two different temperatures ($\beta = 0$ and $\beta = 10^6/\Gamma$) and both approaches. In all cases we observe that $\mathcal{C}$ takes the highest values for very small $T$ caused by an initially steep rise $N_\text{ID}$ as already observed in Fig.~\ref{fig:scaling_ID_wideband}. At larger times, $\mathcal{C}$ quickly converges to a small value of the order $0.1.$ The discussed difference between the analytic approach and AAA at high temperatures and large bandwidth (see first row, left and right plot, respectively) is reflected by a slightly smaller value of $\mathcal{C}$ for AAA, indicated by slightly darker color. 

Given the conceptual differences of both approaches, the excellent agreement of both methods in scaling and absolute values of $N_\text{ID}$ is a strong indication that the determined sets of modes are close to optimal for the simulation parameters considered.

\subsubsection{Linear Spectral Density}

\begin{figure}[h]
    \centering
  
        \includegraphics[width=\linewidth]{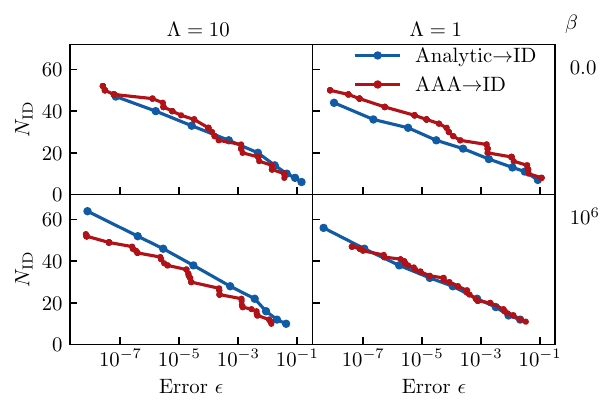}
    \caption{For a {\it linear spectral density}, $\Gamma_\text{lin}^\Lambda(\omega),$ with $\Lambda/\Gamma =  10$ and $\Lambda/\Gamma = 1$ (left and right column, respectively): Scaling of number of modes $N_\text{ID}$ with error $\epsilon$ for an evolution time $T = 100/\Gamma$. ID-compression was applied to a set of modes obtained by the analytic construction (blue) and the AAA algorithm (red), respectively.  Rows correspond to different inverse temperatures $\beta$ as indicated.}
    \label{fig:error_scaling_linear}
\end{figure}

Next, we study the linear spectral density, defined in Eq.~(\ref{eq:linear_spec_dens}). As before, we show the dependence of $N_\text{ID}$ on the error in Fig.~\ref{fig:error_scaling_linear}, here for two different cutoff frequencies, $\Lambda = 10$ and $\Lambda = 1$. We note that, for this particular spectral density, all quantities are given in arbitrary units.

We note that the sharp cusp of $\Gamma_\text{lin}(\omega)$ at $\omega = 0,$  resulting in a large number of AAA poles with $\Re(p_i)\approx 0,$ prevents a high temperatures advantage of AAA with respect to the analytic approach as observed in Fig.~\ref{fig:error_scaling_wideband}. Consequently, both AAA and the analytic approach exhibit similar behaviour for high and low temperatures. The good performance of both methods in this case is linked to the analyticity of $f(z)$ away from the imaginary axis (where it has Matsubara poles): For the analytic approach, this allows to freely rotate the integration contour without crossing any poles, and for AAA it allows to represent the spectral content of the kernel function, $f(\omega),$ with a moderate number of poles away from the imaginary axis.

As a result, similarly to before, both methods show excellent agreement in both scaling and absolute values of $N_\text{ID}.$

 \begin{figure}[h] % 'h' stands for 'here', try to place the figure in this position
    \centering % Center the figure
    \includegraphics[width=\linewidth]{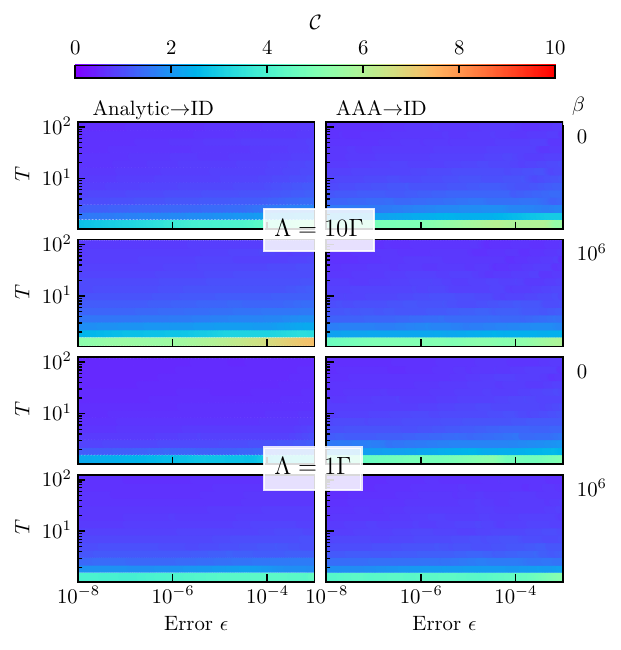} % Adjust the width as needed
    \caption{For a {\it linear spectral density}, $\Gamma_\text{lin}^\Lambda(\omega),$ with $\Lambda =  10$ and $\Lambda = 1$ (first two and last two rows, respectively): Rescaled number of compressed modes, $\mathcal{C},$ as measure for the numerical complexity. Input to the compression is a set of modes obtained with i) an exponential frequency parametrization (left) and ii) the AAA algorithm (right). Results are shown as a function of evolution time $T$ and error $\epsilon$ for inverse temperatures $\beta = 0$ and $\beta = 10^6$ as indicated. All dimensionsful quantities are given in arbitrary units.} % Add a caption
    \label{fig:complexity_linear} % Add a label for referencing
    \end{figure}

Evaluating the complexity $\mathcal{C},$ shown in Fig.~\ref{fig:complexity_linear}, reveals a similar dependence as before: While complexity is high for very small values $T,$ the value $\mathcal{C}$ quickly converges for larger $T$. All parameter combinations studied yield approximately the same value $\mathcal{C}\approx 0.1,$ which is in excellent agreement with the values obtained in Sec.~\ref{sec:flat_band_complexity}.
As before, the agreement between the methods in terms of absolute values $N_\text{ID},$ scaling, and complexity $\mathcal{C}$ is a strong indication for the optimality of the obtained sets of modes.

Next, we study two pathological examples where one of the two approaches outperforms the other or only one approach is applicable, demonstrating the power of our hybrid approach. 

\subsubsection{Gaussian Gapped Spectral Density}

\begin{figure}[h]
    \centering
   \begin{minipage}[b]{.48\textwidth}
        \begin{overpic}[width=\linewidth]{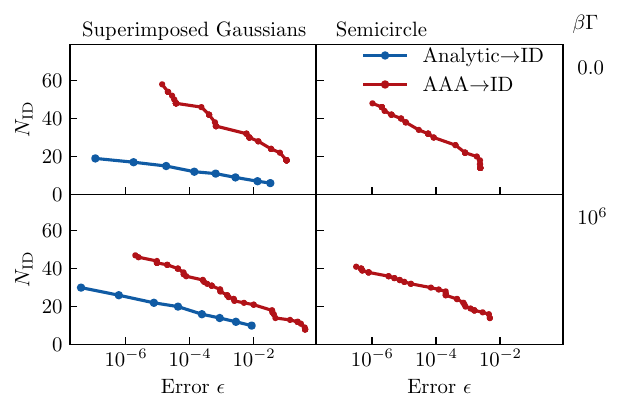}
    \put(5,62){\text{a)}}
    \end{overpic}
    \end{minipage}
    \begin{minipage}[b]{.48\textwidth}
       \begin{overpic}[width=\linewidth]{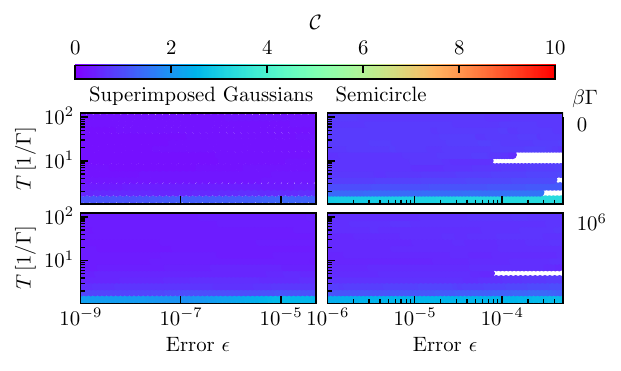}
        \put(5,53){\text{b)}}
    \end{overpic}
    \end{minipage}
    \caption{For a {\it spectral density of superimposed Gaussians}, $\Gamma_\text{gauss}^\nu(\omega),$ with $\nu =  0.05$ and a {\it semicircular spectral density}, $\Gamma_\text{circle}(\omega)$ with $\Lambda = 1$, $\chi=1/2$ (left and right columns, respectively). Results are shown for inverse temperature $\beta = 0$ and $\beta \Gamma = 10^6$. a) Scaling of compressed number of modes $N_\text{ID}$ with the error $\epsilon$ for an evolution time $T = 100/\Gamma$ after ID-compression. Input to the compression is a set of modes obtained by the analytic construction (blue) and the AAA algorithm (red), respectively.  Rows correspond to different inverse temperatures $\beta$ as indicated. b) Rescaled number of compressed modes, $\mathcal{C},$ as measure for the numerical complexity. Results are shown as a function of evolution time $T$ and error $\epsilon.$ The compression results for the i) superimposed Gaussian spectral density refer to the exponential frequency parametrization, while the ii) semicircular spectral density results refer to the AAA construction, both after compression with ID.}
    \label{fig:Gaussians_and_semcircle}
\end{figure}

First, we consider a gapped spectral density composed of three Gaussians, Eq.~(\ref{eq:gaussian_spec_dens}). In Fig.~\ref{fig:Gaussians_and_semcircle}a in the left column, we show the scaling of $N_\text{ID}$ with the error for both approaches. Interestingly, the analytic approach outperforms AAA for both, low and high temperatures. This is due to the inherent difficulty of representing a Gaussian function as rational function $r(z).$ We emphasize that the fact that the bath is gapped does not play a role here since---if a single Gaussian can be represented efficiently---also their sum has a simple representation.

In the analytic approach, however, we can exploit the fact that $\Gamma_\text{gauss}(\omega)$ is analytic in the complex plane, such that again, we can freely rotate the contour without crossing any poles. We thus find that the analytic approach allows for a very compact representation, superior to the AAA representation.

Importantly, the absolute values of $N_\text{ID}$ from AAA results are comparable to those obtained from the previous cases. Rather than viewing this spectral density as particularly challenging, it should be recognized as especially well-suited for the analytic approach, allowing for lower $N_\text{ID}$ values than in most other scenarios.

In Fig.~\ref{fig:Gaussians_and_semcircle}b, we show the complexity, only for the analytic approach, which yield a converged value $\mathcal{C} \lesssim 0.1$ that is slightly smaller, but comparable to the one from the previous cases.

\subsubsection{Semicircular Spectral Density}

Lastly, we consider a semicircular spectral density. As an exact semicircle is non-analytic and cannot be represented by rational functions, we consider a regularized semicircle Eq.~(\ref{eq:spec_dens_circle}), with the regularization parameter $\chi=1/2$.
This spectral density is not suited for the analytic approach as it does not posses a transparent pole structure which would be needed to rotate the integration contour in a controlled way. Instead, upon determining a suitable frequency grid as described in Sec.~\ref{sec:AAA_algorithm}, AAA is straight forwardly applicable and yields the scaling and complexity presented in Figs.~\ref{fig:Gaussians_and_semcircle}\,a--b (right column), again compatible with the $\log(T)\log(1/\epsilon)$ scaling. Altogether, both the absolute values of $N_\text{ID}$ and the complexity are consistent with the numerical results from all other cases.

\section{Pseudomode embedding}
\label{sec:dynamics}

Unraveling the temporal correlations encoded in the hybridization function $\Delta(\tau, \tau^\prime)$ into explicit pseudomodes enables an effective Markovian, i.e., time-local, description of the joint dynamics between the impurity and the pseudomodes. In this Section, we show how to map each mode parameter tuple $(\Gamma_k, \omega_k)$ onto the parameters of a time-evolution prescription of Lindblad type. The main result of this Section is an explicit set of Liouvillians that describes the evolution of the joint system of impurity and pseudomodes, effectively consolidating the particle and hole contributions for a given $\omega_k$ into a single pseudomode.

\subsection{Keldysh Structure}
In this Section, as throughout the paper, we carefully distinguish between the terms:
\begin{enumerate}
    \item ``Hybridization function,'' \begin{equation}
    \Delta(\tau, \tau^\prime) \approx \sum_k \Delta_k (\tau, \tau^\prime),
    \end{equation} which denotes a matrix-valued function defined on the Keldysh contour with time arguments $\tau,\tau^\prime \in \mathcal{C},$ as introduced in Eq.~(\ref{eq:expect_value}).
    
    For later convenience, we now make the Keldysh structure explicit in the notation:
    \begin{equation}
    \Delta_k(\tau,\tau^\prime) \to \Delta_{k,\alpha\beta}(t,t^\prime).\end{equation} 
    Here, $\alpha \in \{+,-\}$ for the forward and backward branch, respectively, and $t,t^\prime \in\mathbb{R}^+.$
    \item ``Kernel function,'' \begin{equation}
        \Delta(t)\approx \sum_k \Delta_k(t), \end{equation} 
        which describes an ordinary function of a single time variable $t\in \mathbb{R}^+$ that has been introduced in Eqs.~(\ref{eq:particle_integral}--\ref{eq:hole_integral}) and has been the main object studied up to this point.
\end{enumerate} Both quantities are easily distinguishable at any point, as the \textit{hybridization function} always has two time arguments, while the \textit{kernel function} always has one time argument. The dependence of the hybridization function on the kernel functions must be consistently derived from Eq.~(\ref{eq:expect_value}). We summarize the resulting relations in Eqs.~(\ref{eq:Keldysh_particle_relation1}--\ref{eq:Keldysh_hole_relation4}).

\subsection{Pseudo-Lindblad description}
\label{sec:Linblad_Dynamics}

Before deriving Liouvillians that describe the joint dynamics of impurity and pseudomodes, we establish the relationship between the hybridization function $\Delta(\tau,\tau^\prime),$ the kernel function $\Delta(t)$ and time evolution by a standard Lindblad equation. In particular, we determine the parameters of the latter such that the two-point functions that it generates on the Keldysh contour relate to the kernel function $\Delta(t)$ precisely as demanded by Eqs.~(\ref{eq:Keldysh_particle_relation1}--\ref{eq:Keldysh_hole_relation4}).

\subsubsection{Complex Couplings Through Two-Site Jump Operators} 
For illustration, we focus on the particle component first. For the single-mode Lindblad equation of the form,
\begin{equation}\label{eq:Lindblad_equation}
     \dot{\rho} = -i[H,\rho] + \gamma_k \big(2 L_k \rho L_k^\dagger - \{L_k^\dagger L_k, \rho\} \big),
\end{equation} we define jump operators:
\begin{equation}\label{eq:nonlocal_jumpoperator}
    L_k = d + \frac{\mu_k}{\gamma_k} c_k,
\end{equation} as well as an auxiliary Hamiltonian describing the pseudomode dynamics:
\begin{equation}\label{eq:H_aux}
    H_\text{aux} = \epsilon_k c_k^\dagger c_k + t_k (d^\dagger c_k + c_k^\dagger d),\text{ with } t_k,\epsilon_k \in \mathbb{R},
\end{equation}
such that $H = H_\text{imp} + H_\text{aux}$ represents the full Hamiltonian. From Eq.~(\ref{eq:Lindblad_equation}), it is straight forward to derive the amplitude of two-point functions on the backward branch:
\begin{equation}\label{eq:Lindblad_twopoint_particle}
 \Delta^p_{k,--}(t_0, t_0+t)  = (it_k -\mu_k)^2 e^{i(\epsilon_k +i\gamma_k) t}.
\end{equation} 

Interestingly, due to the two-site jump operator, the prefactor in Eq.~(\ref{eq:Lindblad_twopoint_particle}) becomes complex, which is necessary for describing a pseudomode $(\Gamma_k^p, \omega_k)$ with $\Gamma_k^p \in \mathbb{C}$. However, this construction also introduces an additional term in the Lindblad equation, Eq.~(\ref{eq:Lindblad_equation}), given by:
\begin{equation}\label{eq:unwanted_dissipation}
    \frac{\mu_k^2}{\gamma_k} \left( 2 d \rho d^\dagger - \{ d^\dagger d, \rho \} \right).
\end{equation}
This term represents local dissipation on the impurity, which is undesirable. Consequently, we remove it from the Lindblad equation manually, noting that this results in unphysical time evolution for an individual pseudomode.

\subsubsection{Determining the Lindblad parameters} The parameters in Eq.~(\ref{eq:Lindblad_twopoint_particle}) can be tuned to reproduce the known result for the two-point function on the backward branch (cf. Eq.~(\ref{eq:Keldysh_particle_relation2})),
\begin{equation}\label{eq:two_point_backward}
   \Delta^p_{k,--}(t_0, t_0+t)  =  -\Delta_k^p(t) = -\Gamma_k^p e^{i\omega_k t},
\end{equation} where $\Gamma_k^p ,\omega_k \in\mathbb{C}.$
To lighten the notation, we introduce the variable $\lambda_k^p = \sqrt{\Gamma_k^p}.$ Then,  by matching the prefactors and exponents of Eq.~(\ref{eq:Lindblad_twopoint_particle}) and Eq.~(\ref{eq:two_point_backward}), we obtain the following relationship between the Lindblad parameters $(t_k, \mu_k, \epsilon_k, \gamma_k)$ and the pseudomode parameters $(\lambda_k^p, \omega_k)$:
\begin{gather}
t_k = \Re(\lambda_k^p),\\
\mu_k = \Im(\lambda_k^p),\\
\epsilon_k = \Re(\omega_k^p),\\
\gamma_k =\Im(\omega_k^p).
\end{gather}
From this, it is obvious that for $\lambda_k^p \in \mathbb{R},$ the jump operators become local and the issue of removing local a local dissipation term does not arise.

\subsection{Deriving Liouvillians as Effective Amplitudes}
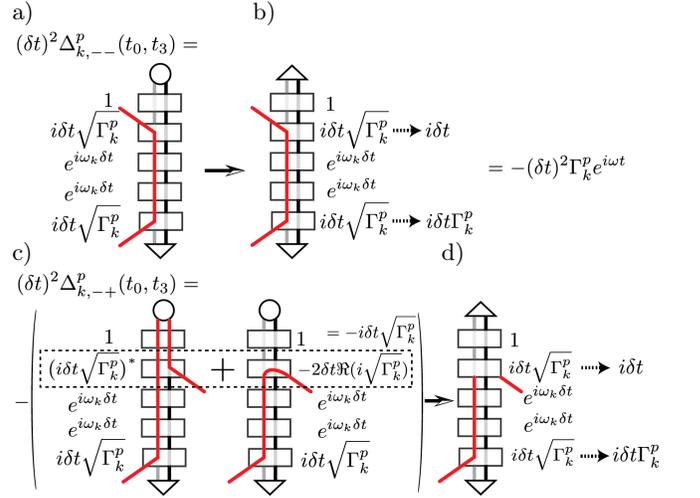
\begin{figure}[h] % 'h' stands for 'here', try to place the figure in this position
    \centering % Center the figure
    \begin{overpic}[width=\linewidth]{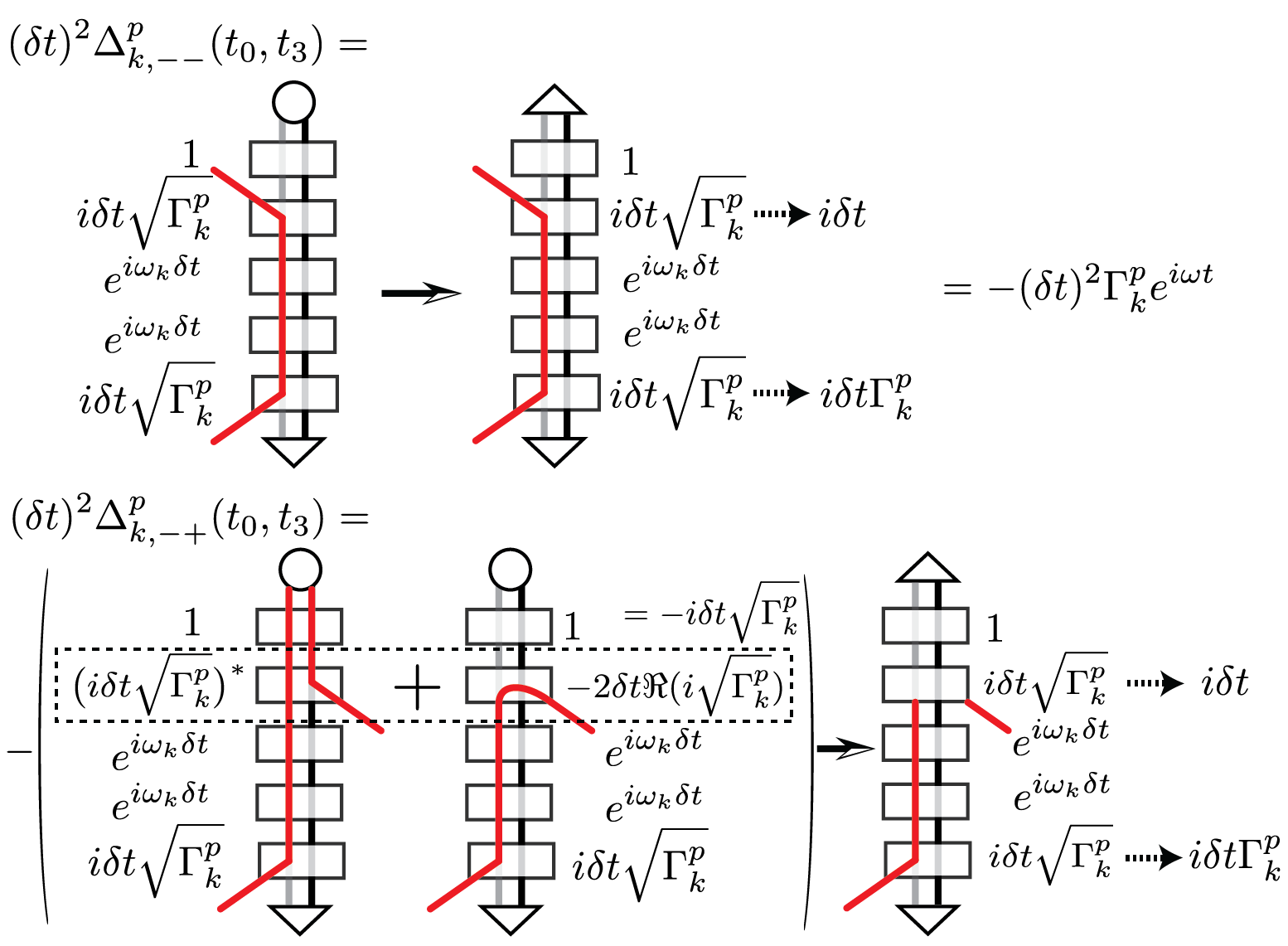} 
    \put(0,74){\text{a)}}
     \put(0,37){\text{c)}}
      \put(37,74){\text{b)}}
     \put(66,37){\text{d)}}
    \end{overpic}
    \caption{Diagrams for the particle component on the Keldysh contour. Fig. a) shows the two-point function on the backward branch, $(\delta t)^2\Delta_{k,--}^p(t,t^\prime),$ which is determined by the kernel function $\Delta_k^p(t) = \Gamma_k^p e^{i\omega_k t}$. The contribution from each time step is determined by the Lindblad equation, Eq.~(\ref{eq:Lindblad_equation}). b) The two-point function $(\delta t)^2\Delta_{-+},$ connecting the forward and backward branch, is a sum of two processes. The resulting value is the same as in a), consistent with the Keldysh relations, Eqs.~(\ref{eq:Keldysh_particle_relation1}--\ref{eq:Keldysh_hole_relation4}). Figures c) and d) show the effective amplitudes obtained by performing the trace at the final time, which is thus replaced by a projection onto the vacuum. The dashed arrows indicate how the effective amplitudes change under the gauge transformation described in the main text. } 
    \label{fig:diagrams_particle} % Add a label for referencing
\end{figure}

At this point, having determined the parameters of the Lindblad equation, we can extract the amplitudes of any time-local process. In Fig.~\ref{fig:diagrams_particle}a, we illustrate how the amplitudes accumulated at each individual time step contribute to the overall two-point function in Eq.~(\ref{eq:two_point_backward}). As this diagram is only nonzero when the initial and final state are the vacuum, we replace the condition at the final time by a projection onto the vacuum as shown in Fig.~\ref{fig:diagrams_particle}b.

Furthermore, we are now in a position to compute the \textit{effective} amplitudes for two-point functions to which multiple diagrams contribute. The determination of these effective amplitudes is the central step to deriving a Liouvillian that captures the full dynamics. While we adopt a diagrammatic approach here, a more formal treatment can be found in Eq.~(\ref{eq:Clavelli-shapiro}).

To reveal the strategy, consider the two-point function \( \Delta_{-+}^p \) shown in Fig.~\ref{fig:diagrams_particle}c, which connects the forward to the backward branch. This two-point function consists of two diagrams that share identical amplitudes at all time steps except one: in the first diagram, a complex hopping amplitude from the impurity to the pseudomode on the forward branch is present, while in the second diagram, the two-site jump operator, defined by Eq.~(\ref{eq:nonlocal_jumpoperator}), induces a decay process. Performing the trace at the final time $T$ corresponds to summing up these two amplitudes to an effective amplitude,
\begin{equation}\label{eq:effective_amplitude}
    (i\delta t \lambda_k^p)^* - 2\delta t \Re(i\lambda_k^p) = -i\delta t \lambda_k^p,
\end{equation} 
and replacing the boundary condition by a projection onto the vacuum state. Cancelling the sign in Eq.~(\ref{eq:effective_amplitude}) with the global sign in front of the diagram in Fig.~\ref{fig:diagrams_particle}c, we thus obtain the \textit{effective} amplitudes shown in Fig.~\ref{fig:diagrams_particle}d. 
Importantly, the values of the two-point functions from Figs.~\ref{fig:diagrams_particle}a--d are equivalent, as required by the Keldysh relation in Eq.~(\ref{eq:Keldysh_particle_relation3}).

By systematically evaluating the effective amplitudes for all Keldysh components, we obtain a comprehensive set of tensors, interpreted as a Liouvillian that determines the joint evolution of impurity and pseudomode. Formally, the evolution is described by an equation of the form (see Eq.~(\ref{eq:dissipator})):
\begin{equation}\label{eq:dissipator_result}
    \frac{d}{dt}\ket{\rho} =\mathcal{L}\ket{\rho},
\end{equation} 
together with the convention that the pseudomode degrees of freedom at the final time $T$ are projected onto the vacuum state,
\begin{equation}
\ket{\rho_{\text{imp}}(T)}=\, _{k}\langle 0\ket{\rho(T)}.
\end{equation}
The derivation for holes is analogous and yields another Liouvillian, as further outlined in App.~\ref{app:Lindblad_construction}.

In summary, by moving the trace to the latest time step where the trajectory on the forward and backward branch of the pseudomode differ, we obtain a set of \textit{effective} amplitudes that serve as an effective Liouvillian defining the evolution of the joint impurity-pseudomode system. 

\subsection{Combining Particle and Hole Liouvillians}
\label{sec:combining_particle_hole}

When a pseudomode from the particle component and one from the hole component share the same complex frequency $\omega_k$, their contributions to the hybridization function can be effectively combined into a single pseudomode. In this subsection, we demonstrate how the previously derived sets of Liouvillians for particles and holes can be merged. As before, we present a simple diagrammatic derivation here and refer to App.~\ref{app:dissipators_equations} for a more formal but equivalent derivation. Crucially, both methods consistently lead to the same evolution equation, Eq.~(\ref{eq:dissipator_result}).

The central idea hinges on the observation that the amplitudes can be grouped into three distinct classes, as shown in Fig.~\ref{fig:all_tensors}: (I) those that contribute solely to the particle component, (II) those that contribute solely to the hole component, and (III) those that appear in both components. Leveraging this structure, we can apply a gauge transformation separately to the set of particle amplitudes and the set of hole amplitudes, shifting all physical information into tensors that are unique to each component.

In Figs.~\ref{fig:diagrams_particle}, \ref{fig:diagrams_hole}, we illustrate with dashed arrows how this gauge transformation impacts the effective amplitudes derived earlier. Importantly, this transformation preserves the values of all two-point functions. As a result, we achieve a situation in which all tensors appearing in both the particle and hole components are independent of the pseudomode parameters and identical across both sets. Consequently, the sets of amplitudes can be directly merged, effectively unifying the descriptions of the particle and hole components into a single pseudomode. Any two-point function on the Keldysh contour can then be expressed as a combination of these tensors, and conversely, every combination of tensors with $2n$ external legs, and only continuous internal lines corresponds to a valid $2n$-point function. The final set of amplitudes is diagramatically summarized in Fig.~\ref{fig:all_tensors} and explicitly stated in Eqs.~(\ref{eq:particle+hole_dissipator}).

\section{Summary and Discussion}

In this paper, we studied efficient bath representations of quantum impurity models in terms of pseudomodes. The influence of the bath on the impurity is fully encoded by the hybridization function, whose Keldysh components can be expressed as ordinary Fourier integrals. Our approach rests on the approximation of the latter in terms of decaying exponentials that correspond to pseudomodes, see Eq.~(\ref{eq:Fourier_integral_intro}). The combination of all pseudomodes represents an auxiliary bath that approximates the influence that the original continuous bath exerts on the impurity.

First, we proved a rigorous error bound on the pseudomode number (\ref{eq:Nbath}) for spectral densities that decay at least exponentially for $|\omega|\to \infty.$ The proof used an exponential frequency parametrization in the complex plane to discretize the Fourier integral in Eq.~(\ref{eq:Fourier_integral_intro}), thus providing an explicit pseudomode construction. We demonstrated that these pseudomodes could be further compressed using ID, which effectively selects a compact subset of modes and renormalizes their couplings. Notably, we found that this procedure parametrically improves the scaling of pseudomode number, see Eq.(\ref{eq:Nbath_improved}). 
Further, to generalize our numerical approach to arbitrary spectral densities, e.g.  when they are defined numerically for a set of points, we employed the AAA algorithm to construct a rational approximation $r(\omega)$ of the spectral content of the Fourier integral, $f(\omega)$. This approach, combined with ID compression, while providing a different way to construct pseudomodes, yielded analogous scaling (\ref{eq:Nbath_improved}) for various spectral densities. 
%The rational approximation $r(\omega)$ is characterized by its poles and residues in the complex plane. 
%Substituting $r(\omega)$ for $f(\omega)$ on the left-hand side of Eq.~(\ref{eq:Fourier_integral_summary}) and closing the integration contour in the upper half-plane, we derive an approximation analogous to the right-hand side of Eq.~(\ref{eq:Fourier_integral_summary}). As before, we use the AAA modes as a starting point for further compression via ID.

%We examine various spectral densities using both the analytic approach and the AAA algorithm, followed by ID compression, consistently observing the numerical scaling stated in Eq.~(\ref{eq:compressed_scaling_summary}) across all cases. Furthermore, we show that the number of pseudomodes $N_\text{ID}$ determined by both methods quantitatively agrees in several instances where both approaches are equally applicable.
Overall, the two approaches are complementary: the AAA method offers greater flexibility, particularly with numerically sampled function values, while the analytic method requires explicit knowledge of the pole structure of \(f(\omega)\). However, the AAA method's effectiveness depends on \(f(\omega)\) being amenable to rational function approximation. We showed that Gaussian functions, with their rapidly decaying tails, present challenges for the AAA method. In contrast, the analytic approach excels for Gaussians, as they are analytic throughout the complex plane, allowing contour rotation without crossing poles. 
%Thus, the choice of the approach depends on the specific nature of the problem.

Finally, from the approximation of the kernel functions, we derive an explicit Liouvillian \(\mathcal{L}\) that governs the time evolution of the joint impurity-pseudomode system, expressed as
\begin{equation}\label{eq:Lindblad_summary}
    \frac{d}{dt}\ket{\rho} = \mathcal{L}\ket{\rho}.
\end{equation}
Importantly, while the time evolution of each individual pseudomode mode is unphysical due to their complex couplings $\Gamma_k\in \mathbb{C},$ their combination approximates the influence of the exact, physical bath. Since the only approximation in this work is the finite-sum representation in Eq.~(\ref{eq:Fourier_integral_intro}), the error bound on observables derived in Sec.~\ref{sec:Bound_on_the_observables} directly applies to the Liouvillian formulation. This completes the formal framework for constructing optimal bath representations, establishing a rigorous upper bound on the simulation complexity of non-equilibrium QIM. 

Turning towards applications, we expect that specific QIM problems may exhibit even lower complexity compared to the general bound. Combining the optimal pseudomode formulation with additional tensor-network compression of the full state is a promising approach to describing the long-time behavior of QIMs, which will be explored in future work. This direction aligns with recent advances in representing the Feynman-Vernon influence function as an effective temporal matrix product state~\cite{ThoennissPRB2022,ThoennissPRB23,NgPRB23,ParkPRB24}, but may have lower computational cost. 

%Another promising application is the real-time evolution of an impurity density matrix coupled to pseudomodes. This approach aligns with recent advances in representing the Feynman-Vernon influence function as an effective temporal matrix product state~\cite{ThoennissPRB2022,ThoennissPRB23,NgPRB23,ParkPRB24}. In this context, successive application of ID would allow for the construction of a bath with dynamically varying size and time-dependent couplings, conceptually akin to the method in Ref.~\cite{ParkPRB24}. By representing the pseudomode bath as a spatial state, the system could be further compressed using standard tensor network techniques, making it well-suited for tensor-based time evolution methods.

Looking ahead, several other applications of the method can be envisioned. One intriguing direction is the evaluation of diagrammatic many-body relations, such as Schwinger-Dyson equations, for computing non-equilibrium Green's functions. For equilibrium systems, a similar approach has been recently developed in ~\cite{KayeDiscrete,kiese24Discrete}. In this context, our method could significantly reduce computational resource demands by replacing nested time integrals with sums over a moderate number of modes.

Lastly, the challenge of approximating Fourier integrals by finite sums of decaying exponentials extends beyond quantum impurity problems, with promising applications in fields that depend on efficient Fourier representations. This approach could potentially enhance spectral methods for partial and stochastic differential equations and streamline the computation of bosonic Green's functions, where accurate treatment of frequency-dependent phenomena is essential for understanding many-body interactions and collective excitations.

{\it Note added---} During final stages of preparing this manuscript for submission, a related paper~\cite{ParkPseudomodesArxiv2024} appeared, where a different algorithm was used to numerically construct a set of pseudomodes. The conclusions of this work appear to be generally consistent with the numerical analysis part of our paper.

\begin{acknowledgements}
We thank Jason Kaye, Benedikt Kloss, and Olivier Parcollet for insightful discussions during the development of this work. Furthermore, we thank Gunhee Park for interesting conversations on related topics, and Alessio Lerose and Michael Sonner for numerous stimulating interactions during this and related projects.
This work was supported by the European Research Council (ERC) under the European Union's Horizon 2020 research and innovation program (grant agreement No. 864597).

\end{acknowledgements}

\appendix
\section{Derivation of Liouvillians}\label{app:Lindblad_construction}

 \subsection{Keldysh Structure}
 The dependence of the hybridization function $\Delta(\tau,\tau^\prime)$---defined on the Keldysh contour---on the kernel function must be consistently derived from Eq.~(\ref{eq:expect_value}).
 % The components of $\Delta(\tau,\tau^\prime)$ are defined as \begin{equation}
     %\Delta_{\alpha\beta}(\tau,\tau^\prime) = \frac{\partial^2}{\partial \eta_{\tau^\prime}\partial\bar{\eta}_{\tau}}\exp\bigg(\int_\mathcal{C}d\tau \int_\mathcal{C} d\tau^\prime \bar{\eta}_{\tau} \Delta(\tau,\tau^\prime) \eta_{\tau^\prime} \bigg).
% \end{equation}
%Applying the second-order derivatives to the full path integral before integrating out the bath modes relates $\Delta(\tau,\tau^\prime)$ to the bath correlation functions, which are, in turn, proportional to the kernel functions. 
For the Hamiltonian of the form (\ref{eq:Ham_full}--\ref{eq:H_bath}), we define two operators $F,F^\dagger$ as:
\begin{equation}
F=\sum\limits_k t_k c_k \,,\quad F^\dagger=\sum\limits_k t_k c^\dagger_k.
\end{equation}
The hybridization function might be calculated as a bath two-point correlation function of the $F,F^\dagger$ operators:
\begin{equation}
\Delta_{\tau,\tau^\prime}=\text{tr}_{\text{bath}}\left(\mathcal{P}F(\tau)F^\dagger(\tau^\prime)\rho_{\text{bath}}(0)\right),
\end{equation}
where $\mathcal{P}$ means the ordering along the Keldysh contour. Omitting the details of the calculation here, we state the resulting relations which we use in Sec.~\ref{sec:dynamics}.
Using the decomposition into a particle and hole component from Eq.~(\ref{eq:hybridization_function}), we find for the particle component:
\begin{align}
  \Delta_{++}^p(t,t^\prime) =& -\Theta(t-t^\prime)\big(\Delta^p(t-t^\prime)\big)^*,\label{eq:Keldysh_particle_relation1}\\
  \Delta_{--}^p(t,t^\prime) =& -\Theta(t^\prime-t)\Delta^p(t^\prime-t) ,\label{eq:Keldysh_particle_relation2}\\
  \Delta_{-+}^p(t,t^\prime) =& \Delta_{++}^p(t,t^\prime) + \Delta_{--}^p(t,t^\prime),\label{eq:Keldysh_particle_relation3}\\
  \Delta_{+-}^p(t,t^\prime) =&   \,0,\label{eq:Keldysh_particle_relation4}
  \end{align}
and for the hole component:

\begin{align}
  \Delta_{++}^h(t,t^\prime) =& \,\Theta(t^\prime-t)\Delta^h(t-t^\prime),\label{eq:Keldysh_hole_relation1}\\
  \Delta_{--}^h(t,t^\prime) =& \,\Theta(t-t^\prime)\big(\Delta^h(t^\prime-t) \big)^*,\label{eq:Keldysh_hole_relation2}\\
  \Delta_{+-}^h(t,t^\prime) =& \,\Delta_{++}^h(t,t^\prime) + \Delta_{--}^h(t,t^\prime),\label{eq:Keldysh_hole_relation3}\\
  \Delta_{-+}^h(t,t^\prime) =&  \, 0.\label{eq:Keldysh_hole_relation4}
  \end{align}

\subsection{Details of Diagrammatic Derivation}
\label{app:dissipators_diagrams}

Here, we comment on the derivation of Liouvillian for the hole component which is largely analogous to the particle case. The Lindblad equation in this case is given by \begin{equation}
    \dot{\rho} = -i[H,\rho] + \gamma_k \big(2 L^\dagger_k \rho L_k - \{L_k L_k^\dagger, \rho\} \big).
\end{equation}  Again, computing two-point functions, this time on the forward branch, yields
\begin{equation}
    \Delta_{++}^h(t_0, t_0+t) =- (-it_k-\mu_k)^2e^{i(\epsilon_k+i\gamma_k)t}.
\end{equation} Matching the prefactor and exponent with Eq.~(\ref{eq:Keldysh_hole_relation1}), analogously to the particle case,
yields the defining equations for the evolution parameters:
\begin{align}\label{eq:Lindblad_params1_hole}
\epsilon_k + i\gamma_k &= \omega_k,\\
   -it_k - \mu_k  &= i\sqrt{\Gamma_k^h}.\label{eq:Lindblad_params2_hole}
    \end{align} These fully determine the Lindblad dynamics. 
    
\begin{figure}[h] % 'h' stands for 'here', try to place the figure in this position
    \centering % Center the figure
    \begin{overpic}[width=\linewidth]{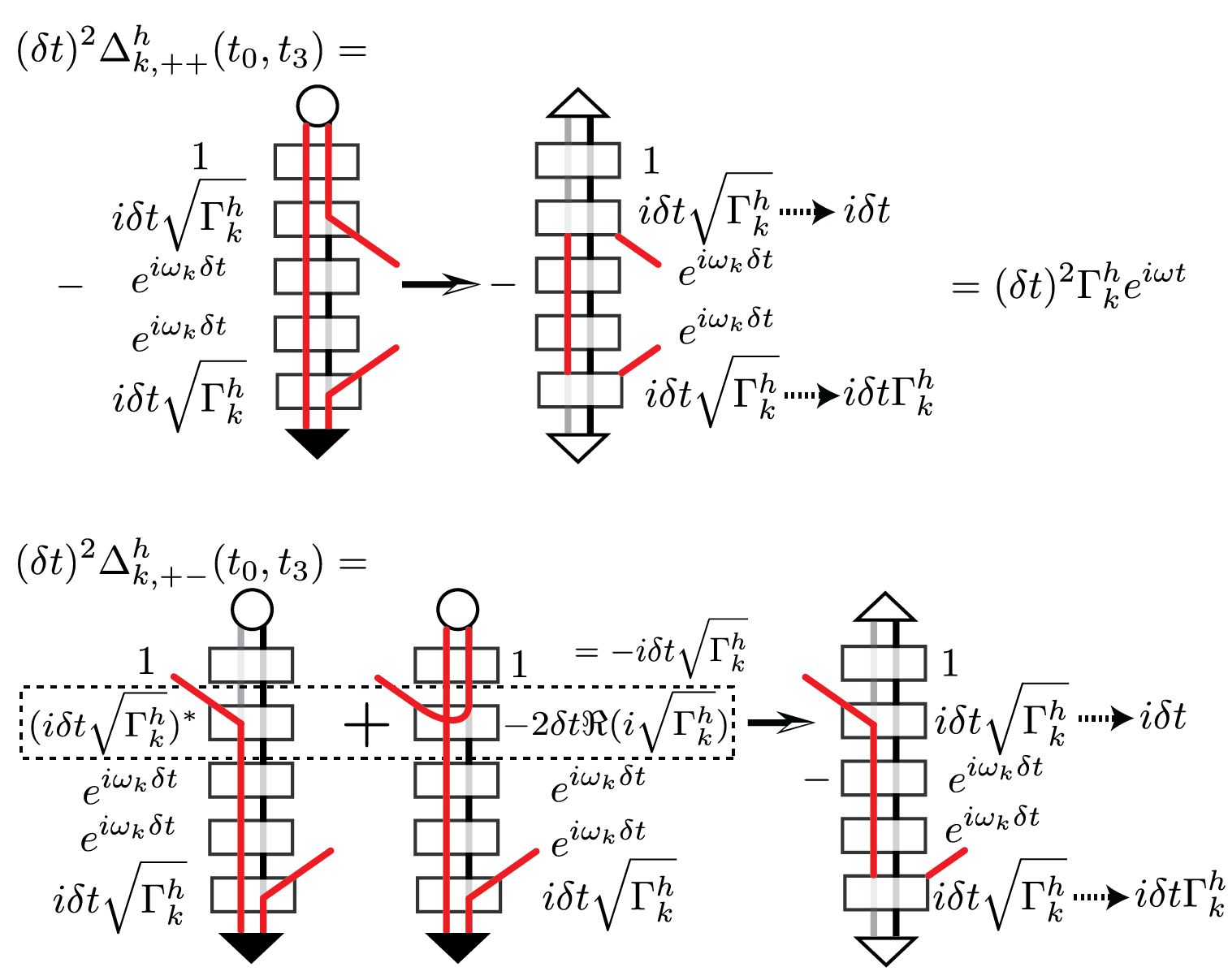} 
     \put(0,81){\text{a)}}
     \put(0,39){\text{c)}}
      \put(37,81){\text{b)}}
     \put(66,39){\text{d)}}
    \end{overpic}
    \caption{Diagrams for the hole component on the Keldysh contour as defined in the lower right corner. a) The figure shows the two-point function on the backward branch, $(\delta t)^2\Delta_{++}^h(t,t^\prime),$ which is determined by the kernel function $\Delta_k^h(t) = \Gamma_k^h e^{i\omega_k t}$.  b) The two-point function $\Delta_{k,+-}^h$ connecting the forward and backward branch, is the sum of two processes. The resulting value is the same as in a), consistent with the Keldysh relations, Eqs.~(\ref{eq:Keldysh_particle_relation1}--\ref{eq:Keldysh_hole_relation4}). Figures c) and d) show the effective amplitudes obtained by performing the trace at the final time, which is thus replaced by a projection onto the vacuum. The dashed arrows indicate how the effective amplitudes change under the gauge transformation described in the main text. . } 
    \label{fig:diagrams_hole} % Add a label for referencing
\end{figure}
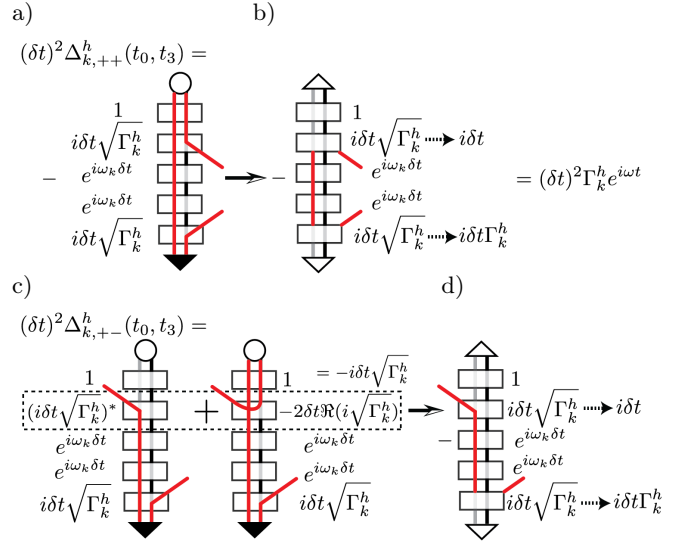
In Fig.~\ref{fig:diagrams_hole}, we repeat the same diagrammatic derivation employed in Fig.~\ref{fig:diagrams_particle} to derive effective amplitudes that are interpreted as a Liouvillian time evolution.

\begin{figure}[h]
    \centering
    \begin{overpic}[width=0.9\linewidth]{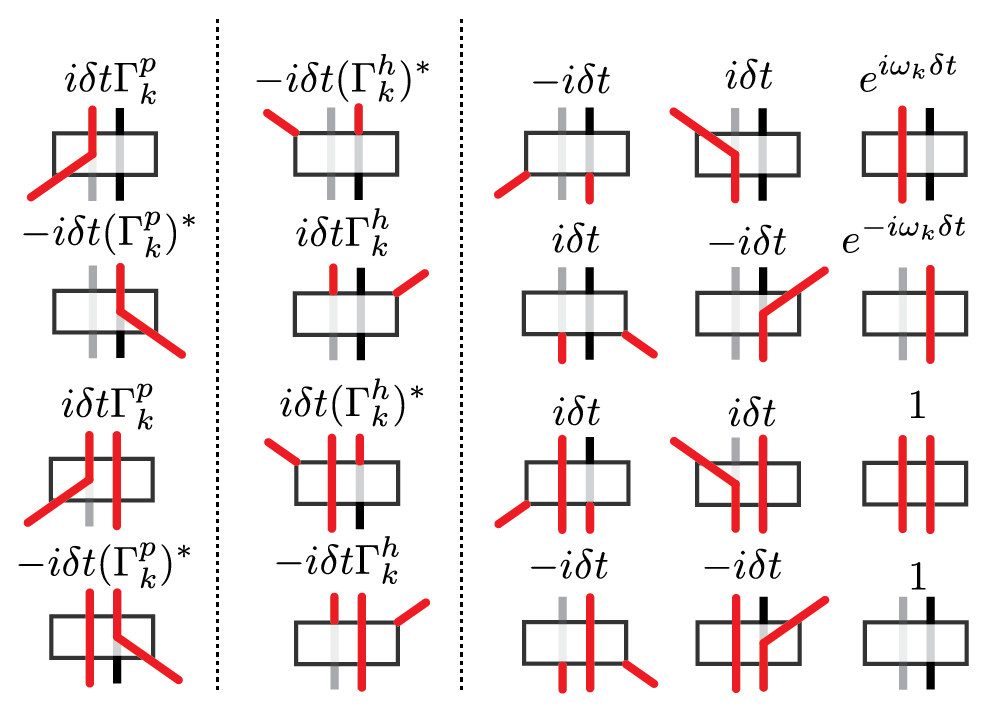}
    \put(3,68){\text{I Particle}}
    \put(26,68){\text{II Hole}}
     \put(50,68){\text{III Both}}
    \end{overpic}
    \caption{Overview of all tensors---derived as effective amplitudes---that describe the joint evolution of impurity and pseudomode. Tensors in class (I) encode amplitudes from the particle component, tensors in class (II) encode amplitudes from the hole component, and tensors in class (III) do not encode physical information. These tensors coincide with the ones derived in Eq.~(\ref{eq:particle+hole_dissipator}).}
    \label{fig:all_tensors}
\end{figure}

Finally, in Fig.~\ref{fig:all_tensors}, we summarize the full Liouvillian, obtained after merging the particle and hole components with the use of two gauge transformations that move all physical content of the Liouvillian to those that contribute solely to the particle and hole contributions, respectively. For details, we refer to the main text, Sec.~\ref{sec:dynamics}, as well as to App.~\ref{app:dissipators_equations} which contains a more formal derivation of these Liouvillians.

\subsection{Derivation of the effective Liouvillians using Operator-State Duality}
\label{app:dissipators_equations}
In this appendix, we complement the diagrammatic derivation with a formal analytic derivation of the particle pseudomode embedding for a given exponential spectral density defined by Eq.~\eqref{eq:two_point_backward} and the condition to explicitly obey the Keldysh relations, Eqs.~(\ref{eq:Keldysh_particle_relation1}--\ref{eq:Keldysh_particle_relation4}).

As in the main text, in order to lighten notation, we introduce the variable
    \begin{equation}\label{eq:lambda_particle}
        \lambda_k^p \equiv t_k + i \mu_k = \sqrt{\Gamma_k^p},
    \end{equation}
    such that the kernel function of the particle component can be expressed as,
    \begin{equation}
     -\Delta_{k,--}^p(t_0, t_0+t) =   \Delta_k^p(t) = (\lambda_k^p)^2 e^{i\omega_k t} = \Gamma_k^p e^{i\omega_k t} .
    \end{equation}
    For holes, we introduce:
    \begin{equation}\label{eq:lambda_hole}
        \lambda_k^h = - t_k +i\mu_k = \sqrt{\Gamma_k^h},
    \end{equation} resulting in,
    \begin{equation}
    \Delta_{++}^h(t_0, t_0 + t) = \Delta_k^h(t) =  (\lambda_k^h)^2 e^{i\omega_kt} = \Gamma_k^h e^{i\omega_kt}. 
    \end{equation}

\paragraph{Derivation for particle contribution}
Our starting point is the Eqs.~(\ref{eq:Lindblad_equation}--\ref{eq:H_aux}). We start by performing the so-called operator-state duality by introducing the two types of fermions $c_{k},\bar{c}_k$ as well as two copies of impurity fermions $d,\bar{d}$ corresponding to a forward and backward branch respectively; and map the density matrix into a state via the formula:
\begin{equation}
\rho \to |\rho\rangle=\rho \sum\limits_i|i\rangle|i\rangle \,
\end{equation}
where $|i\rangle$ forms a basis of states in the total Hilbert space. As a result, we obtain the evolution of a $|\rho\rangle$ state:
\begin{gather}\label{eq:dissipator}
\frac{d}{dt}|\rho\rangle=\tilde{\mathcal{L}} |\rho\rangle,
\end{gather}
with the Liouvillian $\tilde{\mathcal{L}}$ in superoperator form
\begin{equation}\label{eq:Limp_Lk_Lint}
\tilde{\mathcal{L}}=\mathcal{L}_{\text{imp}}+\tilde{\mathcal{L}}^p_{k}+\tilde{\mathcal{L}}^p_{int},
\end{equation}
where each of the three terms is explicitly given as:
\begin{gather}\label{eq:D_imp}
\mathcal{L}_{\text{imp}}=-i (H_{\text{imp}}-\bar{H}_{\text{imp}})-\frac{\mu_k^2}{\gamma_k}(d^\dagger d + \bar{d}^\dagger\bar{d})+\\ \nonumber + 2\frac{\mu_k^2}{\gamma_k}d\bar{d}\\
\tilde{\mathcal{L}}^p_k=(-i\epsilon_k-\gamma_k)c^{\dagger}_kc_k+(i\epsilon_k-\gamma_k)\bar{c}^{\dagger}_k\bar{c}_k+2\gamma c_k\bar{c}_k\\ \label{eq:D_int}
\tilde{\mathcal{L}}^p_{\text{int}}=-i (\lambda_k^p)^{\star} (c^\dagger_k d+d^\dagger c_k)+i\lambda_k^p(\bar{c}^\dagger_k \bar{d}+\bar{d}^\dagger \bar{c}_k)+\\ \nonumber +2\mu_k c_k \bar{d}+2\mu_k d\bar{c}_k,
\end{gather}
where $\lambda_k^p$ is defined in Eq.~(\ref{eq:lambda_particle}). We note that the individual terms in Eq.~\eqref{eq:Limp_Lk_Lint} are not physical Liouvillians as only their combination provides an approximation to the physical evolution.
The Eqs.~(\ref{eq:D_imp}--\ref{eq:D_int}) directly follow from \eqref{eq:Lindblad_equation} with the conventions:
\begin{gather}
c_k\rho \to c_k|\rho\rangle \,,\quad c_k^\dagger\rho \to c_k^\dagger|\rho\rangle, \\
\rho c_k \to \bar{c}_k^\dagger|\rho\rangle \,,\quad \rho d_k^\dagger \to \bar{d}_k|\rho\rangle,\\
d_k\rho \to d_k|\rho\rangle \,,\quad d_k^\dagger\rho \to d_k^\dagger|\rho\rangle, \\
\rho d_k \to \bar{d}_k^\dagger|\rho\rangle \,,\quad \rho d_k^\dagger \to \bar{d}_k|\rho\rangle.
\end{gather}
Our prescription is to evolve the initial density matrix 
\begin{equation}
\rho(0)=\rho_{\text{imp}}\otimes |0\rangle \langle 0|
\end{equation}
according to Eq.~\eqref{eq:physical-Dissipator} and then take trace over the auxiliary fermions. The trace of a density matrix is equal to the expectation value in the state picture:
\begin{equation}\label{eq:trace_k}
\text{tr}_k \rho(T)=\,  _k\langle 0,0| e^{-c_k\bar{c}_k}|\rho(T)\rangle.
\end{equation}
To perform a trick from Fig.~\ref{fig:diagrams_particle} and relate the trace computation to vacuum averaging, we use a known idea from \cite{CLAVELLI1973490}. We first note that the state $|\rho(T)\rangle$ may be considered as an initial state acted by an evolution operator $\mathcal{E}_T=e^{\tilde{\mathcal{L}}T}$.
\begin{equation}
|\rho(T)\rangle=\mathcal{E}_T|\rho_{\text{imp}}\rangle|0,0\rangle_k.
\end{equation}
Then, we can use the commutation relation:
\begin{equation}\label{eq:Clavelli-shapiro}
e^{-c_k\bar{c}_k}\mathcal{E}_T(c_k,\bar{c}_k|c_k^\dagger,\bar{c}_k^{\dagger})=\mathcal{E}_T(c_k,\bar{c}_k|c_k^\dagger-\bar{c}_k,\bar{c}_k^{\dagger}+c_k)e^{-c_k\bar{c}_k},
\end{equation}
and move the $e^{-c_k\bar{c}_k}$ operator to the right until it hits the vacuum. As a result, $c_k^\dagger,\bar{c}_k^\dagger$ operators will be transformed as:
\begin{gather}
c_k^\dagger\to c_k^\dagger-\bar{c_k},\\
\bar{c}_k^\dagger\to \bar{c}_k^\dagger+c_k,
\end{gather}
and we get new effective Liouvillians:
\begin{gather}\label{eq:physical-Dissipator}
\mathcal{L}^p_k=(-i\epsilon_k-\gamma_k)c^{\dagger}_kc_k+(i\epsilon_k-\gamma_k)\bar{c}^{\dagger}_k\bar{c}_k\\
\mathcal{L}^p_{\text{int}}=-i (\lambda_k^p)^{\star} (c^\dagger_k d+d^\dagger c_k)+i \lambda_k^p(\bar{c}^\dagger_k \bar{d}+\bar{d}^\dagger \bar{c}_k)+\\ \nonumber +i(\lambda_k^p)^{\star} c_k \bar{d}+i\lambda_k^p d\bar{c}_k.
\end{gather}
Note that after this transformation the impurity density matrix at moment $T$ is given by an overlap with vacuum instead of a trace in Eq.~\eqref{eq:trace_k}:
\begin{equation}
|\rho_{\text{imp}}(T)\rangle=\, _k\langle 0| \rho(T)\rangle.
\end{equation}
Here, we introduced the shorthand notation $_k\!\bra{0}$ for the vacuum state in the pseudomode space $_k \!\bra{0}\overset{\text{def}}{=}\,_k \!\bra{0,0}$.
After a gauge transformation
\begin{gather}
c_k\to c_k/(\lambda_k^p)^\star\,,\quad c_k^\dagger\to (\lambda_k^p)^{\star} c_k^\dagger\,,\\
\bar{c}_k\to \bar{c}_k/\lambda_k^p\,,\quad \bar{c}_k^\dagger\to \lambda_k^p \,\bar{c}_k^\dagger\,,
\end{gather}
we reproduce the following operators:
\begin{gather}\label{eq:D_k}
\mathcal{L}^p_k=(-i\epsilon_k-\gamma_k)c^{\dagger}_kc_k+(i\epsilon_k-\gamma_k)\bar{c}^{\dagger}_k\bar{c}_k,\\
\mathcal{L}^p_{\text{int}}=-i ((\Gamma_k^p)^\star c^\dagger_k d+d^\dagger c_k)+i (\Gamma_k\bar{c}^\dagger_k \bar{d}+\bar{d}^\dagger \bar{c}_k)+\\ \nonumber +i c_k \bar{d}+i d\bar{c}_k,
\end{gather}
where we switched to the $\Gamma_k^p$ notations defined in Eq.~\eqref{eq:lambda_particle}. These Liouvillians exactly reproduce the diagrammatics from Fig.~\ref{fig:diagrams_particle}.

The derivation for holes is analogical. Remarkably, both particle and hole contributions could be united into a single mode with the interaction part of the Liouvillian:
\begin{gather}\label{eq:particle+hole_dissipator}
\mathcal{L}_{\text{int}}=-i ((\Gamma^p_k)^\star c^\dagger_k+\bar{c}_k) d-id^\dagger(c_k-\bar{c}^\dagger_k\Gamma_k^h) + \\ \nonumber +i(\Gamma_k^p\bar{c}^\dagger_k+c_k) \bar{d}+i\bar{d}^\dagger(\bar{c}_k-(\Gamma_k^h)^\star c_k^\dagger)\,,
\end{gather}
which exactly reproduce the tensors summarized in Fig.~\ref{fig:all_tensors}.

Putting everything together, we obtain a physical evolution:
\begin{equation}\label{eq:final_dissiparors}
\frac{d}{dt}|\rho\rangle=(\mathcal{L}_{\text{imp}}+\mathcal{L}_{\text{int}}+\mathcal{L}_k) |\rho\rangle,
\end{equation}
where $\mathcal{L}_{\text{int}}$ is given in Eq.~\eqref{eq:particle+hole_dissipator}, $\mathcal{L}_{k}=\mathcal{L}^p_{k}$ from Eq.~\eqref{eq:D_k} and $\mathcal{L}_{\text{imp}}$ is given by Eq.~\eqref{eq:D_imp}. As stated at the beginning of this section, the impurity evolution contains an extra dissipation term with decay rate $\mu_k^2/\gamma_k$. In practice, we drop this term, thereby giving up physicality of the time evolution. Formally, this corresponds to substituting:
\begin{equation}
\mathcal{L}_{\text{imp}}=-i(H_{\text{imp}}-\bar{H}_{\text{imp}}).
\end{equation}

\section{Proof of the Theorem}
\label{sec:proof}
Here we provide a detailed proof of the theorem in Section \ref{sec:analytic_bound}, which gives an analytical decomposition of the kernel function into a finite number of exponentials, along with the accompanying error bounds.

We start with an overview of our strategy, which consists of the following steps: 

\noindent (i) Rotating the integration contour by an angle $r_{\rm max}$ (black rotated contour in Fig.~\ref{fig:rotated_contour}), which generates a contribution from the residues of the poles of $\Gamma(\omega)$ situated in sectors I and II (see Fig.~\ref{fig:rotated_contour}). Importantly, this contribution is a sum of exponentials, $e^{i\Omega_k t}$, where $\Omega_k$ are the poles; 

\noindent (ii) The integral over the rotated contour is discretized into an infinite sum of exponentials, with a tunable discretization parameter $h$, which contributes to the overall approximation error;

\noindent (iii) The difference between the integral and the sum is represented using a modification of known result (see McNamee Stenger and Whitney,  Ref.~\cite{McNamee}, theorem 5.2), yielding contributions from poles in regions II and III, and also an error that is controlled by the choice of angle $r$, see below.

\noindent (iv) Finally, the sum of the exponentials arising when the integral over the rotated contour is discretized, is truncated, and the error of this truncation, controlled by parameter $h$, is bounded.

To facilitate the analysis, we split the proof into two parts. First, we prove an auxiliary theorem for steps (i)-(iii) above, and, second, a separate lemma proving the truncation error bound, step (iv). We start with the auxiliary theorem.

\begin{theorem}\label{exp_th}
Consider a spectral density $\Gamma(\omega)$ that is meromorphic in the upper half-place. Fix an angle $0<r_{\text{max}}<\frac{\pi}{4}$ and assume that $\Gamma(\omega)$ decays at least exponentially at infinity,
\begin{equation}\label{eq:high_freqs_cutoff}
|\Gamma(\omega)|< \Gamma e^{-\nu|\Re(\omega)|}\,, \quad \{|\omega| \gg \Lambda\,, \ 0<\arg(\omega)<2r_{\text{max}}\} \ ;
\end{equation}
further, asumme that $\Gamma(\omega)$ has a finite number of poles $\omega=\Omega_k$ in the sector $0<\arg(\omega)<2r_{\text{max}}$.

Then, for any $0<r<r_{\text{max}}$, the kernel function
\begin{equation}\label{g_integral}
\Delta^p_+(t)=\int\limits_{0}^{\infty}\Gamma(\omega)(1-n_{\text{F}}(\omega))e^{i\omega t} \frac{d\omega}{2\pi}
\end{equation}
can be approximated as an infinite number of exponentials,
\begin{equation}\label{eq:Dp_exponentials}  \Delta_+^p(t)=D_1(t)+D_2(t)+\tilde{\Delta}_+^p(t)+\delta_r(t), 
\end{equation}
where $D_{1,2}(t)$ and $\tilde{\Delta}_+^p(t)$ are sums of exponentials, and $\delta_r(t)$ is the error that depends on the choice of $r$, bound for which is provided below. The three sums of exponentials have a different origin; the first contribution from poles $\Omega_k$ in regions I and II arises when the integration contour is rotated, see Fig.~\ref{fig:rotated_contour}
\begin{equation}\label{eq:D1}
    D_1(t)=\sum\limits_{k=1}^{\kappa_2}R_k(1-n_{\text{F}}(\Omega_k))e^{i\Omega_k t},   
\end{equation}
where $R_k= i {\rm Res}_{\Omega_k}(\Gamma(\omega))$ are the residues of the poles. 
The second contribution, 
\begin{multline}\label{eq:D2_appendix}
    D_2(t)=\sum\limits_{k=\kappa_1+1}^{\kappa_2}R_k\frac{(1-n_{\text{F}}(\Omega_k))e^{-2i\pi x_k/h}}{1-e^{-2i\pi x_k/h}}e^{i\Omega_k t}-\\
-\sum\limits_{k=\kappa_2+1}^{\kappa_3}R_k\frac{\big(1-n_{\text{F}}(\Omega_k) \big)e^{2i\pi x_k/h}}{1-e^{2i\pi x_k/h}}e^{i\Omega_k t}, 
\end{multline}
involves poles $k=\kappa_1+1,...\kappa_2$ situated in region II, as well as poles $k=\kappa_2+1,...\kappa_3$ situated in region III. In the above equation, we introduced $x_k=\log\frac{\Omega_k}{\Gamma}$.  
  The third contribution stems from approximating the integral over the rotated contour by an infinite discrete sum. This sum, in turn, reads
\begin{equation}\label{eq:DeltaTilde}     
\tilde{\Delta}^p_+(t)=\frac{h}{2\pi}\sum\limits_{k=-\infty}^{\infty}(1-n_{\text{F}}(\omega_k))e^{i\omega t}\Gamma(\omega_k)\omega_k\,,\quad \omega_k=\Gamma e^{h k+ir}, 
\end{equation}
where $h$ is the discretization parameter and $\Lambda$ is a high-frequency cutoff.

Finally, the error $\delta_r(t)$ is bounded as:
\begin{equation}\label{A_bound}
|\delta_r(t)|<\frac{\mathcal{A}_r(t)e^{-2\pi r/h}}{1-e^{-2\pi r/h}} \,, 
\end{equation}
\begin{equation}
\label{eq:A_r}
\mathcal{A}_r(t) =2 \,{\max}_{v=r_{\text{max}}\pm r} \int\limits_{0}^{\infty}  |\Gamma(\omega e^{iv})(1-n_{\text{F}}(\omega e^{iv}))e^{i\omega e^{iv}t}|\frac{d\omega}{2\pi}
\end{equation}
\end{theorem}
\begin{proof}
First, we make a change of variables $\omega=\Gamma e^x$ and define
\begin{equation}\label{eq:ftx}
f(t,x)=\Gamma(\omega(x))(1-n_{\text{F}}(\omega(x))e^{i\omega(x) t} \Gamma e^{x},
\end{equation}
such that
\begin{equation}
\Delta^p_+(t)=\int\limits_{-\infty}^{\infty} f(t,x) \frac{dx}{2\pi}.
\end{equation}
Then we rotate the integration contour in \eqref{g_integral}, 
$x\to x+ir_{\text{max}}$ into the complex plane, as illustrated in Fig.~\ref{fig:rotated_contour}. By the assumption of the theorem, the spectral density may have poles $\Omega_k%=\Lambda(1+i\frac{k+1/2}{\kappa})
$, $k=1,...\kappa_2$ in regions I and II, which we have to take into account when rotating the contour: 
\begin{equation}\label{g_rotated}
\Delta^p_+(t)=\int\limits_{-\infty}^{\infty}f(t,x+i r_{\text{max}}) dx+\sum\limits_{k=1}^{\kappa_2}R_k(1-n_{\text{F}}(\Omega_k))e^{i\Omega_k t}.
\end{equation}

Next, we aim to approximate the integral over rotated contour by a discrete sum, \eqref{eq:DeltaTilde}. To that end, we first express the discrete sum via the contour integral with the cotangent kernel:
\begin{equation}
\frac{h}{2\pi}\sum\limits_{k-\infty}^{\infty}f(t,h k+ir_{\text{max}})=-\frac{i}{2}\oint\cot(\frac{\pi x}{h})f(t,x+ir_{\text{max}})\frac{dx}{2\pi}.
\end{equation}
It is convenient to rewrite $\frac{i}{2}\cot(\frac{\pi x}{h})=\frac{1}{2}+\frac{e^{2\pi i x/h}}{1-e^{2\pi i x/h}}$ on the upper segment of the contour, and $\frac{i}{2}\cot(\frac{\pi x}{h})=-\frac{1}{2}-\frac{e^{-2\pi i x/h}}{1-e^{-2\pi i x/h}}$ on the lower segment of the contour. The two $\frac{1}{2}$ contributions of a $\cot(\frac{\pi x}{h})$ combine to a continuous integral; thus, we obtain the following estimate for the difference between the discrete sum and the continuous integral:
\begin{equation}\label{discr-continuous}
\int\limits_{-\infty}^{\infty} f(t,x+ir_{\text{max}})\frac{dx}{2\pi}- \frac{h}{2\pi}\sum\limits_{k-\infty}^{\infty}f(t,h k+ir_{\text{max}})=\delta_0(t),
\end{equation}
\begin{multline}
\delta_0(t)=-\int\limits_{-\infty+i0}^{\infty+i0} \frac{e^{2\pi i x/h}}{1-e^{2\pi ix/h}}f(t,x+ir_{\text{max}})\frac{dx}{2\pi}-\\-\int\limits_{-\infty-i0}^{\infty-i0}\frac{e^{-2\pi i x/h}}{1-e^{-2\pi ix/h}}f(t,x+ir_{\text{max}})\frac{dx}{2\pi}.
\end{multline}
Now, for any positive $r<r_{\text{max}}$ we may deform the upper contour up, $x\to x+ir$, and the bottom contour down, $x\to x-ir$. The only obstructions to deform the contours are the poles of a spectral density $\Gamma(\omega)$, situated in region II for the lower contour, and in region III for the upper contour. Including the contribution from these poles, we get:
\begin{multline}\label{delta_0-delta_r}
\delta_0(t)-\delta_r(t)=\sum\limits_{k=\kappa_1+1}^{\kappa_2}R_k\frac{(1-n_{\text{F}}(\Omega_k))e^{-2\pi i x_k/h}}{1-e^{-2\pi i x_k/h}}e^{i\Omega_k t}-\\-\sum\limits_{k=\kappa_2+1}^{\kappa_3}R_k\frac{(1-n_{\text{F}}(\Omega_k))e^{2\pi i x_k/h}}{1-e^{2\pi i x_k/h}}e^{i\Omega_k t},
\end{multline}
where
\begin{multline}
\delta_r(t)=-\int\limits_{-\infty}^{\infty} \frac{e^{-2\pi (r-i x)/h}}{1-e^{-2\pi (r-ix)/h}}f(t,x+ir_{\text{max}}+ir)\frac{dx}{2\pi}-\\-\int\limits_{-\infty}^{\infty} \frac{e^{-2\pi (r+i x)/h}}{1-e^{-2\pi (r+ix)/h}}f(t,x+ir_{\text{max}}-ir)\frac{dx}{2\pi}.
\end{multline}
We can estimate the absolute value of the integrals in this equation by an integral of an absolute value, which implies the stated bound:
\begin{equation} \label{delta_r}
|\delta_r(t)|<\frac{\mathcal{A}_r(t)e^{-2\pi r/h}}{1-e^{-2\pi r/h}},
\end{equation}
\begin{equation}
\mathcal{A}_r(t) = 2\ \text{max}_{v=r_{\text{max}}\pm r} \int\limits_{-\infty}^{\infty}  |f(t,x+iv)|\frac{dx}{2\pi}
\end{equation}
Combining equations \eqref{g_rotated},\eqref{discr-continuous},\eqref{delta_0-delta_r},\eqref{delta_r} we finally obtain the main statement of the theorem \eqref{exp_th}.
\end{proof}

According to the discussion in Sec.~\ref{sec:Bound_on_the_observables}, our aim is to bound the $L_1$ norm on the discrepancy in our approximation. To derive such a bound using the above result, where the error is related to the function $\mathcal{A}_r(t)$ and the angle $r$, we now discuss the behavior of $\mathcal{A}_r(t)$, depending on the spectral density.
 The function $\mathcal{A}_r(t)$ is bounded for any $t$, and decays to zero in the limit $t\to\infty$. The large $t$ asymptotic of the $\mathcal{A}_r(t)$ is determined by the decaying exponent $e^{-\omega t \sin(v)}$ in $f(t,x+iv)$, see Eq.~(\ref{eq:ftx}). We note that to estimate the asymptotic behavior, we switch back to $\omega$ integration variable. Integrating by parts Eq.~\eqref{A_bound}, we obtain: 
 \begin{equation}
\mathcal{A}_r(t)=\text{max}_{v=r_{\text{max}}\pm ir}\frac{1}{2\pi t\sin(v)}\Gamma(0)+o(\frac{1}{t}) \,,\quad \text{for} \ t\to \infty. 
 \end{equation}
We see that for the spectral density non-vanishing at zero, the $L_1$ norm of $\mathcal{A}$ is bounded as 
\begin{equation}\label{eq:norm_A_bound}
\|\mathcal{A}_r(t)\|_{L_1}< C\log(T).
\end{equation}
If the spectral density does vanish at zero, the integral over $t$ converges, and one could get an even better estimate $\|\mathcal{A}\|_{L_1}<C$. In both cases,  we can safely use an upper-bound \eqref{eq:norm_A_bound}.

Now we are in a position to estimate the discretization step $h$, required to achieve the desired error $\varepsilon$ of the approximation of the kernel function (so far, with the untruncated infinite sum arising when the integral is discretized) as:
\begin{equation}
h=\frac{2\pi r}{ \log(C\log(T)\varepsilon^{-1})}\sim \frac{2\pi r}{ \log(\varepsilon^{-1})}.
\end{equation}

We have thus completed steps (i)-(iii) of the proof. We perform the last step (iv), estimate of the error arising when the infinite sum \eqref{eq:DeltaTilde} is truncated by proving the following Lemma:
\begin{lemma}
The infinite sum
\begin{gather}\label{discrete_sum}
\tilde{\Delta}^p_+(t)=\frac{h}{2\pi}\sum\limits_{k=-\infty}^{\infty}(1-n_{\text{F}}(\omega_k))e^{i\omega_k t}\Gamma(\omega_k)\omega_k\,,\\\omega_k=\Gamma e^{h k+ir_{\text{max}}},
\end{gather}
can be truncated to a finite number of terms $N_{\text{bath}}$, scaling as: $N_{\text{bath}}\sim h^{-1}\log( T\varepsilon^{-1})$.
\end{lemma}
\begin{proof}
We would like to truncate the sum such that it runs from $k=-M+1$ to $k=N-1$. Our goal is therefore to estimate the two sums that we are neglecting:
\begin{equation}
    \frac{h}{2\pi}\big|\sum\limits_{k=-\infty}^{-M} f(t,hk)\big|\leq\frac{h}{2\pi}\sum\limits_{k=-\infty}^{-M} |f(t,hk)|,
\end{equation}
\begin{equation}
    \frac{h}{2\pi}\big|\sum\limits_{k=N}^{\infty} f(t,hk)\big|\leq \frac{h}{2\pi}\sum\limits_{k=N}^{\infty} |f(t,hk)|.
\end{equation}
Let us first estimate the second sum:
\begin{equation}
\frac{h}{2\pi}\sum\limits_{k=N}^{\infty}e^{- \Re(\omega_k) t }|\Gamma(\omega_k)\big(1-n_{\text{F}}(\omega_k)\big)\omega_k|.
\end{equation} 
At large frequencies, we could estimate $|1-n_{\text{F}}(\omega_k)|<1$. For the rest of the summand, we may use the high-frequency decay of the spectral density \eqref{eq:high_freqs_cutoff}, and majorate the monotonically decreasing sum by the integral:
\begin{equation}
\frac{h}{2\pi}\sum\limits_{k=N}^{\infty}e^{-\Re(\omega_k) (t+\nu)}|\omega_k|<\int\limits_{x=hN}^{\infty}e^{-\Re(\omega(x)) (t+\nu)}|\omega(x)|\frac{dx}{2\pi},
\end{equation}
where $\omega(x)=\Gamma e^{x+ir_{\text{max}}}$. Direct computation of the integral provides:
\begin{equation}
\frac{h}{2\pi}\sum\limits_{k=N}^{\infty} |f(t,hk)| \leq \frac{1}{2\pi}\frac{\Gamma}{t+\nu}e^{-\Re(\omega_{\text{max}})\nu}.
\end{equation}
The condition that the $L_1$ norm of this function is smaller than $\varepsilon$ leads to the following estimate:
\begin{equation}
\Re(\omega_{\text{max}})=\nu^{-1}\log\left(\frac{\Gamma\varepsilon^{-1}}{2\pi}\log\frac{T+\nu}{\nu}\right)\simeq \nu^{-1}\log\big(\Gamma\varepsilon^{-1}\big).
\end{equation}
This provides for the cutoff:
\begin{equation}
\label{eq:M_cutoff}
N=\frac{1}{h}\log\Big(\frac{\Re(\omega_{\text{max}})}{\Gamma \cos(r_{\text{max}})}\Big)\simeq\frac{1}{h}\log\big(\log(\Gamma\varepsilon^{-1})\big).
\end{equation}
For the lower cutoff, one can neglect the frequency dependence of the spectral density and Fermi distribution:
\begin{multline}
    \frac{h}{2\pi}\sum\limits_{-\infty}^{-M}|e^{i \omega_k t} \Gamma(\omega_k\big(1-n_{\text{F}}(\omega_k)\big)\omega_k|\simeq\\ \simeq \sum\limits_{k=M}^{\infty}\frac{h}{4\pi}\Gamma\Gamma^\prime\cos(r_{\text{max}}) e^{-h k},
\end{multline}
where we introduced $\Gamma^\prime=\Gamma(0)$. Eventually, we arrive at a similar estimate for $M$:
\begin{equation} \label{x_min_exp}
M h =\log\frac{\Gamma\Gamma^\prime\cos(r_{\text{max}}) T}{4 \pi \varepsilon}.
\end{equation}
The total number of terms therefore scales as
\begin{equation}\label{eq:N+M}
N_{\text{bath}}=N+M-2\sim \frac{1}{2\pi r} \log(\varepsilon^{-1}) \log(T\varepsilon^{-1}).
\end{equation}
This can be improved for spectral densities that vanish at zero frequencies. If $\Gamma(\omega)\simeq \Gamma^{\prime{1-s}}\omega^{s}+o(\omega)$, $Mh$ scales as:
\begin{equation}
M h \simeq\frac{1}{s+1}\log\left(\frac{\Gamma^\prime\Gamma \cos(r_{\text{max}}) T}{4\pi \varepsilon}\frac{\Gamma^s}{\Gamma^{\prime s}}\right)\sim \frac{1}{s+1}\log\frac{T}{\varepsilon}.
\end{equation}
If $\Gamma(\omega)$ vanishes at zero faster than any power of $\omega$, the $Mh$ could be even finite, however, the most generic scaling is provided in Eq.~\eqref{eq:N+M}.
\end{proof}
This finishes the proof of theorem \eqref{th:finite_approximation}, the case of negative frequencies and the hole kernel function is analogical.

%\newpage
\bibliography{MyBib}

\begin{thebibliography}{51}%
\makeatletter
\providecommand \@ifxundefined [1]{%
 \@ifx{#1\undefined}
}%
\providecommand \@ifnum [1]{%
 \ifnum #1\expandafter \@firstoftwo
 \else \expandafter \@secondoftwo
 \fi
}%
\providecommand \@ifx [1]{%
 \ifx #1\expandafter \@firstoftwo
 \else \expandafter \@secondoftwo
 \fi
}%
\providecommand \natexlab [1]{#1}%
\providecommand \enquote  [1]{``#1''}%
\providecommand \bibnamefont  [1]{#1}%
\providecommand \bibfnamefont [1]{#1}%
\providecommand \citenamefont [1]{#1}%
\providecommand \href@noop [0]{\@secondoftwo}%
\providecommand \href [0]{\begingroup \@sanitize@url \@href}%
\providecommand \@href[1]{\@@startlink{#1}\@@href}%
\providecommand \@@href[1]{\endgroup#1\@@endlink}%
\providecommand \@sanitize@url [0]{\catcode `\\12\catcode `\$12\catcode
  `\&12\catcode `\#12\catcode `\^12\catcode `\_12\catcode `\%12\relax}%
\providecommand \@@startlink[1]{}%
\providecommand \@@endlink[0]{}%
\providecommand \url  [0]{\begingroup\@sanitize@url \@url }%
\providecommand \@url [1]{\endgroup\@href {#1}{\urlprefix }}%
\providecommand \urlprefix  [0]{URL }%
\providecommand \Eprint [0]{\href }%
\providecommand \doibase [0]{https://doi.org/}%
\providecommand \selectlanguage [0]{\@gobble}%
\providecommand \bibinfo  [0]{\@secondoftwo}%
\providecommand \bibfield  [0]{\@secondoftwo}%
\providecommand \translation [1]{[#1]}%
\providecommand \BibitemOpen [0]{}%
\providecommand \bibitemStop [0]{}%
\providecommand \bibitemNoStop [0]{.\EOS\space}%
\providecommand \EOS [0]{\spacefactor3000\relax}%
\providecommand \BibitemShut  [1]{\csname bibitem#1\endcsname}%
\let\auto@bib@innerbib\@empty
%</preamble>
\bibitem [{\citenamefont {Bulla}\ \emph {et~al.}(2008)\citenamefont {Bulla},
  \citenamefont {Costi},\ and\ \citenamefont {Pruschke}}]{BullaRMP2008}%
  \BibitemOpen
  \bibfield  {author} {\bibinfo {author} {\bibfnamefont {R.}~\bibnamefont
  {Bulla}}, \bibinfo {author} {\bibfnamefont {T.~A.}\ \bibnamefont {Costi}},\
  and\ \bibinfo {author} {\bibfnamefont {T.}~\bibnamefont {Pruschke}},\
  }\bibfield  {title} {\bibinfo {title} {Numerical renormalization group method
  for quantum impurity systems},\ }\href
  {https://doi.org/10.1103/RevModPhys.80.395} {\bibfield  {journal} {\bibinfo
  {journal} {Rev. Mod. Phys.}\ }\textbf {\bibinfo {volume} {80}},\ \bibinfo
  {pages} {395} (\bibinfo {year} {2008})}\BibitemShut {NoStop}%
\bibitem [{\citenamefont {Knap}\ \emph {et~al.}(2012)\citenamefont {Knap},
  \citenamefont {Shashi}, \citenamefont {Nishida}, \citenamefont {Imambekov},
  \citenamefont {Abanin},\ and\ \citenamefont {Demler}}]{Knap2012}%
  \BibitemOpen
  \bibfield  {author} {\bibinfo {author} {\bibfnamefont {M.}~\bibnamefont
  {Knap}}, \bibinfo {author} {\bibfnamefont {A.}~\bibnamefont {Shashi}},
  \bibinfo {author} {\bibfnamefont {Y.}~\bibnamefont {Nishida}}, \bibinfo
  {author} {\bibfnamefont {A.}~\bibnamefont {Imambekov}}, \bibinfo {author}
  {\bibfnamefont {D.~A.}\ \bibnamefont {Abanin}},\ and\ \bibinfo {author}
  {\bibfnamefont {E.}~\bibnamefont {Demler}},\ }\bibfield  {title} {\bibinfo
  {title} {Time-dependent impurity in ultracold fermions: Orthogonality
  catastrophe and beyond},\ }\href {https://doi.org/10.1103/PhysRevX.2.041020}
  {\bibfield  {journal} {\bibinfo  {journal} {Phys. Rev. X}\ }\textbf {\bibinfo
  {volume} {2}},\ \bibinfo {pages} {041020} (\bibinfo {year}
  {2012})}\BibitemShut {NoStop}%
\bibitem [{\citenamefont {Cetina}\ \emph {et~al.}(2016)\citenamefont {Cetina},
  \citenamefont {Jag}, \citenamefont {Lous}, \citenamefont {Fritsche},
  \citenamefont {Walraven}, \citenamefont {Grimm}, \citenamefont {Levinsen},
  \citenamefont {Parish}, \citenamefont {Schmidt}, \citenamefont {Knap},\ and\
  \citenamefont {Demler}}]{Grimm2016}%
  \BibitemOpen
  \bibfield  {author} {\bibinfo {author} {\bibfnamefont {M.}~\bibnamefont
  {Cetina}}, \bibinfo {author} {\bibfnamefont {M.}~\bibnamefont {Jag}},
  \bibinfo {author} {\bibfnamefont {R.~S.}\ \bibnamefont {Lous}}, \bibinfo
  {author} {\bibfnamefont {I.}~\bibnamefont {Fritsche}}, \bibinfo {author}
  {\bibfnamefont {J.~T.~M.}\ \bibnamefont {Walraven}}, \bibinfo {author}
  {\bibfnamefont {R.}~\bibnamefont {Grimm}}, \bibinfo {author} {\bibfnamefont
  {J.}~\bibnamefont {Levinsen}}, \bibinfo {author} {\bibfnamefont {M.~M.}\
  \bibnamefont {Parish}}, \bibinfo {author} {\bibfnamefont {R.}~\bibnamefont
  {Schmidt}}, \bibinfo {author} {\bibfnamefont {M.}~\bibnamefont {Knap}},\ and\
  \bibinfo {author} {\bibfnamefont {E.}~\bibnamefont {Demler}},\ }\bibfield
  {title} {\bibinfo {title} {Ultrafast many-body interferometry of impurities
  coupled to a fermi sea},\ }\href {https://doi.org/10.1126/science.aaf5134}
  {\bibfield  {journal} {\bibinfo  {journal} {Science}\ }\textbf {\bibinfo
  {volume} {354}},\ \bibinfo {pages} {96} (\bibinfo {year} {2016})},\ \Eprint
  {https://arxiv.org/abs/https://www.science.org/doi/pdf/10.1126/science.aaf5134}
  {https://www.science.org/doi/pdf/10.1126/science.aaf5134} \BibitemShut
  {NoStop}%
\bibitem [{\citenamefont {Pustilnik}\ and\ \citenamefont
  {Glazman}(2004)}]{Pustilnik_2004}%
  \BibitemOpen
  \bibfield  {author} {\bibinfo {author} {\bibfnamefont {M.}~\bibnamefont
  {Pustilnik}}\ and\ \bibinfo {author} {\bibfnamefont {L.}~\bibnamefont
  {Glazman}},\ }\bibfield  {title} {\bibinfo {title} {Kondo effect in quantum
  dots},\ }\href {https://doi.org/10.1088/0953-8984/16/16/R01} {\bibfield
  {journal} {\bibinfo  {journal} {Journal of Physics: Condensed Matter}\
  }\textbf {\bibinfo {volume} {16}},\ \bibinfo {pages} {R513} (\bibinfo {year}
  {2004})}\BibitemShut {NoStop}%
\bibitem [{\citenamefont {Georges}\ \emph {et~al.}(1996)\citenamefont
  {Georges}, \citenamefont {Kotliar}, \citenamefont {Krauth},\ and\
  \citenamefont {Rozenberg}}]{GeorgesRMP}%
  \BibitemOpen
  \bibfield  {author} {\bibinfo {author} {\bibfnamefont {A.}~\bibnamefont
  {Georges}}, \bibinfo {author} {\bibfnamefont {G.}~\bibnamefont {Kotliar}},
  \bibinfo {author} {\bibfnamefont {W.}~\bibnamefont {Krauth}},\ and\ \bibinfo
  {author} {\bibfnamefont {M.~J.}\ \bibnamefont {Rozenberg}},\ }\bibfield
  {title} {\bibinfo {title} {Dynamical mean-field theory of strongly correlated
  fermion systems and the limit of infinite dimensions},\ }\href
  {https://doi.org/10.1103/RevModPhys.68.13} {\bibfield  {journal} {\bibinfo
  {journal} {Rev. Mod. Phys.}\ }\textbf {\bibinfo {volume} {68}},\ \bibinfo
  {pages} {13} (\bibinfo {year} {1996})}\BibitemShut {NoStop}%
\bibitem [{\citenamefont {Aoki}\ \emph {et~al.}(2014)\citenamefont {Aoki},
  \citenamefont {Tsuji}, \citenamefont {Eckstein}, \citenamefont {Kollar},
  \citenamefont {Oka},\ and\ \citenamefont {Werner}}]{aoki14neqdmftreview}%
  \BibitemOpen
  \bibfield  {author} {\bibinfo {author} {\bibfnamefont {H.}~\bibnamefont
  {Aoki}}, \bibinfo {author} {\bibfnamefont {N.}~\bibnamefont {Tsuji}},
  \bibinfo {author} {\bibfnamefont {M.}~\bibnamefont {Eckstein}}, \bibinfo
  {author} {\bibfnamefont {M.}~\bibnamefont {Kollar}}, \bibinfo {author}
  {\bibfnamefont {T.}~\bibnamefont {Oka}},\ and\ \bibinfo {author}
  {\bibfnamefont {P.}~\bibnamefont {Werner}},\ }\bibfield  {title} {\bibinfo
  {title} {Nonequilibrium dynamical mean-field theory and its applications},\
  }\href {https://doi.org/10.1103/RevModPhys.86.779} {\bibfield  {journal}
  {\bibinfo  {journal} {Rev. Mod. Phys.}\ }\textbf {\bibinfo {volume} {86}},\
  \bibinfo {pages} {779} (\bibinfo {year} {2014})}\BibitemShut {NoStop}%
\bibitem [{\citenamefont {Kretchmer}\ and\ \citenamefont
  {Chan}(2018)}]{DMEmbeddingChan2018}%
  \BibitemOpen
  \bibfield  {author} {\bibinfo {author} {\bibfnamefont {J.~S.}\ \bibnamefont
  {Kretchmer}}\ and\ \bibinfo {author} {\bibfnamefont {G.~K.-L.}\ \bibnamefont
  {Chan}},\ }\bibfield  {title} {\bibinfo {title} {{A real-time extension of
  density matrix embedding theory for non-equilibrium electron dynamics}},\
  }\href {https://doi.org/10.1063/1.5012766} {\bibfield  {journal} {\bibinfo
  {journal} {The Journal of Chemical Physics}\ }\textbf {\bibinfo {volume}
  {148}},\ \bibinfo {pages} {054108} (\bibinfo {year} {2018})},\ \Eprint
  {https://arxiv.org/abs/https://pubs.aip.org/aip/jcp/article-pdf/doi/10.1063/1.5012766/15538332/054108\_1\_online.pdf}
  {https://pubs.aip.org/aip/jcp/article-pdf/doi/10.1063/1.5012766/15538332/054108\_1\_online.pdf}
  \BibitemShut {NoStop}%
\bibitem [{\citenamefont {M\"uhlbacher}\ and\ \citenamefont
  {Rabani}(2008)}]{Muhlbacher08Realtime}%
  \BibitemOpen
  \bibfield  {author} {\bibinfo {author} {\bibfnamefont {L.}~\bibnamefont
  {M\"uhlbacher}}\ and\ \bibinfo {author} {\bibfnamefont {E.}~\bibnamefont
  {Rabani}},\ }\bibfield  {title} {\bibinfo {title} {Real-time path integral
  approach to nonequilibrium many-body quantum systems},\ }\href
  {https://doi.org/10.1103/PhysRevLett.100.176403} {\bibfield  {journal}
  {\bibinfo  {journal} {Phys. Rev. Lett.}\ }\textbf {\bibinfo {volume} {100}},\
  \bibinfo {pages} {176403} (\bibinfo {year} {2008})}\BibitemShut {NoStop}%
\bibitem [{\citenamefont {Schir\'o}\ and\ \citenamefont
  {Fabrizio}(2009)}]{Schiro09realtime}%
  \BibitemOpen
  \bibfield  {author} {\bibinfo {author} {\bibfnamefont {M.}~\bibnamefont
  {Schir\'o}}\ and\ \bibinfo {author} {\bibfnamefont {M.}~\bibnamefont
  {Fabrizio}},\ }\bibfield  {title} {\bibinfo {title} {Real-time diagrammatic
  monte carlo for nonequilibrium quantum transport},\ }\href
  {https://doi.org/10.1103/PhysRevB.79.153302} {\bibfield  {journal} {\bibinfo
  {journal} {Phys. Rev. B}\ }\textbf {\bibinfo {volume} {79}},\ \bibinfo
  {pages} {153302} (\bibinfo {year} {2009})}\BibitemShut {NoStop}%
\bibitem [{\citenamefont {Werner}\ \emph {et~al.}(2009)\citenamefont {Werner},
  \citenamefont {Oka},\ and\ \citenamefont {Millis}}]{Werner09Diagrammatic}%
  \BibitemOpen
  \bibfield  {author} {\bibinfo {author} {\bibfnamefont {P.}~\bibnamefont
  {Werner}}, \bibinfo {author} {\bibfnamefont {T.}~\bibnamefont {Oka}},\ and\
  \bibinfo {author} {\bibfnamefont {A.~J.}\ \bibnamefont {Millis}},\ }\bibfield
   {title} {\bibinfo {title} {Diagrammatic monte carlo simulation of
  nonequilibrium systems},\ }\href {https://doi.org/10.1103/PhysRevB.79.035320}
  {\bibfield  {journal} {\bibinfo  {journal} {Phys. Rev. B}\ }\textbf {\bibinfo
  {volume} {79}},\ \bibinfo {pages} {035320} (\bibinfo {year}
  {2009})}\BibitemShut {NoStop}%
\bibitem [{\citenamefont {Gull}\ \emph
  {et~al.}(2011{\natexlab{a}})\citenamefont {Gull}, \citenamefont {Reichman},\
  and\ \citenamefont {Millis}}]{gull11NumericallyExact}%
  \BibitemOpen
  \bibfield  {author} {\bibinfo {author} {\bibfnamefont {E.}~\bibnamefont
  {Gull}}, \bibinfo {author} {\bibfnamefont {D.~R.}\ \bibnamefont {Reichman}},\
  and\ \bibinfo {author} {\bibfnamefont {A.~J.}\ \bibnamefont {Millis}},\
  }\bibfield  {title} {\bibinfo {title} {Numerically exact long-time behavior
  of nonequilibrium quantum impurity models},\ }\href
  {https://doi.org/10.1103/PhysRevB.84.085134} {\bibfield  {journal} {\bibinfo
  {journal} {Phys. Rev. B}\ }\textbf {\bibinfo {volume} {84}},\ \bibinfo
  {pages} {085134} (\bibinfo {year} {2011}{\natexlab{a}})}\BibitemShut
  {NoStop}%
\bibitem [{\citenamefont {Gull}\ \emph
  {et~al.}(2011{\natexlab{b}})\citenamefont {Gull}, \citenamefont {Millis},
  \citenamefont {Lichtenstein}, \citenamefont {Rubtsov}, \citenamefont
  {Troyer},\ and\ \citenamefont {Werner}}]{MonteCarloReview}%
  \BibitemOpen
  \bibfield  {author} {\bibinfo {author} {\bibfnamefont {E.}~\bibnamefont
  {Gull}}, \bibinfo {author} {\bibfnamefont {A.~J.}\ \bibnamefont {Millis}},
  \bibinfo {author} {\bibfnamefont {A.~I.}\ \bibnamefont {Lichtenstein}},
  \bibinfo {author} {\bibfnamefont {A.~N.}\ \bibnamefont {Rubtsov}}, \bibinfo
  {author} {\bibfnamefont {M.}~\bibnamefont {Troyer}},\ and\ \bibinfo {author}
  {\bibfnamefont {P.}~\bibnamefont {Werner}},\ }\bibfield  {title} {\bibinfo
  {title} {Continuous-time monte carlo methods for quantum impurity models},\
  }\href {https://doi.org/10.1103/RevModPhys.83.349} {\bibfield  {journal}
  {\bibinfo  {journal} {Rev. Mod. Phys.}\ }\textbf {\bibinfo {volume} {83}},\
  \bibinfo {pages} {349} (\bibinfo {year} {2011}{\natexlab{b}})}\BibitemShut
  {NoStop}%
\bibitem [{\citenamefont {Cohen}\ and\ \citenamefont
  {Rabani}(2011)}]{cohen11memory}%
  \BibitemOpen
  \bibfield  {author} {\bibinfo {author} {\bibfnamefont {G.}~\bibnamefont
  {Cohen}}\ and\ \bibinfo {author} {\bibfnamefont {E.}~\bibnamefont {Rabani}},\
  }\bibfield  {title} {\bibinfo {title} {Memory effects in nonequilibrium
  quantum impurity models},\ }\href
  {https://doi.org/10.1103/PhysRevB.84.075150} {\bibfield  {journal} {\bibinfo
  {journal} {Phys. Rev. B}\ }\textbf {\bibinfo {volume} {84}},\ \bibinfo
  {pages} {075150} (\bibinfo {year} {2011})}\BibitemShut {NoStop}%
\bibitem [{\citenamefont {Cohen}\ \emph {et~al.}(2013)\citenamefont {Cohen},
  \citenamefont {Gull}, \citenamefont {Reichman}, \citenamefont {Millis},\ and\
  \citenamefont {Rabani}}]{cohen13neqkondo}%
  \BibitemOpen
  \bibfield  {author} {\bibinfo {author} {\bibfnamefont {G.}~\bibnamefont
  {Cohen}}, \bibinfo {author} {\bibfnamefont {E.}~\bibnamefont {Gull}},
  \bibinfo {author} {\bibfnamefont {D.~R.}\ \bibnamefont {Reichman}}, \bibinfo
  {author} {\bibfnamefont {A.~J.}\ \bibnamefont {Millis}},\ and\ \bibinfo
  {author} {\bibfnamefont {E.}~\bibnamefont {Rabani}},\ }\bibfield  {title}
  {\bibinfo {title} {Numerically exact long-time magnetization dynamics at the
  nonequilibrium kondo crossover of the anderson impurity model},\ }\href
  {https://doi.org/10.1103/PhysRevB.87.195108} {\bibfield  {journal} {\bibinfo
  {journal} {Phys. Rev. B}\ }\textbf {\bibinfo {volume} {87}},\ \bibinfo
  {pages} {195108} (\bibinfo {year} {2013})}\BibitemShut {NoStop}%
\bibitem [{\citenamefont {Cohen}\ \emph {et~al.}(2015)\citenamefont {Cohen},
  \citenamefont {Gull}, \citenamefont {Reichman},\ and\ \citenamefont
  {Millis}}]{cohen15taming}%
  \BibitemOpen
  \bibfield  {author} {\bibinfo {author} {\bibfnamefont {G.}~\bibnamefont
  {Cohen}}, \bibinfo {author} {\bibfnamefont {E.}~\bibnamefont {Gull}},
  \bibinfo {author} {\bibfnamefont {D.~R.}\ \bibnamefont {Reichman}},\ and\
  \bibinfo {author} {\bibfnamefont {A.~J.}\ \bibnamefont {Millis}},\ }\bibfield
   {title} {\bibinfo {title} {Taming the dynamical sign problem in real-time
  evolution of quantum many-body problems},\ }\href
  {https://doi.org/10.1103/PhysRevLett.115.266802} {\bibfield  {journal}
  {\bibinfo  {journal} {Phys. Rev. Lett.}\ }\textbf {\bibinfo {volume} {115}},\
  \bibinfo {pages} {266802} (\bibinfo {year} {2015})}\BibitemShut {NoStop}%
\bibitem [{\citenamefont {Prior}\ \emph {et~al.}(2010)\citenamefont {Prior},
  \citenamefont {Chin}, \citenamefont {Huelga},\ and\ \citenamefont
  {Plenio}}]{prior10efficient}%
  \BibitemOpen
  \bibfield  {author} {\bibinfo {author} {\bibfnamefont {J.}~\bibnamefont
  {Prior}}, \bibinfo {author} {\bibfnamefont {A.~W.}\ \bibnamefont {Chin}},
  \bibinfo {author} {\bibfnamefont {S.~F.}\ \bibnamefont {Huelga}},\ and\
  \bibinfo {author} {\bibfnamefont {M.~B.}\ \bibnamefont {Plenio}},\ }\bibfield
   {title} {\bibinfo {title} {Efficient simulation of strong system-environment
  interactions},\ }\href {https://doi.org/10.1103/PhysRevLett.105.050404}
  {\bibfield  {journal} {\bibinfo  {journal} {Phys. Rev. Lett.}\ }\textbf
  {\bibinfo {volume} {105}},\ \bibinfo {pages} {050404} (\bibinfo {year}
  {2010})}\BibitemShut {NoStop}%
\bibitem [{\citenamefont {Wolf}\ \emph {et~al.}(2014)\citenamefont {Wolf},
  \citenamefont {McCulloch},\ and\ \citenamefont {Schollw\"ock}}]{WolfPRB14}%
  \BibitemOpen
  \bibfield  {author} {\bibinfo {author} {\bibfnamefont {F.~A.}\ \bibnamefont
  {Wolf}}, \bibinfo {author} {\bibfnamefont {I.~P.}\ \bibnamefont
  {McCulloch}},\ and\ \bibinfo {author} {\bibfnamefont {U.}~\bibnamefont
  {Schollw\"ock}},\ }\bibfield  {title} {\bibinfo {title} {Solving
  nonequilibrium dynamical mean-field theory using matrix product states},\
  }\href {https://doi.org/10.1103/PhysRevB.90.235131} {\bibfield  {journal}
  {\bibinfo  {journal} {Phys. Rev. B}\ }\textbf {\bibinfo {volume} {90}},\
  \bibinfo {pages} {235131} (\bibinfo {year} {2014})}\BibitemShut {NoStop}%
\bibitem [{\citenamefont {N\"u\ss{}eler}\ \emph {et~al.}(2020)\citenamefont
  {N\"u\ss{}eler}, \citenamefont {Dhand}, \citenamefont {Huelga},\ and\
  \citenamefont {Plenio}}]{Nusseler20Efficient}%
  \BibitemOpen
  \bibfield  {author} {\bibinfo {author} {\bibfnamefont {A.}~\bibnamefont
  {N\"u\ss{}eler}}, \bibinfo {author} {\bibfnamefont {I.}~\bibnamefont
  {Dhand}}, \bibinfo {author} {\bibfnamefont {S.~F.}\ \bibnamefont {Huelga}},\
  and\ \bibinfo {author} {\bibfnamefont {M.~B.}\ \bibnamefont {Plenio}},\
  }\bibfield  {title} {\bibinfo {title} {Efficient simulation of open quantum
  systems coupled to a fermionic bath},\ }\href
  {https://doi.org/10.1103/PhysRevB.101.155134} {\bibfield  {journal} {\bibinfo
   {journal} {Phys. Rev. B}\ }\textbf {\bibinfo {volume} {101}},\ \bibinfo
  {pages} {155134} (\bibinfo {year} {2020})}\BibitemShut {NoStop}%
\bibitem [{\citenamefont {Dorda}\ \emph {et~al.}(2014)\citenamefont {Dorda},
  \citenamefont {Nuss}, \citenamefont {von~der Linden},\ and\ \citenamefont
  {Arrigoni}}]{dorda14auxiliary}%
  \BibitemOpen
  \bibfield  {author} {\bibinfo {author} {\bibfnamefont {A.}~\bibnamefont
  {Dorda}}, \bibinfo {author} {\bibfnamefont {M.}~\bibnamefont {Nuss}},
  \bibinfo {author} {\bibfnamefont {W.}~\bibnamefont {von~der Linden}},\ and\
  \bibinfo {author} {\bibfnamefont {E.}~\bibnamefont {Arrigoni}},\ }\bibfield
  {title} {\bibinfo {title} {Auxiliary master equation approach to
  nonequilibrium correlated impurities},\ }\href
  {https://doi.org/10.1103/PhysRevB.89.165105} {\bibfield  {journal} {\bibinfo
  {journal} {Phys. Rev. B}\ }\textbf {\bibinfo {volume} {89}},\ \bibinfo
  {pages} {165105} (\bibinfo {year} {2014})}\BibitemShut {NoStop}%
\bibitem [{\citenamefont {Dorda}\ \emph {et~al.}(2017)\citenamefont {Dorda},
  \citenamefont {Sorantin}, \citenamefont {von~der Linden},\ and\ \citenamefont
  {Arrigoni}}]{Dorda2017}%
  \BibitemOpen
  \bibfield  {author} {\bibinfo {author} {\bibfnamefont {A.}~\bibnamefont
  {Dorda}}, \bibinfo {author} {\bibfnamefont {M.}~\bibnamefont {Sorantin}},
  \bibinfo {author} {\bibfnamefont {W.}~\bibnamefont {von~der Linden}},\ and\
  \bibinfo {author} {\bibfnamefont {E.}~\bibnamefont {Arrigoni}},\ }\bibfield
  {title} {\bibinfo {title} {Optimized auxiliary representation of
  non-markovian impurity problems by a lindblad equation},\ }\href
  {https://doi.org/10.1088/1367-2630/aa6ccc} {\bibfield  {journal} {\bibinfo
  {journal} {New Journal of Physics}\ }\textbf {\bibinfo {volume} {19}},\
  \bibinfo {pages} {063005} (\bibinfo {year} {2017})}\BibitemShut {NoStop}%
\bibitem [{\citenamefont {Lotem}\ \emph {et~al.}(2020)\citenamefont {Lotem},
  \citenamefont {Weichselbaum}, \citenamefont {von Delft},\ and\ \citenamefont
  {Goldstein}}]{Lotem20renormalized}%
  \BibitemOpen
  \bibfield  {author} {\bibinfo {author} {\bibfnamefont {M.}~\bibnamefont
  {Lotem}}, \bibinfo {author} {\bibfnamefont {A.}~\bibnamefont {Weichselbaum}},
  \bibinfo {author} {\bibfnamefont {J.}~\bibnamefont {von Delft}},\ and\
  \bibinfo {author} {\bibfnamefont {M.}~\bibnamefont {Goldstein}},\ }\bibfield
  {title} {\bibinfo {title} {Renormalized lindblad driving: A numerically exact
  nonequilibrium quantum impurity solver},\ }\href
  {https://doi.org/10.1103/PhysRevResearch.2.043052} {\bibfield  {journal}
  {\bibinfo  {journal} {Phys. Rev. Research}\ }\textbf {\bibinfo {volume}
  {2}},\ \bibinfo {pages} {043052} (\bibinfo {year} {2020})}\BibitemShut
  {NoStop}%
\bibitem [{\citenamefont {Schwarz}\ \emph {et~al.}(2018)\citenamefont
  {Schwarz}, \citenamefont {Weymann}, \citenamefont {von Delft},\ and\
  \citenamefont {Weichselbaum}}]{Schwarz18Nonequilibrium}%
  \BibitemOpen
  \bibfield  {author} {\bibinfo {author} {\bibfnamefont {F.}~\bibnamefont
  {Schwarz}}, \bibinfo {author} {\bibfnamefont {I.}~\bibnamefont {Weymann}},
  \bibinfo {author} {\bibfnamefont {J.}~\bibnamefont {von Delft}},\ and\
  \bibinfo {author} {\bibfnamefont {A.}~\bibnamefont {Weichselbaum}},\
  }\bibfield  {title} {\bibinfo {title} {Nonequilibrium steady-state transport
  in quantum impurity models: A thermofield and quantum quench approach using
  matrix product states},\ }\href
  {https://doi.org/10.1103/PhysRevLett.121.137702} {\bibfield  {journal}
  {\bibinfo  {journal} {Phys. Rev. Lett.}\ }\textbf {\bibinfo {volume} {121}},\
  \bibinfo {pages} {137702} (\bibinfo {year} {2018})}\BibitemShut {NoStop}%
\bibitem [{\citenamefont {Dan}\ \emph {et~al.}(2023)\citenamefont {Dan},
  \citenamefont {Xu}, \citenamefont {Stockburger}, \citenamefont {Ankerhold},\
  and\ \citenamefont {Shi}}]{DanPRB2023HEOM}%
  \BibitemOpen
  \bibfield  {author} {\bibinfo {author} {\bibfnamefont {X.}~\bibnamefont
  {Dan}}, \bibinfo {author} {\bibfnamefont {M.}~\bibnamefont {Xu}}, \bibinfo
  {author} {\bibfnamefont {J.~T.}\ \bibnamefont {Stockburger}}, \bibinfo
  {author} {\bibfnamefont {J.}~\bibnamefont {Ankerhold}},\ and\ \bibinfo
  {author} {\bibfnamefont {Q.}~\bibnamefont {Shi}},\ }\bibfield  {title}
  {\bibinfo {title} {Efficient low-temperature simulations for fermionic
  reservoirs with the hierarchical equations of motion method: Application to
  the anderson impurity model},\ }\href
  {https://doi.org/10.1103/PhysRevB.107.195429} {\bibfield  {journal} {\bibinfo
   {journal} {Phys. Rev. B}\ }\textbf {\bibinfo {volume} {107}},\ \bibinfo
  {pages} {195429} (\bibinfo {year} {2023})}\BibitemShut {NoStop}%
\bibitem [{\citenamefont {Thoenniss}\ \emph
  {et~al.}(2023{\natexlab{a}})\citenamefont {Thoenniss}, \citenamefont
  {Lerose},\ and\ \citenamefont {Abanin}}]{ThoennissPRB2022}%
  \BibitemOpen
  \bibfield  {author} {\bibinfo {author} {\bibfnamefont {J.}~\bibnamefont
  {Thoenniss}}, \bibinfo {author} {\bibfnamefont {A.}~\bibnamefont {Lerose}},\
  and\ \bibinfo {author} {\bibfnamefont {D.~A.}\ \bibnamefont {Abanin}},\
  }\bibfield  {title} {\bibinfo {title} {Nonequilibrium quantum impurity
  problems via matrix-product states in the temporal domain},\ }\href
  {https://doi.org/10.1103/PhysRevB.107.195101} {\bibfield  {journal} {\bibinfo
   {journal} {Phys. Rev. B}\ }\textbf {\bibinfo {volume} {107}},\ \bibinfo
  {pages} {195101} (\bibinfo {year} {2023}{\natexlab{a}})}\BibitemShut
  {NoStop}%
\bibitem [{\citenamefont {Thoenniss}\ \emph
  {et~al.}(2023{\natexlab{b}})\citenamefont {Thoenniss}, \citenamefont
  {Sonner}, \citenamefont {Lerose},\ and\ \citenamefont
  {Abanin}}]{ThoennissPRB23}%
  \BibitemOpen
  \bibfield  {author} {\bibinfo {author} {\bibfnamefont {J.}~\bibnamefont
  {Thoenniss}}, \bibinfo {author} {\bibfnamefont {M.}~\bibnamefont {Sonner}},
  \bibinfo {author} {\bibfnamefont {A.}~\bibnamefont {Lerose}},\ and\ \bibinfo
  {author} {\bibfnamefont {D.~A.}\ \bibnamefont {Abanin}},\ }\bibfield  {title}
  {\bibinfo {title} {Efficient method for quantum impurity problems out of
  equilibrium},\ }\href {https://doi.org/10.1103/PhysRevB.107.L201115}
  {\bibfield  {journal} {\bibinfo  {journal} {Phys. Rev. B}\ }\textbf {\bibinfo
  {volume} {107}},\ \bibinfo {pages} {L201115} (\bibinfo {year}
  {2023}{\natexlab{b}})}\BibitemShut {NoStop}%
\bibitem [{\citenamefont {Ng}\ \emph {et~al.}(2023)\citenamefont {Ng},
  \citenamefont {Park}, \citenamefont {Millis}, \citenamefont {Chan},\ and\
  \citenamefont {Reichman}}]{NgPRB23}%
  \BibitemOpen
  \bibfield  {author} {\bibinfo {author} {\bibfnamefont {N.}~\bibnamefont
  {Ng}}, \bibinfo {author} {\bibfnamefont {G.}~\bibnamefont {Park}}, \bibinfo
  {author} {\bibfnamefont {A.~J.}\ \bibnamefont {Millis}}, \bibinfo {author}
  {\bibfnamefont {G.~K.-L.}\ \bibnamefont {Chan}},\ and\ \bibinfo {author}
  {\bibfnamefont {D.~R.}\ \bibnamefont {Reichman}},\ }\bibfield  {title}
  {\bibinfo {title} {Real-time evolution of anderson impurity models via tensor
  network influence functionals},\ }\href
  {https://doi.org/10.1103/PhysRevB.107.125103} {\bibfield  {journal} {\bibinfo
   {journal} {Phys. Rev. B}\ }\textbf {\bibinfo {volume} {107}},\ \bibinfo
  {pages} {125103} (\bibinfo {year} {2023})}\BibitemShut {NoStop}%
\bibitem [{\citenamefont {Park}\ \emph
  {et~al.}(2024{\natexlab{a}})\citenamefont {Park}, \citenamefont {Ng},
  \citenamefont {Reichman},\ and\ \citenamefont {Chan}}]{ParkPRB24}%
  \BibitemOpen
  \bibfield  {author} {\bibinfo {author} {\bibfnamefont {G.}~\bibnamefont
  {Park}}, \bibinfo {author} {\bibfnamefont {N.}~\bibnamefont {Ng}}, \bibinfo
  {author} {\bibfnamefont {D.~R.}\ \bibnamefont {Reichman}},\ and\ \bibinfo
  {author} {\bibfnamefont {G.~K.-L.}\ \bibnamefont {Chan}},\ }\bibfield
  {title} {\bibinfo {title} {Tensor network influence functionals in the
  continuous-time limit: Connections to quantum embedding, bath discretization,
  and higher-order time propagation},\ }\href
  {https://doi.org/10.1103/PhysRevB.110.045104} {\bibfield  {journal} {\bibinfo
   {journal} {Phys. Rev. B}\ }\textbf {\bibinfo {volume} {110}},\ \bibinfo
  {pages} {045104} (\bibinfo {year} {2024}{\natexlab{a}})}\BibitemShut
  {NoStop}%
\bibitem [{\citenamefont {Feynman}\ and\ \citenamefont
  {Vernon}(1963)}]{FeynmanVernon}%
  \BibitemOpen
  \bibfield  {author} {\bibinfo {author} {\bibfnamefont {R.}~\bibnamefont
  {Feynman}}\ and\ \bibinfo {author} {\bibfnamefont {F.}~\bibnamefont
  {Vernon}},\ }\bibfield  {title} {\bibinfo {title} {The theory of a general
  quantum system interacting with a linear dissipative system},\ }\href
  {https://doi.org/https://doi.org/10.1016/0003-4916(63)90068-X} {\bibfield
  {journal} {\bibinfo  {journal} {Annals of Physics}\ }\textbf {\bibinfo
  {volume} {24}},\ \bibinfo {pages} {118 } (\bibinfo {year}
  {1963})}\BibitemShut {NoStop}%
\bibitem [{\citenamefont {Lerose}\ \emph {et~al.}(2021)\citenamefont {Lerose},
  \citenamefont {Sonner},\ and\ \citenamefont {Abanin}}]{lerose2021}%
  \BibitemOpen
  \bibfield  {author} {\bibinfo {author} {\bibfnamefont {A.}~\bibnamefont
  {Lerose}}, \bibinfo {author} {\bibfnamefont {M.}~\bibnamefont {Sonner}},\
  and\ \bibinfo {author} {\bibfnamefont {D.~A.}\ \bibnamefont {Abanin}},\
  }\bibfield  {title} {\bibinfo {title} {Scaling of temporal entanglement in
  proximity to integrability},\ }\href
  {https://doi.org/10.1103/PhysRevB.104.035137} {\bibfield  {journal} {\bibinfo
   {journal} {Phys. Rev. B}\ }\textbf {\bibinfo {volume} {104}},\ \bibinfo
  {pages} {035137} (\bibinfo {year} {2021})}\BibitemShut {NoStop}%
\bibitem [{Note1()}]{Note1}%
  \BibitemOpen
  \bibinfo {note} {We note that a quasipolynomial-time classical algorithm for
  finding the ground state energy of QIM is available~\cite
  {Bravyi2016}}\BibitemShut {NoStop}%
\bibitem [{\citenamefont {Vilkoviskiy}\ and\ \citenamefont
  {Abanin}(2024)}]{VilkoviskiyPRB2024}%
  \BibitemOpen
  \bibfield  {author} {\bibinfo {author} {\bibfnamefont {I.}~\bibnamefont
  {Vilkoviskiy}}\ and\ \bibinfo {author} {\bibfnamefont {D.~A.}\ \bibnamefont
  {Abanin}},\ }\bibfield  {title} {\bibinfo {title} {Bound on approximating
  non-markovian dynamics by tensor networks in the time domain},\ }\href
  {https://doi.org/10.1103/PhysRevB.109.205126} {\bibfield  {journal} {\bibinfo
   {journal} {Phys. Rev. B}\ }\textbf {\bibinfo {volume} {109}},\ \bibinfo
  {pages} {205126} (\bibinfo {year} {2024})}\BibitemShut {NoStop}%
\bibitem [{\citenamefont {Tamascelli}\ \emph {et~al.}(2018)\citenamefont
  {Tamascelli}, \citenamefont {Smirne}, \citenamefont {Huelga},\ and\
  \citenamefont {Plenio}}]{TamascelliPRL2018}%
  \BibitemOpen
  \bibfield  {author} {\bibinfo {author} {\bibfnamefont {D.}~\bibnamefont
  {Tamascelli}}, \bibinfo {author} {\bibfnamefont {A.}~\bibnamefont {Smirne}},
  \bibinfo {author} {\bibfnamefont {S.~F.}\ \bibnamefont {Huelga}},\ and\
  \bibinfo {author} {\bibfnamefont {M.~B.}\ \bibnamefont {Plenio}},\ }\bibfield
   {title} {\bibinfo {title} {Nonperturbative treatment of non-markovian
  dynamics of open quantum systems},\ }\href
  {https://doi.org/10.1103/PhysRevLett.120.030402} {\bibfield  {journal}
  {\bibinfo  {journal} {Phys. Rev. Lett.}\ }\textbf {\bibinfo {volume} {120}},\
  \bibinfo {pages} {030402} (\bibinfo {year} {2018})}\BibitemShut {NoStop}%
\bibitem [{\citenamefont {Somoza}\ \emph {et~al.}(2019)\citenamefont {Somoza},
  \citenamefont {Marty}, \citenamefont {Lim}, \citenamefont {Huelga},\ and\
  \citenamefont {Plenio}}]{Somoza2019}%
  \BibitemOpen
  \bibfield  {author} {\bibinfo {author} {\bibfnamefont {A.~D.}\ \bibnamefont
  {Somoza}}, \bibinfo {author} {\bibfnamefont {O.}~\bibnamefont {Marty}},
  \bibinfo {author} {\bibfnamefont {J.}~\bibnamefont {Lim}}, \bibinfo {author}
  {\bibfnamefont {S.~F.}\ \bibnamefont {Huelga}},\ and\ \bibinfo {author}
  {\bibfnamefont {M.~B.}\ \bibnamefont {Plenio}},\ }\bibfield  {title}
  {\bibinfo {title} {Dissipation-assisted matrix product factorization},\
  }\href {https://doi.org/10.1103/PhysRevLett.123.100502} {\bibfield  {journal}
  {\bibinfo  {journal} {Phys. Rev. Lett.}\ }\textbf {\bibinfo {volume} {123}},\
  \bibinfo {pages} {100502} (\bibinfo {year} {2019})}\BibitemShut {NoStop}%
\bibitem [{\citenamefont {Xu}\ \emph {et~al.}(2022)\citenamefont {Xu},
  \citenamefont {Yan}, \citenamefont {Shi}, \citenamefont {Ankerhold},\ and\
  \citenamefont {Stockburger}}]{Xu2022}%
  \BibitemOpen
  \bibfield  {author} {\bibinfo {author} {\bibfnamefont {M.}~\bibnamefont
  {Xu}}, \bibinfo {author} {\bibfnamefont {Y.}~\bibnamefont {Yan}}, \bibinfo
  {author} {\bibfnamefont {Q.}~\bibnamefont {Shi}}, \bibinfo {author}
  {\bibfnamefont {J.}~\bibnamefont {Ankerhold}},\ and\ \bibinfo {author}
  {\bibfnamefont {J.~T.}\ \bibnamefont {Stockburger}},\ }\bibfield  {title}
  {\bibinfo {title} {Taming quantum noise for efficient low temperature
  simulations of open quantum systems},\ }\href
  {https://doi.org/10.1103/PhysRevLett.129.230601} {\bibfield  {journal}
  {\bibinfo  {journal} {Phys. Rev. Lett.}\ }\textbf {\bibinfo {volume} {129}},\
  \bibinfo {pages} {230601} (\bibinfo {year} {2022})}\BibitemShut {NoStop}%
\bibitem [{Note2()}]{Note2}%
  \BibitemOpen
  \bibinfo {note} {Nevertheless, the {\protect \it combination} of pseudomodes
  closely approximates the original, physical evolution, and, moreover, it can
  be complemented to fully physical evolution if a dissipative decay channel is
  added to the impurity evolution.}\BibitemShut {Stop}%
\bibitem [{\citenamefont {Kaye}\ \emph {et~al.}(2022)\citenamefont {Kaye},
  \citenamefont {Chen},\ and\ \citenamefont {Parcollet}}]{KayeDiscrete}%
  \BibitemOpen
  \bibfield  {author} {\bibinfo {author} {\bibfnamefont {J.}~\bibnamefont
  {Kaye}}, \bibinfo {author} {\bibfnamefont {K.}~\bibnamefont {Chen}},\ and\
  \bibinfo {author} {\bibfnamefont {O.}~\bibnamefont {Parcollet}},\ }\bibfield
  {title} {\bibinfo {title} {Discrete lehmann representation of imaginary time
  green's functions},\ }\href {https://doi.org/10.1103/PhysRevB.105.235115}
  {\bibfield  {journal} {\bibinfo  {journal} {Phys. Rev. B}\ }\textbf {\bibinfo
  {volume} {105}},\ \bibinfo {pages} {235115} (\bibinfo {year}
  {2022})}\BibitemShut {NoStop}%
\bibitem [{\citenamefont {Kiese}\ \emph {et~al.}(2024)\citenamefont {Kiese},
  \citenamefont {Strand}, \citenamefont {Chen}, \citenamefont {Wentzell},
  \citenamefont {Parcollet},\ and\ \citenamefont {Kaye}}]{kiese24Discrete}%
  \BibitemOpen
  \bibfield  {author} {\bibinfo {author} {\bibfnamefont {D.}~\bibnamefont
  {Kiese}}, \bibinfo {author} {\bibfnamefont {H.~U.~R.}\ \bibnamefont
  {Strand}}, \bibinfo {author} {\bibfnamefont {K.}~\bibnamefont {Chen}},
  \bibinfo {author} {\bibfnamefont {N.}~\bibnamefont {Wentzell}}, \bibinfo
  {author} {\bibfnamefont {O.}~\bibnamefont {Parcollet}},\ and\ \bibinfo
  {author} {\bibfnamefont {J.}~\bibnamefont {Kaye}},\ }\href
  {https://arxiv.org/abs/2405.06716} {\bibinfo {title} {Discrete lehmann
  representation of three-point functions}} (\bibinfo {year} {2024}),\ \Eprint
  {https://arxiv.org/abs/2405.06716} {arXiv:2405.06716 [physics.comp-ph]}
  \BibitemShut {NoStop}%
\bibitem [{\citenamefont {Nakatsukasa}\ \emph {et~al.}(2018)\citenamefont
  {Nakatsukasa}, \citenamefont {S\`{e}te},\ and\ \citenamefont
  {Trefethen}}]{NakatsukasaAAA}%
  \BibitemOpen
  \bibfield  {author} {\bibinfo {author} {\bibfnamefont {Y.}~\bibnamefont
  {Nakatsukasa}}, \bibinfo {author} {\bibfnamefont {O.}~\bibnamefont
  {S\`{e}te}},\ and\ \bibinfo {author} {\bibfnamefont {L.~N.}\ \bibnamefont
  {Trefethen}},\ }\bibfield  {title} {\bibinfo {title} {The aaa algorithm for
  rational approximation},\ }\href {https://doi.org/10.1137/16M1106122}
  {\bibfield  {journal} {\bibinfo  {journal} {SIAM Journal on Scientific
  Computing}\ }\textbf {\bibinfo {volume} {40}},\ \bibinfo {pages} {A1494}
  (\bibinfo {year} {2018})}\BibitemShut {NoStop}%
\bibitem [{\citenamefont {Beylkin}\ and\ \citenamefont
  {Monz{\'o}n}(2010)}]{Beylkin2010ApproximationBE}%
  \BibitemOpen
  \bibfield  {author} {\bibinfo {author} {\bibfnamefont {G.}~\bibnamefont
  {Beylkin}}\ and\ \bibinfo {author} {\bibfnamefont {L.}~\bibnamefont
  {Monz{\'o}n}},\ }\bibfield  {title} {\bibinfo {title} {Approximation by
  exponential sums revisited},\ }\href
  {https://api.semanticscholar.org/CorpusID:18481955} {\bibfield  {journal}
  {\bibinfo  {journal} {Applied and Computational Harmonic Analysis}\ }\textbf
  {\bibinfo {volume} {28}},\ \bibinfo {pages} {131} (\bibinfo {year}
  {2010})}\BibitemShut {NoStop}%
\bibitem [{\citenamefont {McNamee}\ \emph {et~al.}(1971)\citenamefont
  {McNamee}, \citenamefont {Stenger},\ and\ \citenamefont {Whitney}}]{McNamee}%
  \BibitemOpen
  \bibfield  {author} {\bibinfo {author} {\bibfnamefont {J.}~\bibnamefont
  {McNamee}}, \bibinfo {author} {\bibfnamefont {F.}~\bibnamefont {Stenger}},\
  and\ \bibinfo {author} {\bibfnamefont {E.~L.}\ \bibnamefont {Whitney}},\
  }\bibfield  {title} {\bibinfo {title} {Whittaker's cardinal function in
  retrospect},\ }\href {http://www.jstor.org/stable/2005140} {\bibfield
  {journal} {\bibinfo  {journal} {Mathematics of Computation}\ }\textbf
  {\bibinfo {volume} {25}},\ \bibinfo {pages} {141} (\bibinfo {year}
  {1971})}\BibitemShut {NoStop}%
\bibitem [{\citenamefont {Gu}\ and\ \citenamefont
  {Eisenstat}(1996)}]{GuEfficient}%
  \BibitemOpen
  \bibfield  {author} {\bibinfo {author} {\bibfnamefont {M.}~\bibnamefont
  {Gu}}\ and\ \bibinfo {author} {\bibfnamefont {S.~C.}\ \bibnamefont
  {Eisenstat}},\ }\bibfield  {title} {\bibinfo {title} {Efficient algorithms
  for computing a strong rank-revealing qr factorization},\ }\href@noop {}
  {\bibfield  {journal} {\bibinfo  {journal} {SIAM J. Sci. Comput.}\ }\textbf
  {\bibinfo {volume} {17}},\ \bibinfo {pages} {848–869} (\bibinfo {year}
  {1996})}\BibitemShut {NoStop}%
\bibitem [{\citenamefont {Cheng}\ \emph {et~al.}(2005)\citenamefont {Cheng},
  \citenamefont {Gimbutas}, \citenamefont {Martinsson},\ and\ \citenamefont
  {Rokhlin}}]{ChengCompression2005}%
  \BibitemOpen
  \bibfield  {author} {\bibinfo {author} {\bibfnamefont {H.}~\bibnamefont
  {Cheng}}, \bibinfo {author} {\bibfnamefont {Z.}~\bibnamefont {Gimbutas}},
  \bibinfo {author} {\bibfnamefont {P.~G.}\ \bibnamefont {Martinsson}},\ and\
  \bibinfo {author} {\bibfnamefont {V.}~\bibnamefont {Rokhlin}},\ }\bibfield
  {title} {\bibinfo {title} {On the compression of low rank matrices},\ }\href
  {https://doi.org/10.1137/030602678} {\bibfield  {journal} {\bibinfo
  {journal} {SIAM J. Sci. Comput.}\ }\textbf {\bibinfo {volume} {26}},\
  \bibinfo {pages} {1389–1404} (\bibinfo {year} {2005})}\BibitemShut
  {NoStop}%
\bibitem [{\citenamefont {Liberty}\ \emph {et~al.}(2007)\citenamefont
  {Liberty}, \citenamefont {Woolfe}, \citenamefont {Martinsson}, \citenamefont
  {Rokhlin},\ and\ \citenamefont {Tygert}}]{LibertyRandomized}%
  \BibitemOpen
  \bibfield  {author} {\bibinfo {author} {\bibfnamefont {E.}~\bibnamefont
  {Liberty}}, \bibinfo {author} {\bibfnamefont {F.}~\bibnamefont {Woolfe}},
  \bibinfo {author} {\bibfnamefont {P.-G.}\ \bibnamefont {Martinsson}},
  \bibinfo {author} {\bibfnamefont {V.}~\bibnamefont {Rokhlin}},\ and\ \bibinfo
  {author} {\bibfnamefont {M.}~\bibnamefont {Tygert}},\ }\bibfield  {title}
  {\bibinfo {title} {Randomized algorithms for the low-rank approximation of
  matrices},\ }\href@noop {} {\bibfield  {journal} {\bibinfo  {journal} {Proc.
  Natl. Acad. Sci. U.S.A.}\ }\textbf {\bibinfo {volume} {104}},\ \bibinfo
  {pages} {20167–20172} (\bibinfo {year} {2007})}\BibitemShut {NoStop}%
\bibitem [{\citenamefont {Woolfe}\ \emph {et~al.}(2008)\citenamefont {Woolfe},
  \citenamefont {Liberty}, \citenamefont {Rokhlin},\ and\ \citenamefont
  {Tygert}}]{Woolfe08Fast}%
  \BibitemOpen
  \bibfield  {author} {\bibinfo {author} {\bibfnamefont {F.}~\bibnamefont
  {Woolfe}}, \bibinfo {author} {\bibfnamefont {E.}~\bibnamefont {Liberty}},
  \bibinfo {author} {\bibfnamefont {V.}~\bibnamefont {Rokhlin}},\ and\ \bibinfo
  {author} {\bibfnamefont {M.}~\bibnamefont {Tygert}},\ }\bibfield  {title}
  {\bibinfo {title} {A fast randomized algorithm for the approximation of
  matrices},\ }\href
  {https://doi.org/https://doi.org/10.1016/j.acha.2007.12.002} {\bibfield
  {journal} {\bibinfo  {journal} {Applied and Computational Harmonic Analysis}\
  }\textbf {\bibinfo {volume} {25}},\ \bibinfo {pages} {335} (\bibinfo {year}
  {2008})}\BibitemShut {NoStop}%
\bibitem [{\citenamefont {Virtanen}\ \emph {et~al.}(2020)\citenamefont
  {Virtanen}, \citenamefont {Gommers}, \citenamefont {Oliphant}, \citenamefont
  {Haberland}, \citenamefont {Reddy}, \citenamefont {Cournapeau}, \citenamefont
  {Burovski}, \citenamefont {Peterson}, \citenamefont {Weckesser},
  \citenamefont {Bright}, \citenamefont {{van der Walt}}, \citenamefont
  {Brett}, \citenamefont {Wilson}, \citenamefont {Millman}, \citenamefont
  {Mayorov}, \citenamefont {Nelson}, \citenamefont {Jones}, \citenamefont
  {Kern}, \citenamefont {Larson}, \citenamefont {Carey}, \citenamefont {Polat},
  \citenamefont {Feng}, \citenamefont {Moore}, \citenamefont {{VanderPlas}},
  \citenamefont {Laxalde}, \citenamefont {Perktold}, \citenamefont {Cimrman},
  \citenamefont {Henriksen}, \citenamefont {Quintero}, \citenamefont {Harris},
  \citenamefont {Archibald}, \citenamefont {Ribeiro}, \citenamefont
  {Pedregosa}, \citenamefont {{van Mulbregt}},\ and\ \citenamefont {{SciPy 1.0
  Contributors}}}]{Virtanen20Scipy}%
  \BibitemOpen
  \bibfield  {author} {\bibinfo {author} {\bibfnamefont {P.}~\bibnamefont
  {Virtanen}}, \bibinfo {author} {\bibfnamefont {R.}~\bibnamefont {Gommers}},
  \bibinfo {author} {\bibfnamefont {T.~E.}\ \bibnamefont {Oliphant}}, \bibinfo
  {author} {\bibfnamefont {M.}~\bibnamefont {Haberland}}, \bibinfo {author}
  {\bibfnamefont {T.}~\bibnamefont {Reddy}}, \bibinfo {author} {\bibfnamefont
  {D.}~\bibnamefont {Cournapeau}}, \bibinfo {author} {\bibfnamefont
  {E.}~\bibnamefont {Burovski}}, \bibinfo {author} {\bibfnamefont
  {P.}~\bibnamefont {Peterson}}, \bibinfo {author} {\bibfnamefont
  {W.}~\bibnamefont {Weckesser}}, \bibinfo {author} {\bibfnamefont
  {J.}~\bibnamefont {Bright}}, \bibinfo {author} {\bibfnamefont {S.~J.}\
  \bibnamefont {{van der Walt}}}, \bibinfo {author} {\bibfnamefont
  {M.}~\bibnamefont {Brett}}, \bibinfo {author} {\bibfnamefont
  {J.}~\bibnamefont {Wilson}}, \bibinfo {author} {\bibfnamefont {K.~J.}\
  \bibnamefont {Millman}}, \bibinfo {author} {\bibfnamefont {N.}~\bibnamefont
  {Mayorov}}, \bibinfo {author} {\bibfnamefont {A.~R.~J.}\ \bibnamefont
  {Nelson}}, \bibinfo {author} {\bibfnamefont {E.}~\bibnamefont {Jones}},
  \bibinfo {author} {\bibfnamefont {R.}~\bibnamefont {Kern}}, \bibinfo {author}
  {\bibfnamefont {E.}~\bibnamefont {Larson}}, \bibinfo {author} {\bibfnamefont
  {C.~J.}\ \bibnamefont {Carey}}, \bibinfo {author} {\bibfnamefont
  {{\.I}.}~\bibnamefont {Polat}}, \bibinfo {author} {\bibfnamefont
  {Y.}~\bibnamefont {Feng}}, \bibinfo {author} {\bibfnamefont {E.~W.}\
  \bibnamefont {Moore}}, \bibinfo {author} {\bibfnamefont {J.}~\bibnamefont
  {{VanderPlas}}}, \bibinfo {author} {\bibfnamefont {D.}~\bibnamefont
  {Laxalde}}, \bibinfo {author} {\bibfnamefont {J.}~\bibnamefont {Perktold}},
  \bibinfo {author} {\bibfnamefont {R.}~\bibnamefont {Cimrman}}, \bibinfo
  {author} {\bibfnamefont {I.}~\bibnamefont {Henriksen}}, \bibinfo {author}
  {\bibfnamefont {E.~A.}\ \bibnamefont {Quintero}}, \bibinfo {author}
  {\bibfnamefont {C.~R.}\ \bibnamefont {Harris}}, \bibinfo {author}
  {\bibfnamefont {A.~M.}\ \bibnamefont {Archibald}}, \bibinfo {author}
  {\bibfnamefont {A.~H.}\ \bibnamefont {Ribeiro}}, \bibinfo {author}
  {\bibfnamefont {F.}~\bibnamefont {Pedregosa}}, \bibinfo {author}
  {\bibfnamefont {P.}~\bibnamefont {{van Mulbregt}}},\ and\ \bibinfo {author}
  {\bibnamefont {{SciPy 1.0 Contributors}}},\ }\bibfield  {title} {\bibinfo
  {title} {{{SciPy} 1.0: Fundamental Algorithms for Scientific Computing in
  Python}},\ }\href {https://doi.org/10.1038/s41592-019-0686-2} {\bibfield
  {journal} {\bibinfo  {journal} {Nature Methods}\ }\textbf {\bibinfo {volume}
  {17}},\ \bibinfo {pages} {261} (\bibinfo {year} {2020})}\BibitemShut
  {NoStop}%
\bibitem [{\citenamefont {Berrut}\ and\ \citenamefont
  {Trefethen}(2004)}]{BerrutBarycentric}%
  \BibitemOpen
  \bibfield  {author} {\bibinfo {author} {\bibfnamefont {J.-P.}\ \bibnamefont
  {Berrut}}\ and\ \bibinfo {author} {\bibfnamefont {L.~N.}\ \bibnamefont
  {Trefethen}},\ }\bibfield  {title} {\bibinfo {title} {Barycentric lagrange
  interpolation},\ }\href {https://doi.org/10.1137/S0036144502417715}
  {\bibfield  {journal} {\bibinfo  {journal} {SIAM Review}\ }\textbf {\bibinfo
  {volume} {46}},\ \bibinfo {pages} {501} (\bibinfo {year} {2004})},\ \Eprint
  {https://arxiv.org/abs/https://doi.org/10.1137/S0036144502417715}
  {https://doi.org/10.1137/S0036144502417715} \BibitemShut {NoStop}%
\bibitem [{\citenamefont {Hofreither}(2021)}]{Hofreither21Algorithm}%
  \BibitemOpen
  \bibfield  {author} {\bibinfo {author} {\bibfnamefont {C.}~\bibnamefont
  {Hofreither}},\ }\bibfield  {title} {\bibinfo {title} {An algorithm for best
  rational approximation based on barycentric rational interpolation},\
  }\href@noop {} {\bibfield  {journal} {\bibinfo  {journal} {Numerical
  Algorithms}\ ,\ \bibinfo {pages} {365}} (\bibinfo {year} {2021})}\BibitemShut
  {NoStop}%
\bibitem [{\citenamefont {Gimbutas}\ \emph {et~al.}(2020)\citenamefont
  {Gimbutas}, \citenamefont {Marshall},\ and\ \citenamefont
  {Rokhlin}}]{Gimbutas20Fast}%
  \BibitemOpen
  \bibfield  {author} {\bibinfo {author} {\bibfnamefont {Z.}~\bibnamefont
  {Gimbutas}}, \bibinfo {author} {\bibfnamefont {N.~F.}\ \bibnamefont
  {Marshall}},\ and\ \bibinfo {author} {\bibfnamefont {V.}~\bibnamefont
  {Rokhlin}},\ }\bibfield  {title} {\bibinfo {title} {A fast simple algorithm
  for computing the potential of charges on a line},\ }\href
  {https://doi.org/https://doi.org/10.1016/j.acha.2020.06.002} {\bibfield
  {journal} {\bibinfo  {journal} {Applied and Computational Harmonic Analysis}\
  }\textbf {\bibinfo {volume} {49}},\ \bibinfo {pages} {815} (\bibinfo {year}
  {2020})}\BibitemShut {NoStop}%
\bibitem [{Note3()}]{Note3}%
  \BibitemOpen
  \bibinfo {note} {The differences in the particle and hole component reflect
  algorithmic subtleties related to the order of sample points and do not
  contradict the equivalence of the particle and hole component upon
  substituting $\omega \to -\omega .$}\BibitemShut {NoStop}%
\bibitem [{\citenamefont {Park}\ \emph
  {et~al.}(2024{\natexlab{b}})\citenamefont {Park}, \citenamefont {Huang},
  \citenamefont {Zhu}, \citenamefont {Yang}, \citenamefont {Chan},\ and\
  \citenamefont {Lin}}]{ParkPseudomodesArxiv2024}%
  \BibitemOpen
  \bibfield  {author} {\bibinfo {author} {\bibfnamefont {G.}~\bibnamefont
  {Park}}, \bibinfo {author} {\bibfnamefont {Z.}~\bibnamefont {Huang}},
  \bibinfo {author} {\bibfnamefont {Y.}~\bibnamefont {Zhu}}, \bibinfo {author}
  {\bibfnamefont {C.}~\bibnamefont {Yang}}, \bibinfo {author} {\bibfnamefont
  {G.~K.-L.}\ \bibnamefont {Chan}},\ and\ \bibinfo {author} {\bibfnamefont
  {L.}~\bibnamefont {Lin}},\ }\href {https://arxiv.org/abs/2408.15529}
  {\bibinfo {title} {Quasi-lindblad pseudomode theory for open quantum
  systems}} (\bibinfo {year} {2024}{\natexlab{b}}),\ \Eprint
  {https://arxiv.org/abs/2408.15529} {arXiv:2408.15529 [quant-ph]} \BibitemShut
  {NoStop}%
\bibitem [{\citenamefont {Clavelli}\ and\ \citenamefont
  {Shapiro}(1973)}]{CLAVELLI1973490}%
  \BibitemOpen
  \bibfield  {author} {\bibinfo {author} {\bibfnamefont {L.}~\bibnamefont
  {Clavelli}}\ and\ \bibinfo {author} {\bibfnamefont {J.}~\bibnamefont
  {Shapiro}},\ }\bibfield  {title} {\bibinfo {title} {Pomeron factorization in
  general dual models},\ }\href
  {https://doi.org/https://doi.org/10.1016/0550-3213(73)90113-2} {\bibfield
  {journal} {\bibinfo  {journal} {Nuclear Physics B}\ }\textbf {\bibinfo
  {volume} {57}},\ \bibinfo {pages} {490} (\bibinfo {year} {1973})}\BibitemShut
  {NoStop}%
\end{thebibliography}%

\end{document}